\RequirePackage{fix-cm}
\documentclass[smallextended]{svjour3}       
\smartqed  
\usepackage{amssymb}
\usepackage{paralist} 
\usepackage{graphicx}
\usepackage{breqn}
\usepackage{tikz, xcolor}
\usepackage{array}
\usepackage{multicol}

\usepackage{thm-restate}

\definecolor{dnrbl}{rgb}{0,0,0.5}
\definecolor{dnrgr}{rgb}{0,0.5,0}
\definecolor{dnrre}{rgb}{0.5,0,0}
\usepackage[colorlinks=true, citecolor=dnrgr, linkcolor=dnrre, urlcolor=dnrbl]{hyperref}

\newcommand{\N}{\mathcal{N}}

\newcommand{\RR}{\mathbb{R}}

\newcommand{\NN}{\mathbb{N}}
\newcommand{\UG}{U\! G}
\newcommand{\UR}{U\! R}
\newcommand{\HG}{H\! G}
\newcommand{\HR}{H\! R}
\newcommand{\FG}{F \! G}
\newcommand{\FR}{F \! R}
\newcommand{\GD}{G\! D}

\newcommand{\Sm}{\mathtt{Smooth}_{k,\eps'}}

\newtheorem{thmm}[subsection]{Theorem}
\newtheorem{prop}[subsection]{Proposition}
\newtheorem{lem}[subsection]{Lemma}
\newtheorem{rem}[subsection]{Remark}
\newtheorem{coro}[subsection]{Corollary}

\newtheorem{conj}[subsection]{Conjecture}
\newtheorem{defin}[subsection]{Definition}
\newtheorem{questionn}[subsection]{Open Question}

\newcommand{\bpm}{\begin{pmatrix}}
\newcommand{\epm}{\end{pmatrix}}

\newcommand{\eps}{\varepsilon}

\begin{document}

\title{Tipping Points in Schelling Segregation}

\author{George Barmpalias         \and
        Richard Elwes  \and 
        Andy Lewis-Pye
}


\institute{{\bf George Barmpalias:}
(1) State Key Lab of Computer Science,
Institute of Software,
Chinese Academy of Sciences,
Beijing 100190,
China and 
(2) School of Mathematics, Statistics and Operations Research,
Victoria University, Wellington, New Zealand
\email{barmpalias@gmail.com,  www.barmpalias.net}            
           \and
{\bf Richard Elwes}, School of Mathematics,
University of Leeds, LS2 9JT Leeds, United Kingdom.
\email{r.elwes@gmail.com, www.richardelwes.co.uk} 
\and
{\bf Andy Lewis-Pye}, Department of Mathematics, London School of Economics, Houghton Street, London WC2A 2AE.
 \email{andy@aemlewis.com, www.aemlewis.co.uk} (Andy Lewis-Pye was previously called Andrew Lewis.) Authors are listed alphabetically.}

\date{}

\maketitle

\begin{abstract}
Thomas Schelling's spacial proximity model illustrated how racial segregation can emerge, unwanted, from the actions of citizens acting in accordance with their individual local preferences. One of the earliest agent-based models, it is closely related both to the spin-1 models of statistical physics, and to cascading phenomena on networks. Here a 1-dimensional unperturbed variant of the model is studied, which is \emph{open} in the sense that agents may enter and exit the model. Following the authors' previous work \cite{BEL1} and that of Brandt, Immorlica, Kamath, and Kleinberg in \cite{BK}, rigorous asymptotic results are established.

This model's openness allows either race to take over almost everywhere. Tipping points are identified between the regions of takeover and staticity. In a significant generalization of the models considered in \cite{BEL1} and \cite{BK}, the model's parameters comprise the initial proportions of the two races, along with independent values of the tolerance for each race.

\end{abstract}

\setcounter{tocdepth}{1}

\section{Introduction} \label{section:intro}

The game theorist Thomas Schelling proposed two models of racial segregation, which have both proved highly influential as computational/mathematical approaches to understanding social phenomena, and which contributed to his receiving the Nobel Memorial Prize in 2005. In each case, the model comprises a finite number of agents of two races, which we shall take to be green and red. The \emph{spacial proximity} or \emph{checkerboard model} entails agents taking up positions on a graph or grid (see \cite{S1}, \cite{S2}, \cite{S4}). Segregated regions may then appear as agents swap to neighbourhoods whose racial make-up is more to their liking. Subsequently, numerous authors (\cite{SS,DM,PW,GVN,GO}) have observed the structural similarity between this model and variants of the Ising model considered in statistical mechanics for the analysis of phase-transitions.

Schelling's second \emph{bounded neighbourhood} or \emph{tipping model} (see \cite{S1}, \cite{S2}, \cite{S3}, \cite{S4}) is a non-spacial model, in which a number of agents share a single neighbourhood, and where the initial proportions and preferences of the two races can give rise to total takeover by one race or the other. The simplest case is when no agent wishes to be in the minority, and move out when they are, to be replaced by an agent of the other race. In this case whichever race has more agents initially will take over totally, and thus the \emph{tipping point} is at the 50\% mark. Since Schelling's insights, tipping points have become a focus of interest in both the academic literature and popular culture. Models closely related to Schelling's have subsequently been investigated from a number of angles, notably in the work of Granovetter \cite{Gr} and as popularised by Gladwell \cite{Gl}.

More immediately, the current paper has its roots in the work of Brandt, Immorlica, Kamath, and Kleinberg \cite{BK}, which represented a major departure from the numerous previous studies of Schelling segregation, for the first time providing a rigorous mathematical analysis of an \emph{unperturbed} Schelling ring (which is to say a 1-dimensional spacial proximity model). Earlier work, such as that by Young \cite{Y} and Zhang (\cite{Z1}, \cite{Z2}, \cite{Z3}), had concentrated on \emph{perturbed} models, where agents have a small probability ($\eps$) of acting against their own interests. The introduction of this tiny amount of noise ensures that the resulting Markov process is reversible, and thus considerably easier to analyse. Yet it was established in  \cite{BK} that the unperturbed model may give rise to dramatically different patterns of segregation from the limiting case (letting $\eps \to 0$) of the corresponding perturbed model.

The breakthrough in \cite{BK} was built upon in \cite{BEL1} by the present authors, which also provided a thorough analysis of an unperturbed Schelling ring, but over a much larger range of parameters. The present model, which we shall describe shortly, continues in the same vein in again providing rigorous mathematical analyses of unperturbed 1-dimensional models, but represents a significant generalization again in terms of their parameters. The major innovation is that we enrich the model with the introduction of two extra parameters which break the symmetry between the two races, in two ways. Firstly, we allow the two races to exist in different numbers from one another initially. Secondly the two races may now exhibit unequal levels of racial tolerance. Both of these are natural extensions of previous models, and indeed were proposed by Schelling (see for instance \cite{S2} p 152).

Another important difference from the work of \cite{BK} and \cite{BEL1} is that the current model will be \emph{open}, meaning that at each time step, a single unhappy agent is selected uniformly at random and replaced by one of the opposite race, if doing so will cause the new agent to be happy. Most previous versions of the model are \emph{closed}, and the dynamic has involved two unhappy agents swapping at each time step. We thus assume the presence of an unlimited number of agents of both races outside the model who are ready to move in, given the opportunity. We remark that the openness of the current model brings it closer to the standard spin-1 models of statistical physics (although the closed variant also has a counterpart in systems employing Kawasaki dynamics). Progress has also recently been made by the authors in analysing closed models, similarly enriched with more parameters than \cite{BK} and \cite{BEL1}. See \cite{BEL2}.

We shall follow tradition in discussing our model in terms of racial segregation. However, as has often been remarked, it is equally applicable to any other geographical division of people along binary lines. Examples from \cite{S1} include women from men, students from faculty, or officers from enlisted men. What is more, the open model analysed here can equally well be interpreted in terms of static agents who modify some personal attribute in response to their neighbours, as happens in the class of \emph{voter models} (see for instance \cite{Lig}). In a voter model an agent will alter its state to mimic a randomly selected neighbour. In our setting, it is the proportions of its neighbours occupying the two states which determines its action. Thus we might propose an interpretation as a model of peer pressure, with two asymmetric states. Relevant examples may include whether or not to smoke \cite{CJR}, whether or not to clap at the end of a performance \cite{Mann}, or preferences for rival musical subcultures such as hip-hop or heavy metal \cite{STBM}.

Viewing the model from this perspective places it squarely within the category of cascading behaviour within networks. A central topic of study in this area is a \emph{general threshold model}. The setting here is a graph, in which every node $v$ is equipped with two things: a function $g_v$ which assigns a value $g_v(N)$ to every subset $N$ of the neighbours of $v$, and a threshold $\tau_v \in [0,1]$. Some nodes are initially \emph{active} while others are not. At each time step $t$, every inactive node $v$ computes $g_v(A^t_v)$ where $A^t_v$ is the set of neighbours of $v$ which are active at time $t$. Then $v$ activates at time $t+1$ if and only if $g_v(A^t_v) \geq \tau_v$. The primary question of interest here is to find conditions (on the set of initially active nodes, the topology of the graph, the functions $g_v$, and the thresholds $\tau_v$) which guarantee that the whole graph (or most of it) will eventually become activated, or alternatively that the cascade will quickly fizzle out. See \cite{Kl} for a good survey of this area.

The parallels with the study of Schelling segregation are striking. One major difference, however, is that while general threshold models evolve according to a  \emph{synchronous} dynamic (every agent that may change will do so at each time-step), the literature on Schelling segregation traditionally has one agent (or pair of agents) changing at each time-step. In the current work we shall consider variants in our Schelling model's dynamics including introducing synchronicity, and see that in many cases (but not all) our conclusions are unaffected by the choice between them.

Although our model is an instance of Schelling's spacial proximity model rather than any kind of hybrid or unified model, we nevertheless identify interesting phenomena in the spirit of his tipping or bounded neighbourhood models. That is to say, we shall identify thresholds in parameter-space on one one side of which one race takes over, and on the other side of which the other does. This behaviour is of course only possible in an open model, and furthermore is is only visible from our current \emph{asymptotic} perspective: we shall prove precise results concerning the ring's final configuration which are valid as the neighbourhood radius ($w$) grows large, and the ring size ($n$) grows large relative to $w$. More precisely, depending on the initial parameters, one of three conclusions will usually follow in the long run: either the ring will remain essentially static or one race or the other will take over. The asymptotic interpretation is critical here because takeover in this setting need not entail the complete absence of the other race, but rather takeover \emph{almost everywhere} in a measure-theoretic sense, meaning that in the ring's final configuration the majority race will almost certainly outnumber the other by any required margin for large enough values of $w$ and $n$. Thus takeover or staticity may not be apparent in simulations involving small $w$ and $n$. We will identify boundaries between these three regions within the parameter-space of the model.

While the results of \cite{BEL1} were somewhat counterintuitive (and perhaps politically discouraging) in that increased tolerance was seen in certain situations to lead to increased segregation, our results (described below) on the open model suggest the maxim ``tolerance wins out''. Loosely speaking more tolerant groups thrive at the expense of their less tolerant neighbours, although we emphasise that the details are highly sensitive to the initial proportions of the two groups. We also identify two very different regions of staticity, in which very few people move. These occur at the extremes: in one case a city comprising only very tolerant individuals in which almost everyone is happy with their neighbourhood. We think of this as the region of \emph{contentment}. In contrast, the region of \emph{frustration} comprises people so intolerant that, although almost all are unhappy with their neighbourhood, they are also unable to find anyone else prepared to move in, and thus are forced to stay put.

In more detail, and as already mentioned, the parameters in question represent a considerable generalisation of those from \cite{BEL1} in two directions. Firstly, we will no longer assume that the initial distribution is symmetric between the two races. Instead, each site will be occupied initially by a green agent with probability $\rho$, and by a red agent with probability $1-\rho$. Thus it might be that our model describes a racially homogeneous red region, into which a few green individuals have recently moved (meaning a small value of $\rho$). It is clearly of interest to be able to predict whether the newcomers will eventually take over the region, or will themselves be squeezed out.

Secondly, we drop the assumption that the preferences of the two races are simple mirror images, and allow the two groups to exhibit different tolerances. This is in line with social research, which has suggested in the past, for example, that black US citizens are happier in integrated neighbourhoods than their white compatriots (see for instance \cite{SSBK}). Thus we introduce two independent parameters $\tau_g$ and $\tau_r$ representing the tolerance of green and red agents, respectively.

The model runs as follows. First we fix the parameters $n, w \in \NN$ and $\rho, \tau_g, \tau_r \in (0,1)$. The ring then comprises nodes numbered $0$ to $n-1$. We arrange these in a circle, meaning that addresses are computed mod $n$ in everything that follows. Initially we populate the ring with agents of two races (red and green), with the colour of each node decided independently according to the toss of a biased coin, each node being green initially with a probability of $\rho$ and red with a probability of $1-\rho$. At each time-step, a node $x$ is solely concerned with its own neighbourhood of size $2w+1$, meaning the interval $\N(x)=[x-w,x+w]$ (understood mod $n$ as usual). If $x$ is a green (red) node, it will be \emph{happy} so long as the proportion of green (red) nodes in its neighbourhood is at least $\tau_g(2w+1)$ (respectively $\tau_r(2w+1)$), and \emph{unhappy} otherwise. We say that an unhappy node is \emph{hopeful} if a change of colour would cause it to be happy.

Now we introduce three possible dynamics by which the model may evolve:
\begin{itemize}[$\bullet$]
\item Our primary object of study will be the \emph{selective model}. Here, at each time-step a hopeful node is selected uniformly at random and its colour changed.

\item The \emph{incremental model} is similar: at each time-step an unhappy node is selected uniformly at random and its colour changed (regardless of whether this will make it happy).

\item In the \emph{synchronous model}, at each time-step every currently unhappy node alters its colour (again, regardless of the effect on their happiness).
\end{itemize}

In all cases, the process continues until no further changes are possible, at which stage we say the ring (or process) has \emph{finished}. (We shall establish in Lemma \ref{lem:itends} that this is guaranteed to occur for the selective model, and we shall later establish this in certain other cases.) Our principal concern is to find the probability that a randomly selected node is green in the finished ring. We shall show that this probability is usually close to either $\rho$, $0$, or $1$.

We shall describe these tipping phenomena in terms of numerical relationships between $\rho$, $\tau_g$, and $\tau_r$. In particular, for any $\rho$ there exist thresholds $\kappa^\rho_g$ \& $\kappa^\rho_r$ and $\mu^\rho_g = 1- \kappa^\rho_r$ \& $\mu^\rho_r = 1- \kappa^\rho_g$ as illustrated in Figure \ref{fig:Thresh1}.


\begin{figure}[!htbp]
\begin{tabular}{cc}
\def\svgwidth{7cm} 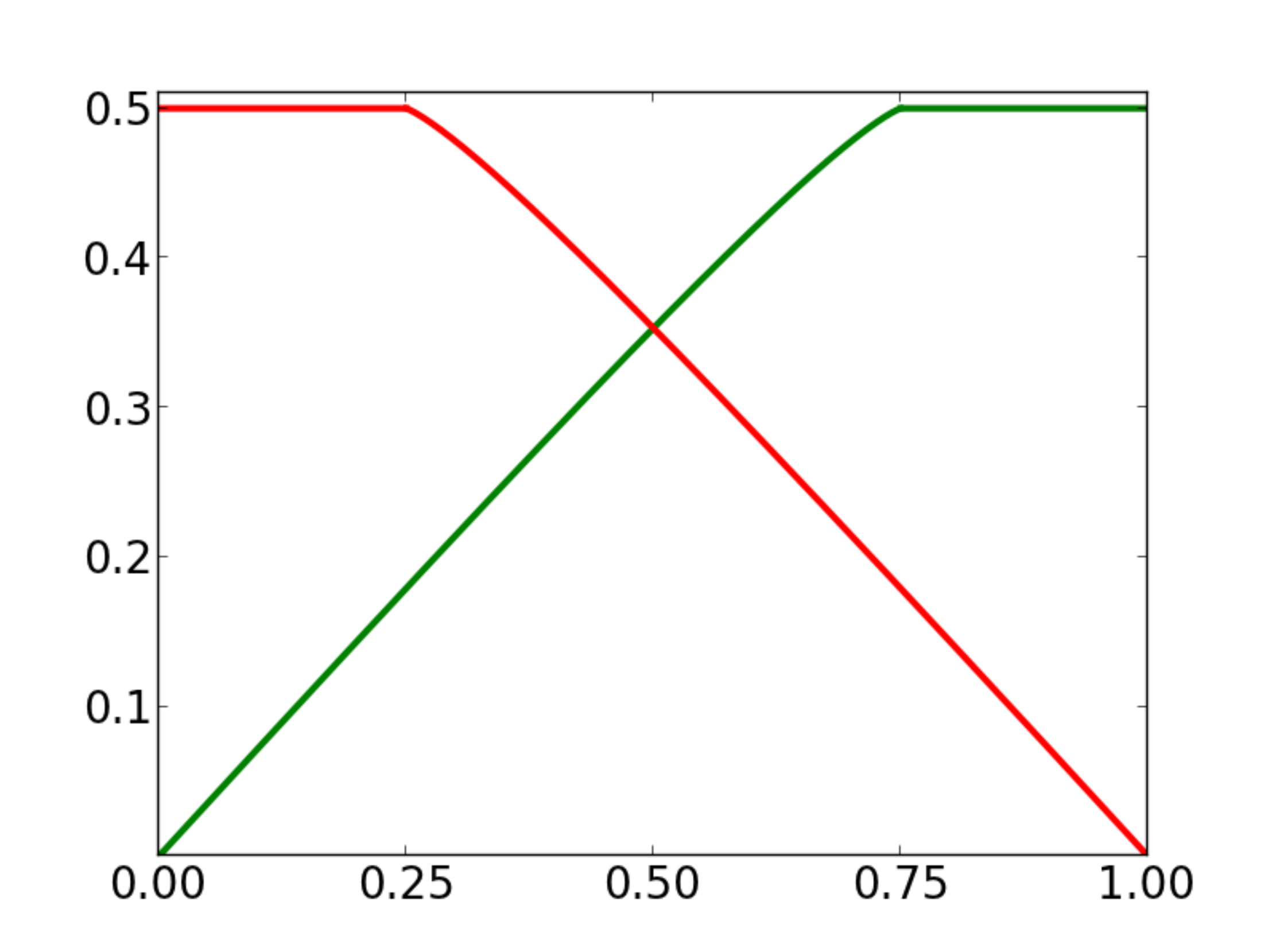 & \def\svgwidth{7cm} 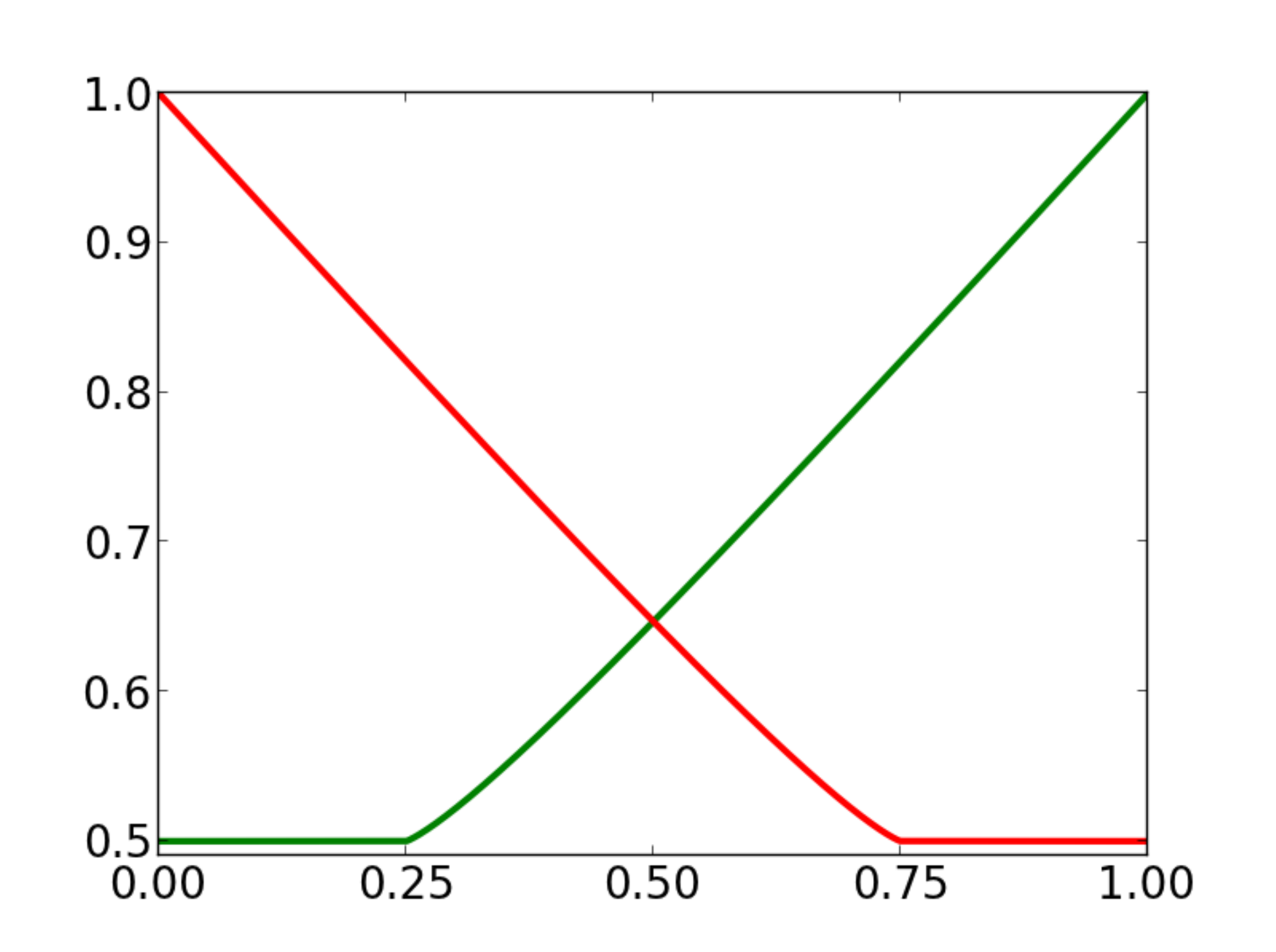\\
\end{tabular}
\caption{The thresholds $\kappa^\rho_g$ \& $\kappa^\rho_r$ and $\mu^\rho_g = 1- \kappa^\rho_r$ \& $\mu^\rho_r = 1- \kappa^\rho_g$}
\label{fig:Thresh1}
\end{figure}

The following give more details:

\begin{itemize}[$\bullet$]
\item For $\rho \leq \frac{1}{4}$ we have $\kappa^\rho_g < \kappa^\rho_r = \frac{1}{2} = \mu^\rho_g < \mu^\rho_r$.
\item For $\frac{1}{4} < \rho < \frac{1}{2}$ we have $\kappa^\rho_g < \kappa^\rho_r < \frac{1}{2}< \mu^\rho_g < \mu^\rho_r$.
\item For $\rho=\frac{1}{2}$ we have $\kappa^\rho_g = \kappa^\rho_r \approx 0.353092313$, which is the threshold $\kappa$ found in \cite{BEL1}, and $\mu^\rho_g = \mu^\rho_r = 1 - \kappa \approx 0.64690768667$.
\item For $\frac{1}{2} < \rho < \frac{3}{4}$ we have $\kappa^\rho_r < \kappa^\rho_g < \frac{1}{2} < \mu^\rho_r < \mu^\rho_g$.
\item For $\rho \geq \frac{3}{4}$ we have $\kappa^\rho_r < \kappa^\rho_g = \frac{1}{2} = \mu^\rho_r < \mu^\rho_g$.
\end{itemize}

As the pair $(\tau_g, \tau_r)$ ranges across the unit square, we shall find that the final configuration of the ring depends principally on where $\tau_g$ stands in relation to the thresholds $\kappa^\rho_g$, $\frac{1}{2}$, \& $\mu^\rho_g$ and where $\tau_r$ stands regarding $\kappa^\rho_r$, $\frac{1}{2}$, \& $\mu^\rho_r$, dividing the unit square into up to 16 open regions. We shall be able to analyse several of these simultaneously, but in some cases we shall find a more delicate dependency between $\tau_g$ and $\tau_r$. Throughout we shall leave open the intriguing question of what happens when the parameters exactly coincide with the thresholds. (We remark in passing, that the parameters of the models in \cite{BK} and \cite{BEL1} both constitute threshold cases of the current situation.) We shall also leave open the outcomes of the process in two small open regions of the parameter space (see Questions \ref{questionn:static2?} and \ref{questionn:static3?} below). We encourage others to investigate these matters.

As already mentioned, our results are asymptotic in nature, and we will use the shorthand ``all $n \gg w \gg 0$'' which carries the meaning ``all sufficiently large $w$, and all $n$ sufficiently large compared to $w$''. By a \emph{scenario} we mean the class of all rings with fixed values of $\rho$, $\tau_g$, and $\tau_r$, but $w$ and $n$ varying. We will identify a scenario with its signature triple $(\rho, \tau_g,\tau_r)$, and say that a value of $\rho$ \emph{admits} scenarios satisfying some property $X$, if there exist $\tau_g$ and $\tau_r$ such that $X$ holds for $(\rho, \tau_g,\tau_r)$. Our conclusions are then of three types:

\begin{itemize}[$\bullet$]
\item A scenario is \emph{static almost everywhere} if for every $\eps>0$ and all $n \gg w \gg 0$ a node $x$ chosen uniformly at random has a probability below $\eps$ of having changed colour at any stage before the ring finishes.
\item A scenario suffers \emph{green (red) takeover almost everywhere} if for every $\eps>0$ and all $n \gg w \gg 0$ a node $x$ chosen uniformly at random has a probability exceeding $1-\eps$ of being green (red) in the finished ring.
\end{itemize}

In some situations we will be able to strengthen the conclusion, and say that green (red) \emph{takes over totally} if for every $\eps>0$ and all $n \gg w \gg 0$ the probability that all nodes are green (red) in the finished ring exceeds $1-\eps$.

We now state our main results, and some open questions. (These depend on the existence of $\kappa^\rho_g$ \& $\kappa^\rho_r$ and $\mu^\rho_g$ \& $\mu^\rho_r$, which will be established rigorously in sections \ref{section:threshes} and \ref{section:themus} respectively.) Theorems \ref{thmm:static} - \ref{thmm:stag2} which follow are encapsulated (for the cases $\rho=0.42$ and $\rho=0.3$) by Figures \ref{fig:map42} and \ref{fig:map30}. Although the details of the diagrams are specific to $\rho=0.42$ and $\rho=0.30$, the major features apply more generally. Points coloured grey correspond to scenarios static almost everywhere, while green (red) points indicate green (red) takeover. Purple open regions represent those scenarios, other than those on the thresholds, whose outcome remains unclear. (Notice that there are no such regions in Figure \ref{fig:map42}, which is not unusual. See Questions \ref{questionn:static2?} and \ref{questionn:static3?}.) Similar diagrams for other values of $\rho$ are given in Figure \ref{fig:21}.


\begin{figure}[ht!]
\centering
\def\svgwidth{15cm}
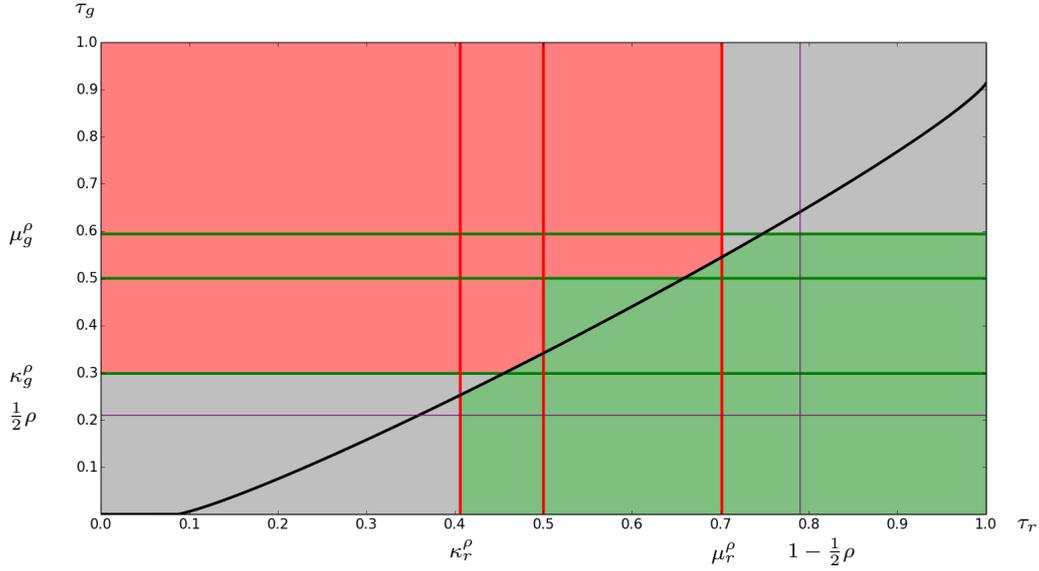
\caption{The landscape for $\rho=0.42$ under the selective dynamic}
\label{fig:map42}
\end{figure}

\begin{figure}[ht!]
\centering
\def\svgwidth{15cm}
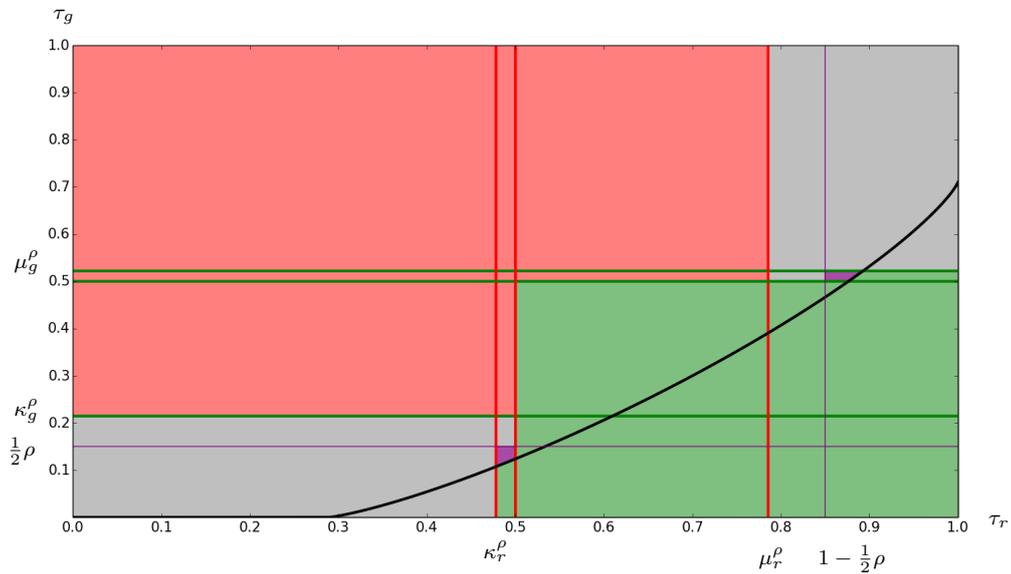
\caption{The landscape for $\rho=0.3$ under the selective dynamic}
\label{fig:map30}
\end{figure}

In each case the roles of red and green may be interchanged by swapping the relevant words, exchanging $\rho$ with $1-\rho$, and $\kappa^\rho_g$ with $\kappa^\rho_r$, and $\mu^\rho_g$ with $\mu^\rho_r$. Our first clutch of results (Theorem \ref{thmm:static} - \ref{thmm:straddle}) apply to all three dynamics, in the situation that at least one of $\tau_g, \tau_r < \frac{1}{2}$. 

\begin{thmm} \label{thmm:static}
Under all three dynamics, if $\tau_g<\kappa^\rho_g$ and $\tau_r<\kappa^\rho_r$ then the scenario will be static almost everywhere.
\end{thmm}

Scenarios where $\kappa^\rho_r < \tau_r < \frac{1}{2}$ and $\tau_g < \frac{1}{2}$ exhibit a more intricate dependency on $(\rho, \tau_g, \tau_r)$. To resolve matters here we require the following numerical condition (Definition \ref{defin:dom} below). A scenario $(\rho, \tau_g,\tau_r)$ where $\rho \in (0,1)$ and $\tau_g,\tau_r \in \left( 0,1 \right)$ and $\tau_g + \tau_r \neq 1$ is {\emph{red dominating}} if $$\rho \cdot \left( \tau_g^{\frac{\tau_g}{1-\tau_g - \tau_r}}\right) \left( 1 -\tau_g \right)^{\frac{1-\tau_g}{ 1- \tau_g - \tau_r}}   < (1-\rho) \cdot \left(\tau_r^{\frac{\tau_r}{1-\tau_g - \tau_r}} \right) \left(1-\tau_r\right)^{\frac{1-\tau_r}{1-\tau_g- \tau_r}} .$$
It is {\emph{green dominating}} if the reverse strict inequality holds.

\begin{thmm} \label{thmm:take1}
Under all three dynamics, if $\tau_g<\frac{1}{2}$ and $\kappa^\rho_r<\tau_r <\frac{1}{2}$ and $(\rho, \tau_g,\tau_r)$ is green dominating, then green will take over almost everywhere.
\end{thmm}

\begin{thmm} \label{thmm:static2}
Under all three dynamics, if $\tau_g<\kappa^\rho_g$ and $\kappa^\rho_r<\tau_r <\frac{1}{2}$, where the scenario is red dominating and $\tau_g> \frac{1}{2} \rho$, then the scenario is static almost everywhere.
\end{thmm}

\begin{questionn} \label{questionn:static2?}
What is the outcome, under any dynamic, if $\tau_g<\kappa^\rho_g$ and $\kappa^\rho_r<\tau_r <\frac{1}{2}$, the scenario is red dominating and $\tau_g \leq \frac{1}{2} \rho$?
\end{questionn}

We shall discuss this further at the end of Section \ref{section:longsect}, but remark that this problematic region only exists for a limited range of values of $\rho$, namely $\frac{1}{4} < \rho <0.3616$ approximately. (Of course exchanging the roles of red and green produces problematic scenarios in the range approximately $0.6384<\rho<\frac{3}{4}$.)

Notice that the following result has no dependency on $\rho$:

\begin{thmm} \label{thmm:straddle}
Under all three dynamics, if $\tau_g < \frac{1}{2} < \tau_r$ green will take over totally.
\end{thmm}

For the second batch of results (Theorems \ref{thmm:take2} - \ref{thmm:stag2}), we turn our attention to the case where both $\tau_g, \tau_r > \frac{1}{2}$. Here, the different dynamics diverge, and our focus will be on the selective case:

\begin{thmm} \label{thmm:take2}
Under the selective dynamic, if $\frac{1}{2} < \tau_g < \mu^\rho_g$ and $\frac{1}{2} < \tau_r$, and if $(\rho, \tau_g, \tau_r)$ is green dominating, then green will take over almost everywhere.
\end{thmm}

\begin{thmm} \label{thmm:stag1}
Under the selective dynamic, if $\frac{1}{2} < \tau_g < \mu^\rho_g$ and $\mu^\rho_r < \tau_r$, where the scenario is red dominating and additionally $\tau_r < 1- \frac{1}{2}\rho$, then the scenario is static almost everywhere.
\end{thmm}

\begin{questionn} \label{questionn:static3?} 
What is the outcome, under the selective dynamic, of scenarios where $\frac{1}{2} < \tau_g < \mu^\rho_g$ and $\mu^\rho_r < \tau_r$, the scenario is red dominating and $\tau_r \geq 1- \frac{1}{2}\rho$?
\end{questionn}

The range of $\rho$ for which this mysterious region exists is the same as that for Question \ref{questionn:static2?} above.

Our analysis of the selective dynamic culminates in a region of frustration as discussed earlier:

\begin{thmm} \label{thmm:stag2}
Under the selective dynamic, if $\mu^\rho_g < \tau_g$ and $\mu^\rho_r < \tau_r$, then the scenario is static almost everywhere.
\end{thmm}

For the incremental and synchronous dynamics, we shall leave the case $\tau_g, \tau_r > \frac{1}{2}$ largely open. However, we make the following conjecture:

\begin{conj} \label{conj:abovehalf}
Under the incremental and synchronous dynamics, suppose that $\frac{1}{2} < \tau_g < \tau_r$. Then green will take over totally.

If $\frac{1}{2} < \tau_g = \tau_r$ then for any fixed $\rho$ the probability of both red and green total takeover tends to $\frac{1}{2}$ as $w,n \to \infty$.
\end{conj}

Some intuition and partial results towards this conjecture are established in section \ref{section:bigtau}, along with a related discussion of the process' run-time.\\

\subsection{Interpreting the thresholds} \label{subsection:outline}

Before we prove Theorems \ref{thmm:static} - \ref{thmm:stag2} rigorously, let us outline the general intuition. Although there will be several technicalities to overcome, the overall strategy is not too complicated and, in the case where $\tau_g<\frac{1}{2}$, amounts to comparing the relative probabilities of unhappy nodes of each colour along with those of \emph{stable intervals} of each colour, in the initial configuration. 

The significance of initially unhappy red nodes is that they are likely to spark the growth of \emph{green firewalls}, which is to say runs of $\geq w+1$ successive green nodes. When $\tau_g<\frac{1}{2}$, such a firewall is guaranteed to grow until it hits a \emph{stably red interval}, meaning an interval of length $w+1$ containing enough red nodes (specifically $\geq \tau_r(2w+1)$ many) to ensure that all remain perpetually happy. It is not difficult to see that stably red intervals stop the growth of green firewalls, and that they are the only things to do so.

Given this picture, it is natural that the relative frequencies of stable intervals and unhappy nodes in the initial configuration should be important, and indeed we shall establish that such considerations are decisive. The various thresholds we identify within the region $\tau_g <\frac{1}{2}$ can be understood as follows:

\begin{itemize}[$\bullet$]
\item For $\tau_r>\frac{1}{2}$, stably red intervals cannot occur for all large enough $w$. Thus if $\tau_g<\frac{1}{2}< \tau_r$ as in Theorem \ref{thmm:straddle}, stable green intervals will likely exist in large enough rings (and will serve to prevent total red takeover), as will unhappy nodes of both colours, while stable red intervals will not, suggesting that, eventually, total green takeover cannot be resisted.

\item $\tau_g = \kappa^\rho_g$ is the point below which stably green intervals become more likely than unhappy green nodes. Thus if $\tau_g < \kappa^\rho_g$, stably green intervals are more numerous than unhappy green nodes, making red takeover unlikely.

\item Green domination will correspond to unhappy red nodes being more common than unhappy green nodes. Hence, if $\tau_g<\frac{1}{2}$ and $\kappa^\rho_r<\tau_r <\frac{1}{2}$ hold alongside green domination, as in Theorem \ref{thmm:take1}, there will be many more unhappy red nodes than green. Since stably red intervals are also infrequent, it follows that many more nodes will be consumed by green firewalls than by red.
\end{itemize}

When $\tau_g, \tau_r >\frac{1}{2}$, under the selective dynamic, similar considerations apply, with a couple of changes: in the place of unhappy nodes we consider \emph{hopeful} nodes. (Recall these are unhappy nodes for whom changing colour would produce happiness. This is only automatic, for large enough $w$, when $\tau_g + \tau_r < 1$.) In the place of stability we consider \emph{intractability}, which similarly obstructs the growth of firewalls (but cannot occur when $\tau_g, \tau_r <\frac{1}{2}$).

An interval $J$ of length $w+1$ is \emph{green intractable} if it contains so few green nodes (specifically $<\tau_g(2w+1) - (w+1)$ many), that no red node inside $J$ can ever become hopeful, no matter what occurs outside $J$. Thus no red node within $J$ will ever change colour. Such an interval is therefore, in this setting, the only thing which can halt the growth of a green firewall. We may now interpret the remaining thresholds:

\begin{itemize}[$\bullet$]
\item If $\tau_g > \mu^\rho_g$, then green intractable intervals are more likely than hopeful red nodes, making green takeover improbable.

\item Green domination has an alternative characterisation when $\tau_g, \tau_r >\frac{1}{2}$, as saying that hopeful red nodes are more likely than hopeful green nodes. If this holds, along with the assumptions that $\frac{1}{2} < \tau_g < \mu^\rho_g$ and $\frac{1}{2} < \tau_r$ as in Theorem \ref{thmm:take2}, then more green firewalls will start than red ones, and since there are few green intractable intervals to impede them, we may expect many more nodes to end up green than red.
\end{itemize}

\subsection{Arguing that the ring finishes} As a final step before we launch into an analysis of segregation patterns, we address the question of whether the process is guaranteed to finish. For the selective dynamic, the following result, whose proof is included in Appendix \ref{section:appA}, is sufficient for our purposes:

\begin{restatable}{lem}{itends}
\label{lem:itends}
For any scenario $(\rho,\tau_g,\tau_r)$ and for all large enough $w$, the selective dynamic guarantees that the process will finish.
\end{restatable}

The observant reader will notice from the proof that the requirement that $w$ be large is not necessary in scenarios where $\tau_g, \tau_r < \frac{1}{2}$. Indeed, we expect that it could be dropped in all cases, although this does not follow from the current argument.

Of course, this implies that a ring under an incremental dynamic and with $\tau_g, \tau_r < \frac{1}{2}$ will also finish. We shall establish certain other cases as consequences of the results in sections \ref{section:straddle} and \ref{section:bigtau}, however we do not have complete answers for the incremental and synchronous dynamics when $\tau_g, \tau_r > \frac{1}{2}$. We expect (indeed it is implicit in Conjecture \ref{conj:abovehalf}) that for any scenario, for any $\eps>0$, and for all large enough $n \gg w \gg 0$, the probability that the ring will finish eventually exceeds $1-\eps$.

\section{Stability and Contentment: below the threshold \texorpdfstring{$\kappa^\rho_r$}{kr}} \label{section:threshes}

We begin with some notation. Given a node $a$, we shall write $\N(a):=[a-w,a+w]$ for $a$'s neighbourhood. Given some collection $A$ of nodes and a time $t$, we write $G_t(A):=|\{ x \in A : x \textrm{ is green at time } t \}|$. We similarly define $R_t(A)$, $U_t(A)$, $F_t(A)$, $\UG_t(A)$, $\UR_t(A)$, $\HG_t(A)$, $\HR_t(A)$, $\FG_t(A)$ and $\FR_t(A)$  to be the number of red, unhappy, hopeful, unhappy green, unhappy red, happy green, happy red nodes, hopeful green, and hopeful red nodes in $A$ at time $t$ respectively, and will omit $t$ when its meaning is understood from context. Thus a green node $a$ is happy if $G(\N(a)) \geq \tau_g (2w+1)$ and an unhappy red node $b$ is hopeful if $G(\N(b)) \geq \tau_g(2w+1)-1$. We will also apply this in the case that $A=\{a\}$ is a singleton. Abusing notation slightly, $G(a)$ can be thought of as the green characteristic function of $a$, taking values $0$ or $1$.

Similarly, an interval $[a,a+w]$ of length $w+1$ is \emph{stably green} if $G[a,a+w] \geq \tau_g (2w+1)$.  We mentioned in \ref{subsection:outline} that stably green intervals have the ability to halt the growth of red \emph{firewalls} (stretches of at least $w+1$ consecutive red nodes). We make this precise:

\begin{lem} \label{lem:halt}
Suppose that $u_1$ and $u_2$ are nodes such that at time $t=0$ each of $u_1$ and $u_2$ lie in (possibly different) stably green intervals, and there is no unhappy green node in $[u_1, u_2]$. Then every green node in $[u_1,u_2]$ will remain perpetually and happily green.
\end{lem}

\begin{proof}
Suppose not. Then let $v$ in $[u_1,u_2]$ be the first green node to become unhappy. Now $v$ may only become unhappy once some other green node $v' \in \N(v)$ has turned red. Since $v' \not\in [u_1,u_2]$, either $v'<u_1$ or $v'>u_2$. We assume without loss of generality that $v'<u_1$. Then $v' \in \N(u_1)$. Now by assumption, $u_1 \in [a,a+w]$ for some stably green interval $[a,a+w]$. By stability, we cannot have $v' \in [a,a+w]$. Thus $v' \in [u_1 - w, a-1]$, whence it follows that $[a, a+w] \subseteq \N(v)$ meaning, by stability, that $\N(v)$ contains enough green nodes to keep $v$ happy, which is a contradiction.
\end{proof}

Notice that the possibility of stably green intervals in arbitrarily large rings requires that $\tau_g \leq \frac{1}{2}$. In fact we shall assume that $\tau_g < \frac{1}{2}$ throughout this section, unless stated otherwise. Our goal is to prove the following:

\begin{prop} \label{prop:kappasumup}
For any $\rho \in (0,1)$, we work in the initial condition, and let $U_g$ be the probability that a uniformly randomly selected green node is unhappy, and $S_g$ be that of a uniformly randomly selected node lying within a stably green interval. Then there exists a threshold $\kappa^\rho_g$ satisfying $\rho > \kappa^\rho_g>\frac{1}{2} \rho$, defined as the unique root of the equation $$f(s):=\frac{ \left( \frac{1}{2} - s \right)^{(1 - 2 s)}}{(1-s)^{2(1- s)}} = \frac{1}{2 (1 - \rho)}.$$ This is such such that for any $\tau_g \in (0,1)$:

\begin{itemize}[$\bullet$]
\item If $\tau_g < \kappa^\rho_g$, there exists $\zeta \in (0,1)$ so that $U_g < \zeta^{w} S_g$ for all $w$.
\item If $\tau_g > \kappa^\rho_g$, there exists $\zeta \in (0,1)$ so that $S_g < \zeta^{w} U_g$ for all $w$.
\end{itemize}

Similarly, there exists a threshold $\kappa_r^\rho$ where $(1 - \rho) > \kappa^\rho_r>\frac{1}{2}(1- \rho)$, defined as the unique root of $f(s)=\frac{1}{2 \rho}$ such that corresponding statements about $U_r$ and $S_r$ hold.
\end{prop}

With threshold in place, we shall argue that when $\tau_g< \kappa^\rho_g$, any randomly selected node is highly likely be closer on each side to a stably green interval than to an unhappy green node in the initial configuration. This will establish that such a node can never turn red, and will be enough to establish Theorem \ref{thmm:static}.\\

Before we proceed with the proof of \ref{prop:kappasumup}, we recall two important probabilistic results. The first is a classical result from \cite{Hoeff}. (The full statement is more general, but this is the version which shall find most useful.)

\begin{prop}[Hoeffding's inequality] \label{prop:Hoeffding}
Let $X_1, \ldots, X_N$ be independent random variables such that $\textbf{P}(X_i=1)=p$ and $\textbf{P}(X_i=0)=1-p$. Then for any $\delta>0$ we have $$\textbf{P}\left( \sum_{i=1}^N X_i \geq (p + \delta)N \right) \leq \exp \left( -2N \delta^2 \right).$$
\end{prop}

Secondly, we shall require Theorem 1.1 from \cite{Bollo}, which appears as Lemma 3.1 in \cite{BEL1}, and which we restate:

\begin{lem} \label{lem:binom}
Suppose $h: \bf{Z} \to \bf{Z}$ and $p \in (0,1)$ are such that there exist $k \in (0,1)$ so that for all large enough $N$, we have $\left( 1+ \left( \frac{1}{p} -1 \right) k \right) h(N)> N \geq h(N)> pN > 0$. Then for all large enough $N$, if $X_N \sim b(N,p)$, we have 
$$P \left( X_N = h(N) \right) \ \ \leq \ \  P \left( X_N \geq h(N) \right) \ \ \leq \ \ \left( \frac{1}{1-k} \right) \cdot P \left( X_N = h(N) \right).$$
That is to say in asymptotic notation, $P \left( X_N \geq h(N) \right)= \Theta \left( P \left( X_N = h(N) \right) \right)$.
\end{lem}

For current purposes, the appropriate asymptotic notion is weaker than $\Theta$:

\begin{rem} \label{rem:asym}
If $f$ and $g$ are functions of $w$, we shall write $f \approx g$ to mean that there are rational functions $P$ and $Q$ such that $P(w), Q(w)>0$ and $P(w) g(w)  \leq f(w) \leq Q(w) g(w)$ for all large enough $w$.
\end{rem}

\begin{proof}[Proof of Proposition \ref{prop:kappasumup}]

We fix a scenario $(\rho, \tau_g, \tau_r)$ and work always in the initial configuration. For some green node $b$, we wish to compare the probability $U_g$ that $b$ is unhappy with the probability $S_g$ that $[b-i,b+w-i]$ is stably green for some $i$ where $0 \leq i \leq w$. Our first step is to approximate $S_g$ by focusing on the case $i=0$. Let $S^0_g$ be the probability that $[b,b+w]$ is stably green. Then $S_g \leq (w+1)S^0_g$, meaning that $S_g \approx S^0_g$. We shall therefore work with $S^0_g$ in place of $S_g$, and observe later that this introduces no problems.

Hence the first probability we shall compute is that of the interval $[b,b+w]$ nodes being stably green. Here, the relevant distribution is binomial: $X \sim b(w,\rho)$, describing the number of green elements other than $b$ in the interval, and $S_g \approx {\bf P}(X \geq \tau_g(2w +1) -1)$.

For understanding the likelihood of a green node $b$ being unhappy, it will be convenient to count the red nodes in $\N(b)$. This is given by the distribution $Z \sim b(2w,1-\rho)$. Then $U_g = {\bf P}(Z> (1-\tau_g)(2w+1))$.

\begin{rem} \label{rem:easycases}
The behaviour of $\frac{U_g}{S_g}$ is easy to determine in certain situations:

\begin{itemize}[$\bullet$]
\item If $\tau_g \leq \frac{\rho}{2}$ then $S_g \geq \frac{1}{2}$ while $U_g \to 0$.
\item If $\tau_g = \rho$ then $S_g \to 0$ and $U_g \to \frac{1}{2}$.
\item If $\tau_g > \rho$, then $S_g \to 0$ and $U_g \to 1$.
\end{itemize}
All limits are taken as $w \to \infty$. Furthermore, it is a straightforward consequence of Hoeffding's inequality (Proposition \ref{prop:Hoeffding}) that the quantities tending to $0$ do so at an exponential rate in $w$, meaning that they are bounded above by $\nu^w$ for some $\nu \in (0,1)$.
\end{rem}

Hence for the remainder of this proof we shall concentrate on scenarios where  $\frac{\rho}{2}< \tau_g < \rho$. We now apply Lemma \ref{lem:binom} in the current context, with $N=w$, $p= \rho$, and $h(N) = \lceil (2w +1)\tau_g \rceil -1$. Furthermore, making the assumption that $\tau_g > \frac{\rho}{2}$ we may find $k$ where $1>k>\frac{\rho}{1-\rho} \cdot \frac{\frac{1}{2}- \tau_g}{\tau_g}$. Thus we get  

$$S_g \approx \rho^h (1-\rho)^{w-h} \bpm w \\ h \epm.$$
Similarly, assuming only that $\tau_g< \rho$, we may take $N=2w$, $p=1- \rho$, and choose $k'$ so that $1> k'> \frac{1-\rho}{\rho} \cdot \frac{\tau_g}{1-\tau_g}$, to get
\begin{equation} \label{equation:unhap} U_g \approx (1-\rho)^{h'} \rho^{2w-h'} \bpm 2w \\ h' \epm \end{equation}
where $h'=\lfloor (1-\tau_g)(2w+1) \rfloor +1$.
Putting these two estimates together, so long as $\frac{\rho}{2}< \tau_g < \rho$, we find $$\frac{U_g}{S_g} \approx (1-\rho)^{h'+h -w} \rho^{2w-h'-h}\frac{\bpm 2w \\ h' \epm}{\bpm w \\ h \epm}.$$

\noindent We now employ Stirling's formula, that $n! \approx n^{n + \frac{1}{2}} e^{-n}$. Then, the powers of $e$ cancel and we see:
\begin{equation} \label{equation:R2} \frac{U_g}{S_g} \approx (1-\rho)^{h'+h -w} \rho^{2w-h'-h} \frac{(2w)^{2w+ \frac{1}{2}}(h)^{h+\frac{1}{2}} (w-h)^{w-h+\frac{1}{2}}}{(h')^{h'+\frac{1}{2}}(2w-h')^{2w-h'+\frac{1}{2}} (w)^{w+\frac{1}{2}}}. \end{equation}

Now we introduce the approximations $2w \tau_g$ and $2w (1-\tau_g)$ for $h$ and $h'$ respectively. Notice that $|h - 2w \tau_g| \leq 1 $ and $|h' - 2w (1-\tau_g)| \leq 2$. It follows easily that $h^{h + \frac{1}{2}} \approx (2w \tau_g)^{2w \tau_g + \frac{1}{2}}$ with `$\approx$' interpreted as in Remark \ref{rem:asym}. (We observe in passing that this estimate would not hold under the asymptotic notion $\Theta$.) Similar remarks apply to the other terms in the estimate, allowing us to deduce the following:

$$\frac{U_g}{S_g} \approx (1-\rho)^{w} \frac{(2w)^{2w+ \frac{1}{2}} (w(1-2 \tau_g))^{w(1-2 \tau_g)+\frac{1}{2}}}{(2w(1-\tau_g))^{2w(1-\tau_g)+\frac{1}{2}} (w)^{w+\frac{1}{2}}}.$$

Hence we obtain the following key estimate:
\begin{equation} \label{equation:balance} \frac{U_g}{S_g} \approx \left( \frac{ (\frac{1}{2} - \tau_g)^{(1 - 2 \tau_g)}}{(1-\tau_g)^{2(1- \tau_g)}} \cdot 2 \cdot (1-\rho) \right)^w. \end{equation}

The question now is whether the term inside the brackets in \ref{equation:balance} is greater than or less than $1$. In many cases there is a threshold, $\kappa^\rho_g$, where it is equal to $1$. That is, $\kappa^\rho_g$ is the root, if it exists, of the equation: \begin{equation} \label{equation:froot} f(s) := \frac{ \left( \frac{1}{2} - s \right)^{(1 - 2 s)}}{(1-s)^{2(1- s)}} = \frac{1}{2 (1 - \rho)}. \end{equation}

To establish the existence of this root we shall appeal to the intermediate value theorem, noticing first that for $0<s<\frac{1}{2}$ we have $f'(s)>0$ meaning that a root, if it exists, will be unique.

We have required that $\frac{\rho}{2}<\tau_g<\rho$ and $\tau_g < \frac{1}{2}$. Now we claim the following:
\begin{enumerate}[(i)]
\item For $0< \rho < \frac{1}{2}$, we have $f(\rho)>\frac{1}{2 (1 - \rho)}$.
\item For $0< \rho < \frac{3}{4}$ we have $\frac{1}{2 (1 - \rho)}>f(\frac{\rho}{2})$.
\end{enumerate}

To prove (i), define $g_1(\rho):=(1-\rho)f(\rho) = \left(\frac{\frac{1}{2} - \rho}{1-\rho} \right)^{1-2 \rho}$. We shall show that $g_1(\rho)>\frac{1}{2}$. Notice that $g_1(0)=\frac{1}{2}$ so it suffices to show that $g_1'(\rho)>0$. Taking logarithms and differentiating, we find that $g_2(\rho):=\frac{g_1'(\rho)}{g_1(\rho)} = 2 \ln (1 - \rho) -2 \ln (\frac{1}{2} - \rho) - \frac{1}{1-\rho}$. Since $g_1(\rho)>0$ it suffices to show that $g_2(\rho)>0$. Well $g_2(0)=2 \ln 2 - 1>0$ and differentiating again establishes that $g_2'(\rho)>0$.

The proof of (ii) is similar. Define $g_3(\rho):=(1-\rho)f \left( \frac{\rho}{2} \right) = \left( \frac{1}{2}\right)^{1-\rho} \left( \frac{1-\rho}{1- \frac{\rho}{2}} \right)^{2-\rho}$. We'll show that $g_3(\rho)< \frac{1}{2}$ by a similar argument. Notice that $g_3(0)=\frac{1}{2}$, hence it will suffice to show that $g_3'(\rho)<0$. Again, we take logarithms and differentiate, getting $g_4(\rho):=\frac{g_3'(\rho)}{g_3(\rho)}= \ln 2 + \ln \left( 1 - \frac{\rho}{2} \right) - \ln \left( 1- \rho \right) - \frac{1}{1-\rho}$.  Again, $g_3(\rho)>0$ and we shall show $g_4(\rho)<0$. Well, $g_4(0) = \ln 2 - 1 <0$ and again differentiating establishes that $g_4'(\rho)<0$.

What is more, for $\frac{1}{2} \leq \rho < \frac{3}{4}$, it holds that $f(\frac{1}{2})>\frac{1}{2 (1 - \rho)}$, where we extend by continuity to take $f\left( \frac{1}{2} \right)=2$. Along with (i) and (ii), this allows us to apply the intermediate value theorem.

Hence for any $0 < \rho <\frac{3}{4}$ the threshold $\kappa^{\rho}_g$ exists. For $\tau_g<\kappa^{\rho}_g$ we will have $S_g \gg U_g$ for all large enough $w$, while for $\tau_g>\kappa^{\rho}_g$ we will have $U_g \gg S_g$. 

However, for $\rho \geq \frac{3}{4}$, we get $f(s)<\frac{1}{2(1 - \rho)}$ for all $s<\frac{1}{2}$. Hence in this region $S_g \gg U_g$, for all large enough $w$, whatever the value of $\tau_g < \frac{1}{2}$. On the other hand, for $\tau_g > \frac{1}{2}$ we have $S_g = 0 < U_g$. Thus it makes sense to set $\kappa^{\rho}_g:=\frac{1}{2}$ in this case.

We can similarly compute $\kappa^\rho_r$, simply by replacing $\tau_g$ with $\tau_r$ and $1-\rho$ by $\rho$ in the above analysis, making $\kappa^\rho_r$ the root of $ f(s) = \frac{1}{2 \rho}$. By symmetry, we find that for $\frac{1}{4} < \rho <1$ the threshold $\kappa^{\rho}_r$ exists. But when $\rho \leq \frac{1}{4}$, we have that $S_r \gg U_r$ for all large enough $w$ whatever the value of $\tau_r < \frac{1}{2}$.

Finally, recall that at the start of the proof we made the approximation $S_g \approx S^0_g$. Since estimate \ref{equation:balance} is exponential in $w$, the asymptotic limits are unaffected by this move, meaning that  $\kappa^\rho_g$ and $\kappa^\rho_r$ represent exactly the thresholds we seek. Combining these observations with Remark \ref{rem:easycases} (and the impossibility of stably green intervals when $\tau_g > \frac{1}{2}$) we have completed the proof of Proposition \ref{prop:kappasumup}. \end{proof}

We may now build towards the proof of our first theorem, that a scenario where $\tau_g<\kappa^\rho_g$ and $\tau_r<\kappa^\rho_r$ will be static almost everywhere.

We begin at a node $u_0$ selected uniformly at random. Looking outwards from $u_0$ in both directions, we may encounter unhappy nodes and/or stable intervals of both colours. We need to understand the most likely order in which we will meet these. It seems plausible, by Proposition \ref{prop:kappasumup}, that we are more likely to find green stable intervals before unhappy green nodes, and red stable intervals before unhappy red nodes. Establishing this will suffice, as Lemma \ref{lem:halt} then guarantees that there can then be no way for the influence of any unhappy node to reach $u_0$, which must therefore remain unchanged.

We restate the following, which is Lemma 3.2 from \cite{BEL1}, recalling that ``the first node to the left'' of some given node $u$ satisfying some criterion means the first in the sequence $u,u-1,u-2,\cdots$ to satisfy the condition.

\begin{lem} \label{lem:gen}  

Let $P(u)$ and $Q(u)$ be events which only depend on the neighbourhood of $u$ in the initial configuration, meaning that if the neighbourhood of $v$ in the initial configuration is identical that of  $u$ (i.e. for all $i \in [-w,w]$, $u+i$ is of the same type as $v+i$), then $P(u)$ holds if and only if $P(v)$ holds and similarly for $Q(u)$ and $Q(v)$. Suppose also that:
\begin{enumerate}[(i)]  
\item ${\bf P}(P(u))\neq 0$ and ${\bf P}(Q(u))\neq 0$. 
\item  For all $k$, for all sufficiently large $w$ compared to $k$, ${\bf P}(P(u))/{\bf P}(Q(u)) >kw$.
\end{enumerate}   
For any $u$, let $x_u$ be the first node to the left of $u$ such that either $P(x_u)$ or $Q(x_u)$ holds.  For any $\eps>0$, if $0\ll w \ll n$ then the following occurs with probability $>1-\eps$ for $u$ chosen uniformly at random:  $x_u$ is defined and for no node $v$  in $[x_u-2w,x_u]$ does $Q(v)$ hold.

An analogous result holds when `left' is replaced by `right'. 
\end{lem} 

We can now establish Theorem \ref{thmm:static}. We apply Lemma \ref{lem:gen}, interpreting $P(u)$ as the event that the node $u$ lies in a green stable interval and $Q(u)$ as its being green and unhappy, with Proposition \ref{prop:kappasumup} providing the necessary probabilistic bounds. This tells us that, for any $\eps'>0$ and all large enough $ n \gg w \gg 0$, if we pick a node $u_0$ uniformly at random, then with probability $>1- \eps'$ the nearest stable green intervals to $u_0$ will be closer on both sides than the nearest unhappy green nodes. Thus by Lemma \ref{lem:halt}, $u_0$, if green will never turn red.

Then we simply repeat the argument with the roles of red an green interchanged, noting that if two events each have probability tending to $1$, then so must their conjunction.

\section{Unhappiness and Domination} \label{section:domsect}

To analyse the case $\kappa^\rho_r < \tau_r < \frac{1}{2}$, we need to answer the following question: in the initial configuration, which colour is likely to yield more unhappy nodes?

\begin{defin} \label{defin:dom}
A scenario $(\rho,\tau_g, \tau_r)$ where $\tau_g + \tau_r \neq 1$ is {\emph{red dominating}} if $$\rho \cdot \left( \tau_g^{\frac{\tau_g}{1-\tau_g - \tau_r}}\right) \left( 1 -\tau_g \right)^{\frac{1-\tau_g}{ 1- \tau_g - \tau_r}}   < (1-\rho) \cdot \left(\tau_r^{\frac{\tau_r}{1-\tau_g - \tau_r}} \right) \left(1-\tau_r\right)^{\frac{1-\tau_r}{1-\tau_g- \tau_r}} .$$
It is {\emph{green dominating}} if the reverse strict inequality holds.
\end{defin}

Our choice of terminology will be justified below in Propositions \ref{prop:outnumber} and \ref{prop:outnumber2}. Firstly, however we establish some facts about domination, deferring the proof until Appendix \ref{section:appB}:

\begin{restatable}{lem}{domfacts}
\label{lem:domfacts}
Let $S:= \left( 0, 1 \right) \times \left( 0, 1 \right)$. We divide $S$ into the two triangles $T_1:=\{(x,y) \in S : x+y<1 \}$ and $T_2:=\{(x,y) \in S : x+y>1 \}$ and the line $L=\{(x,y) \in S : x+y=1 \}$. Also define $S_1:= \left( 0, \frac{1}{2} \right) \times \left( 0, \frac{1}{2} \right)$ and $S_2:= \left( \frac{1}{2},1 \right) \times \left( \frac{1}{2},1 \right)$. (Notice that $S_i \subset T_i$.) Then the following hold:
\begin{enumerate}
\item Suppose that $(\tau_g, \tau_r), (\tau_g', \tau_r') \in T_i$ and that $(\rho, \tau_g, \tau_r)$ is red dominating. If $\tau_g' \geq \tau_g$, and $\tau_r \geq \tau_r'$ then $(\rho, \tau_g', \tau_r')$ is red dominating. Conversely, if $(\rho, \tau_g', \tau_r')$ is green dominating, so too is $(\rho, \tau_g, \tau_r)$.
\item For $i \in \{1,2\}$, every scenario where $\rho \leq \frac{1}{5}$ (respectively $\rho \geq \frac{4}{5}$) and $(\tau_g, \tau_r) \in S_i$ is red (green) dominating.
\item Any value of $\rho$ where $\frac{1}{5} < \rho < \frac{4}{5}$ admits both red and green dominating scenarios in both $S_1$ and $S_2$.
\end{enumerate}
\end{restatable}

In some cases, red domination is easy to determine:

\begin{restatable}{coro}{easierdom}
\label{coro:easierdom}
Suppose $(\rho,\tau_g, \tau_r)$ is a scenario where $\tau_g+\tau_r \neq 1$ and $\tau_g \geq \rho$ and $\tau_r \leq 1- \rho$. Then $(\rho,\tau_g, \tau_r)$ is red dominating. (Similarly, green domination follows when both of the reverse inequalities hold.)
\end{restatable}

Again we defer the proof to Appendix \ref{section:appB}. The following result justifies our choice of terminology for scenarios where $\tau_g, \tau_r < \frac{1}{2}$:

\begin{prop} \label{prop:outnumber}
Suppose that $\tau_g, \tau_r < \frac{1}{2}$. The scenario $(\rho, \tau_g, \tau_r)$ is red dominating if and only if there exists $\eta \in (0,1)$ so that for all $w$ we have $U_r < \eta^w U_g$.

The same holds with the roles of red and green interchanged. \end{prop}

\begin{proof}

Suppose first that $\tau_g \geq \rho$, then automatically $1-\tau_r \geq \rho$, so we may apply Corollary \ref{coro:easierdom} to establish red domination. Also, $U_g \to 1$ (if $\tau_g > \rho$) or $U_g \to \frac{1}{2}$ (if $\tau_g = \rho$) as $n \gg w \to \infty$. Meanwhile $U_r \to 0$ at an exponential rate in $w$. Thus the result follows.

By an identical argument, if $\tau_r \geq 1- \rho$ then automatically $\tau_g < \rho$ and the result follows by Corollary \ref{coro:easierdom} with the roles of red and green exchanged.

This leaves us with the case where $\tau_g < \rho$ and $\tau_r < 1 - \rho$. Under the assumption that $\tau_g < \rho$, in Equation \ref{equation:unhap}, we derived an asymptotic expression for $U_g$. Applying Stirling's formula and the approximation $h' \approx 2w (1-\tau_g)$ as previously, it follows that
$$U_g \approx (1-\rho)^{2w(1-\tau_g)} \rho^{2w \tau_g} \frac{(2w)^{2w+\frac{1}{2}}}{(2w(1-\tau_g))^{2w(1-\tau_g)+\frac{1}{2}}(2w \tau_g)^{2w \tau_g+\frac{1}{2}}}.$$

Of course, by interchanging $\rho$ with $1- \rho$, as well as $\tau_g$ with $\tau_r$, under the assumption that $\tau_r < 1- \rho$ we may form an analogous expression for $U_r$. We may now take the quotient of these two, to find

\begin{equation} \label{equation:unhapratio}
\frac{U_g}{U_r} \approx \left( \frac{(1 - \tau_r)^{1 - \tau_r} \cdot \tau_r^{\tau_r}}{(1 - \tau_g)^{1 - \tau_g} \cdot \tau_g ^{\tau_g}} \cdot \rho^{\tau_g + \tau_r - 1} \cdot (1 - \rho)^{1 - \tau_g - \tau_r} \right)^{2w}.
\end{equation}

The term within the bracket is then $>1$ (respectively $<1$) if and only if $(\rho, \tau_g, \tau_r)$ is red (green) dominating, thus establishing the result.
\end{proof}

Figure \ref{fig:21} depicts the boundary between red and green domination for a few values of $\rho$. Other details are also shown (exactly as in Figure \ref{fig:map42}), in particular green (red) points on the plane represent scenarios which suffer green (red) takeover, while grey points represent static scenarios. Recall from the discussion following \ref{equation:balance} that for $\rho \leq \frac{1}{4}$ our zone of current interest $\kappa^\rho_r < \tau_r < \frac{1}{2}$ does not exist, with the fate of each scenario entirely determined by the value of $\tau_r$ relative to $\frac{1}{2}$ and $\mu^\rho_r$ and of $\tau_g$ relative to $\kappa^\rho_g$ and $\frac{1}{2}$. Nevertheless red/green domination still makes sense as a numerical condition.

For $\rho > \frac{1}{4}$, the zone $\kappa^\rho_r < \tau_r <\frac{1}{2}$ does exist, and we observe a form of threshold at $\rho = \lambda \approx 0.38493708$ where $\tau^\lambda_g \approx 0.27407242$ and $\tau^\lambda_r \approx 0.42832491$. We shall not need to refer again to $\lambda$, so let us briefly discuss it here. By definition, $\lambda$ is such that $(\frac{1}{2}, \kappa^\lambda_g)$ and $(\kappa^\lambda_r, \frac{1}{2})$ lie exactly on the boundary between red and green domination. The point of it, therefore, is that for $\tau_r< \frac{1}{2}$ and $\rho< \lambda$, green domination automatically implies that $\tau_g < \kappa^\rho_g$.

However, for $\lambda< \rho \leq \frac{1}{2}$ green dominating scenarios are also admissible within the zone $\kappa^\rho_r < \tau_r < \frac{1}{2}$ and $\kappa^\rho_g < \tau_g < \frac{1}{2}$. Notice that at $\rho=\frac{1}{2}$, the threshold between red and green domination is simply the line $\tau_g=\tau_r$.

It is clear from Figure \ref{fig:21}, that for $\rho > \frac{1}{4}$, red/green domination can play a decisive role in determining the fate of any scenario. We shall now prove this.


\begin{figure}[!htbp]
\begin{tabular}{cc}
\includegraphics[width=8cm]{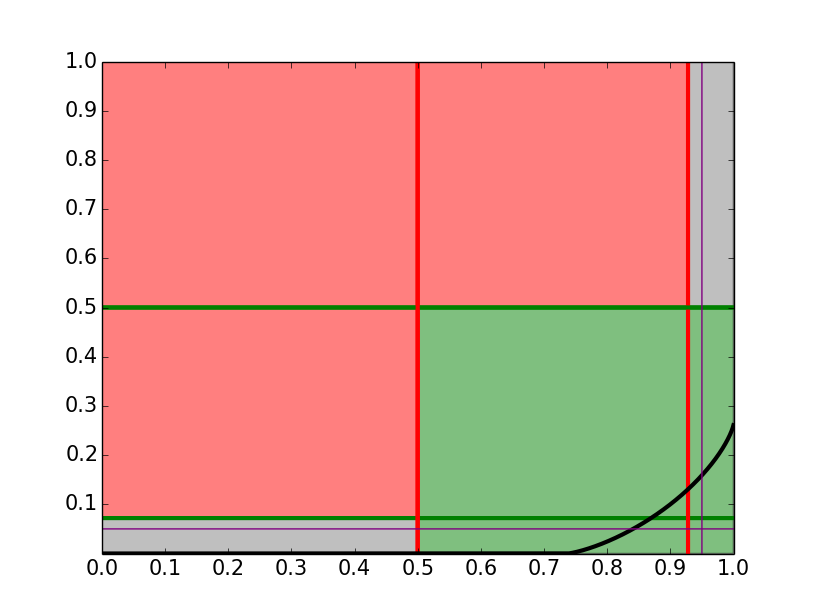}  & \includegraphics[width=8cm]{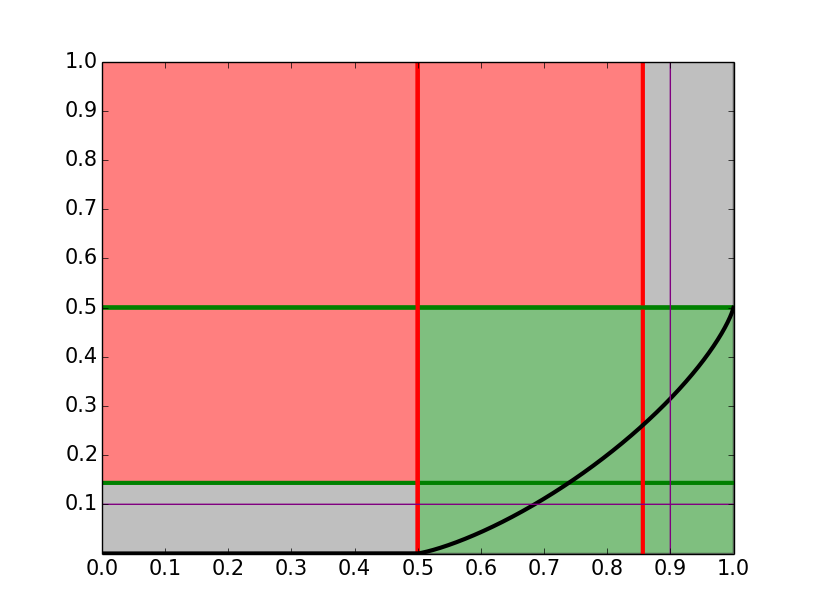}\\
$\rho=0.1$ & $\rho=0.2$\\
\includegraphics[width=8cm]{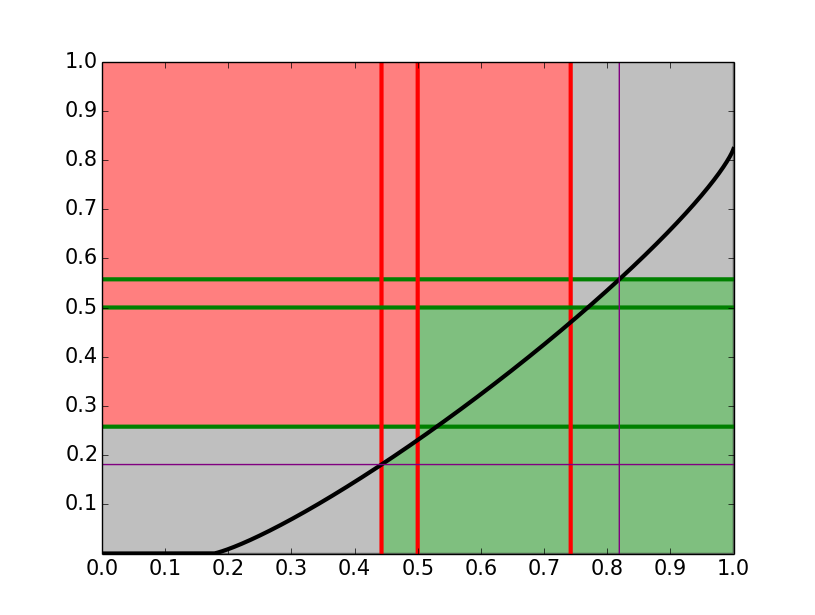} & \includegraphics[width=8cm]{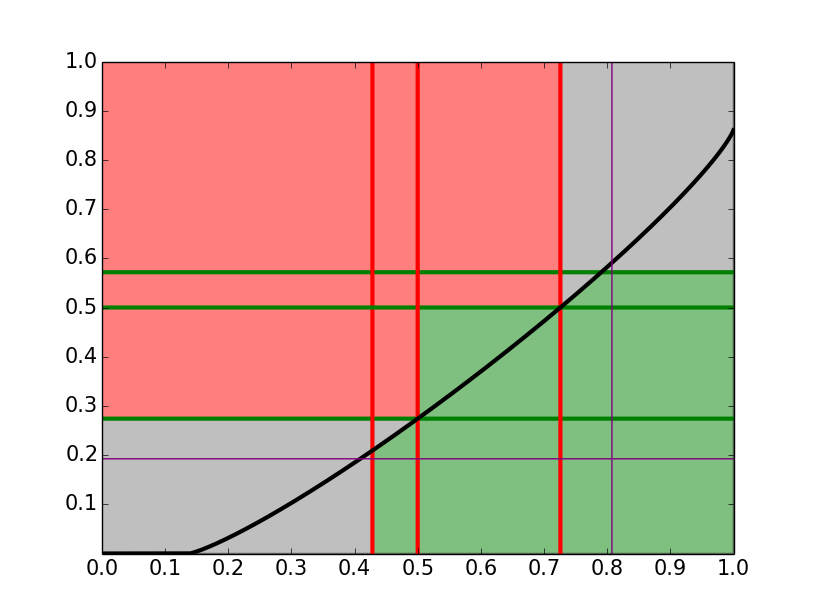}\\
$\rho=0.3616$ & $\rho=\lambda \approx 0.38493708$\\
\includegraphics[width=8cm]{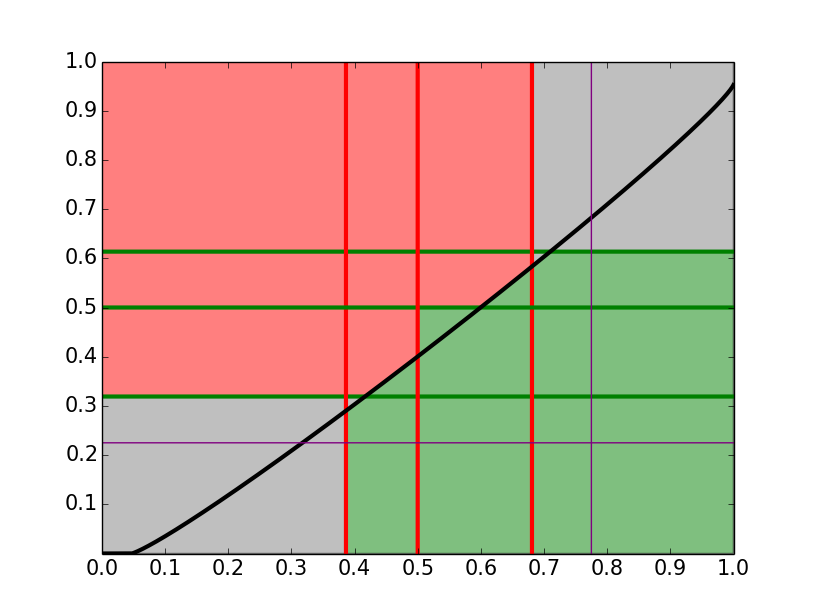} & \includegraphics[width=8cm]{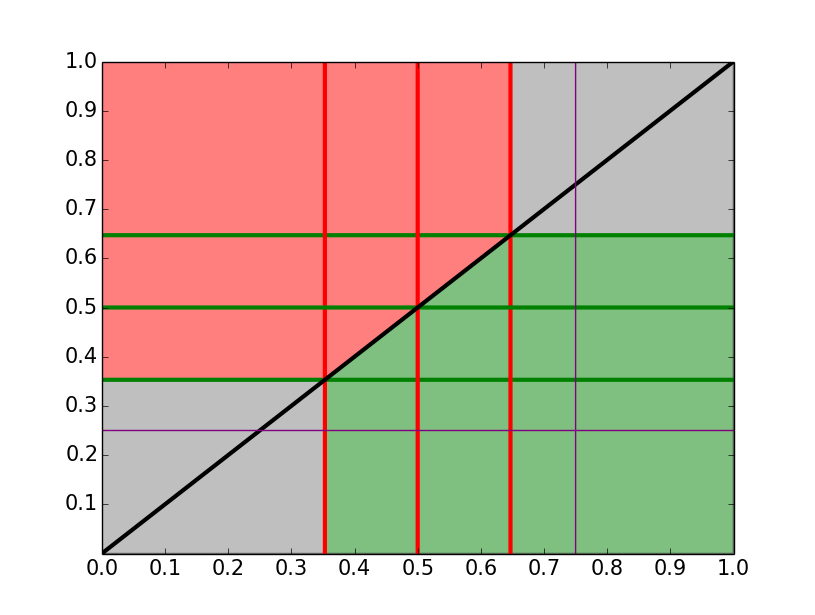}\\
$\rho=0.45$ & $\rho=0.5$
\end{tabular}
\caption{Domination thresholds for various values of $\rho$. In each case $\tau_r$ is plotted on the horizontal axis against $\tau_g$ on the vertical. The marked vertical lines represent (in increasing order) $\kappa^\rho_r$, $\frac{1}{2}$, $\mu^\rho_r$, and $1-\frac{1}{2}\rho$. The horizontal lines are $\frac{1}{2}\rho$, $\kappa^\rho_g$, $\frac{1}{2}$, and $\mu^\rho_g$. The black line is the threshold between red and green domination. Regions of red (respectively green) takeover are marked in red (green), and static regions are marked in grey.}
\label{fig:21}
\end{figure}

\section{Close Contests: between the thresholds \texorpdfstring{$\kappa^\rho_r$}{kr} and \texorpdfstring{$\frac{1}{2}$}{1/2}} \label{section:longsect}

We now wish to apply the results of the previous section to understand scenarios where $\kappa^\rho_r < \tau_r < \frac{1}{2}$. This automatically requires $\rho > \frac{1}{4}$. We shall make both these assumptions throughout this section, and also insist that $\tau_g < \frac{1}{2}$, since cases where $\tau_g > \frac{1}{2}$ will be subsumed into Theorem \ref{thmm:straddle}. However some of our lemmas will have weaker hypotheses which we shall state explicitly for later reuse.

To begin with, we shall deal with scenarios where $(\rho, \tau_g, \tau_r)$ is green dominating, establishing green takeover almost everywhere, thus proving Theorem \ref{thmm:take1}. An example of such a ring is illustrated in Figure \ref{fig:take3}. (We briefly explain how to interpret such a figure: the initial configuration is shown as the innermost ring, the initially unhappy elements are depicted outside that, and the final configuration is shown in the outermost ring. In between, the elements which change are shown in their new colour, with their distance from the centre proportional to their time of change.)

Our proof will be a modification of Section 4 of \cite{BEL1}, and indeed certain things will be simpler in the current case. In outline, the proof will proceed by letting $\eps>0$ while picking a node $u_0$ uniformly at random, and then seeking to establish that $u_0$ will be green in the final configuration with probability exceeding $1-\eps$ for all $n \gg w \gg 0$. As discussed in \ref{subsection:outline}, a key notion will be that of a green firewall, meaning a sequence of at least $w+1$ consecutive green nodes. Recall that any firewall is guaranteed to grow in both directions until it hits a stable interval of the opposite colour. Our plan is thus to establish that green firewalls are highly likely to form on both sides of $u_0$, with no stable red intervals or unhappy green nodes (which may spawn stable red intervals) in positions to block their paths from merging and encompassing $u_0$.

The first step, in Lemma \ref{lem:theli}, will be to identify a sequence of nodes $l_i$ stretching to the left of $u_0$ and $r_i$ to the right. Essentially $l_1$ will turn out to be the first node to the left of $u_0$ whose neighbourhood is such that it will be unhappy if red. Then $l_2$ will be the first such node to the left of $l_1 - (2w+1)$, and so on, with the $r_i$ emerging similarly to the right.

We shall then prove that each of the following statements holds with probability at least $1 - \eps'$ for arbitrary $\eps'>0$, conditional on the previous statements holding. (We withhold the technicalities for now, including suppressing several intermediate notions.)

\begin{itemize}[$\bullet$]
\item The $l_i$ and $r_i$ exist and satisfy various criteria including the absence of red stable intervals and unhappy green nodes between them (Lemma \ref{lem:theli}).
\item The distribution of green nodes in the vicinities of each $l_i$ and $r_i$ is \emph{smooth}, meaning that there are no awkward concentrations of red or green nodes nearby. (See Definition \ref{defin:smth2} and Corollary \ref{coro:theliaresmooth}). 
\item The vicinity of each $l_i$ and $r_i$ is likely to reach maturity without interference from beyond the $l_i$ or $r_i$, where red firewalls may be growing (Definition \ref{defin:greencomplete} and Lemma \ref{lem:thelicomplete}). Thus we can be confident that within our region of interest all the changes that occur will consist of red nodes turning green, rather than vice versa.
\item Smoothness will then allow us to argue that each $l_i$ and $r_i$ stands a reasonable chance of originating a green firewall (Corollary \ref{coro:thelispark} and Lemma \ref{lem:sparkgrows}).
\end{itemize}

Together, these will establish that green firewalls are highly likely to grow on both the left and right of $u_0$, and furthermore there will be no red stable intervals in positions to block these firewalls from eventually meeting and consuming $u_0$.

We now begin the proof by recalling some notation from \cite{BEL1} which will be useful when we wish to divide some interval $I$ into $k$ pieces. The following definition addresses this situation when the length of $I$ is not divisible by $k$:

\begin{figure}[!tbp]
\includegraphics[width=15cm, clip=true, trim= 2cm 3cm 2cm 3cm]{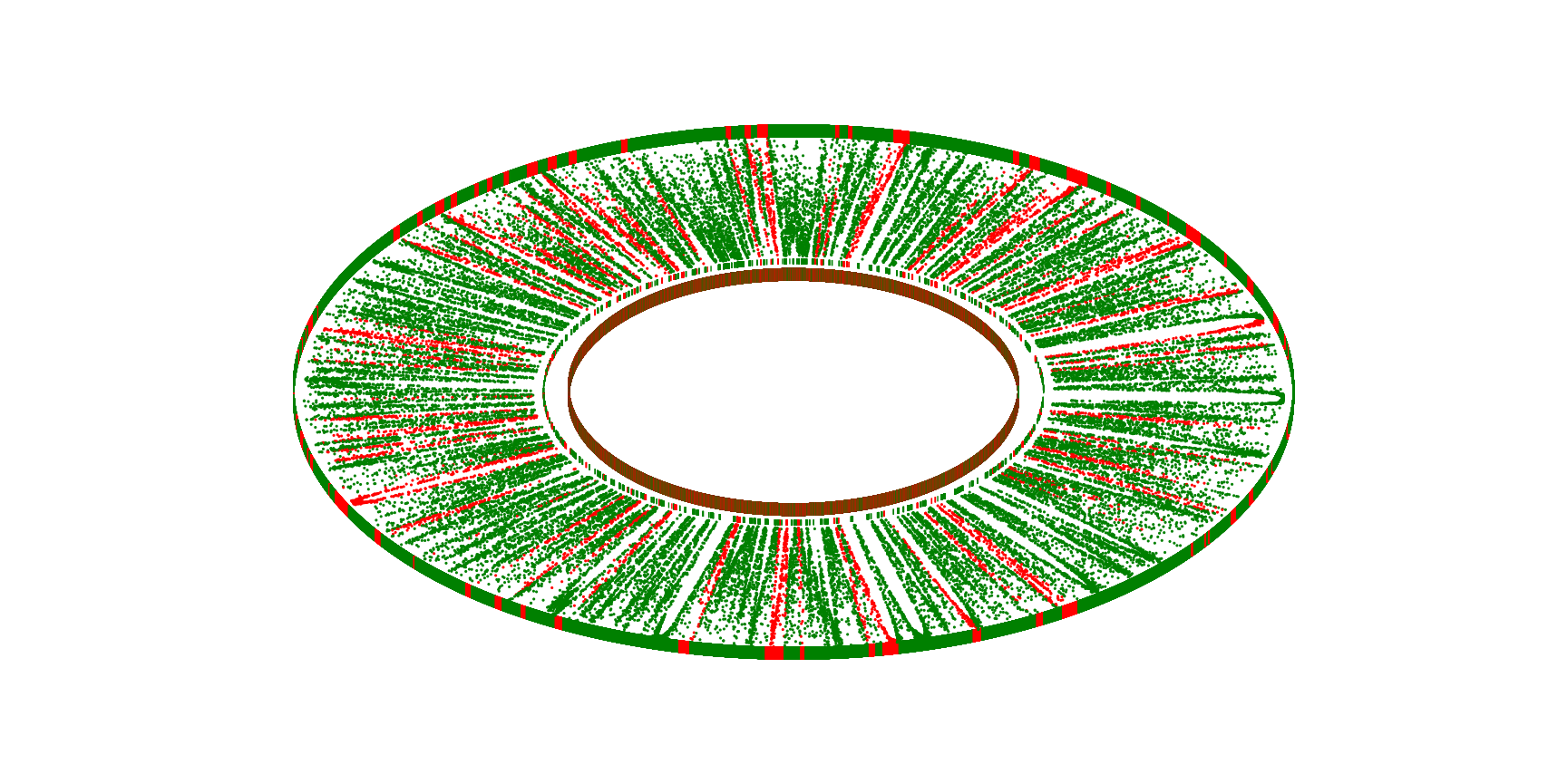} 
\caption{$\rho=0.48$, $\tau_g=0.38$, $\tau_r=0.46$, $w=50$, $n=100,000$, selective dynamic}
\label{fig:take3}
\end{figure}

\begin{defin} \label{defin:divide} Let $I=[a,b]$ and suppose $k\geq 1$.  We define the subintervals 
$I(1:k):= \left[ a, a+\left\lfloor \frac{b-a}{k} \right\rfloor \right]:=$ $\left[ I(1:k)_1,I(1:k)_2 \right]$ and $$I(j:k) := \left[ a+ \left\lfloor \frac{(j-1) (b-a)}{k} \right\rfloor +1, a+\left\lfloor \frac{j (b-a)}{k} \right\rfloor \right] := \left[ I(j:k)_1,I(j:k)_2 \right]$$ for $2 \leq j \leq k$.
\end{defin} 
 
It will sometimes be useful to count the subintervals from right to left: 
 
 \begin{defin} Let $I=[a,b]$ and suppose $k\geq 1$. For $1\leq j\leq k$ we define $I(j:k)^-= I(k-j+1:k)$,  $I(j:k)^-_1= I(k-j+1:k)_1$ and $I(j:k)^-_2=I(k-j+1:k)_2$. 
  \end{defin} 

We now begin to analyse the scenario by picking a node $u_0$ uniformly at random. The aim of the proof will be to show that for any $\eps$, in the finished ring $u_0$ is green with probability $>1-\eps$ for all $n \gg w \gg 0$. We shall deal separately with the cases $1- \rho \leq \tau_r$ and $1- \rho > \tau_r$. We postpone the former situation, and begin with the case where $\gamma:= 1- \rho - \tau_r >0$. Notice that under this assumption, green domination implies that $\rho > \tau_g$ via Corollary \ref{coro:easierdom}.

\begin{rem} \label{rem:Claim1} There are $0<\eta, \zeta <1$ so that $U_g < \eta^{w} U_r$ and $S_r < \zeta^{w} U_r$. \end{rem}

This follows from the assumptions of green domination and $\tau_r> \kappa^\rho_r$, by Propositions \ref{prop:kappasumup} and \ref{prop:outnumber}. 

\begin{rem} \label{rem:Claim2} Red nodes are unlikely to be unhappy in the initial configuration: $U_r \leq \exp{ \left(-2\gamma^2 (2w+1) \right)}$.\end{rem}

To justify this, let $u$ be a randomly selected red node. Then we think of $R_0(\N(u))$ as a sum of $2w+1$ independent random variables taking the value $1$ or $0$. Clearly its expected value is $(1- \rho)(2w+1)$. Then $U_r={\bf P} \big( R_0(\N(u))< \tau_r(2w+1) \big) = {\bf P} \big( (1- \rho)(2w+1) - R_0(\N(u)) > \gamma(2w+1) \big)$. The remark then follows by Hoeffding's inequality (Proposition \ref{prop:Hoeffding}).\\

\begin{defin} \label{defin:smth1} Let $u$ be a node, and let $\theta \in [0,1]$. We say $u$ has a \emph{local green density of $\theta$} or that $\GD_\theta(u)$ holds, if $\frac{G_0(\N(u))}{|\N(u))|}=\theta$.
\end{defin}
We shall be particularly interested in the case $\GD_{\theta^*}(u)$ where $\theta^*$ is as follows:

\begin{defin} \label{def:thetaask}
Let $\theta^*$ be minimal such that $\GD_{\theta^*} (v)$ implies that $v$ is unhappy if red. That is:
$$\theta^*:=\min \left\{ \frac{m}{2w+1}: \frac{m}{2w+1}>1-\tau_r \ \& \ m \in \NN \right\}.$$
\end{defin}
Clearly then, $\theta^* \to 1-\tau_r$ as $w \to \infty$, and it follows from our standing assumption $1-\tau_r - \rho >0$ that $\theta^*> \rho$ for all $w$.\\

\noindent \textbf{Definition of the $l_i$ and $r_i$.} \ \ We proceed recursively, with $l_0:=u_0$. Now define $l_{i+1}$ to be the first node to the left of $l_i - (2w+1)$ which is either unhappy, or satisfies $\GD_{\theta^*}$, or belongs to a red stable interval, so long as this node lies within $[u_0 - \frac{n}{4}]$. The $r_i$ are defined identically to the right. A little later we shall choose a specific value of $k_0$, not depending on $w$. For now we keep it flexible.

\begin{lem} \label{lem:theli}
For any $k_0>0$ and any $\eps'>0$, there exists $d>0$ such that for all large enough $w$ and $n$ large enough relative to $w$, the following hold with probability $>1 - \eps'$ 
\begin{enumerate}
\item $l_{k_0}, \ldots, l_{1}, r_{1}, \ldots, r_{k_0}$ are all defined. 
\item $l_{k_0}, \ldots, l_{1}, r_{1}, \ldots, r_{k_0}$ all satisfy $\GD_{\theta^*}$.
\item There are no unhappy green nodes in $[l_{k_0}, r_{k_0}]$.
\item No node in $[l_{k_0}, r_{k_0}]$ belongs to a stable red interval.
\item For $i \geq 2$, we have $|l_{i-1}-l_i|, |r_{i+1}-r_i|, |r_1 - l_1| \geq e^{dw}$.
\end{enumerate} \end{lem}

\begin{proof}

The first four points follow from Remark \ref{rem:Claim1} and Lemma \ref{lem:gen}. The fifth follows from Remark \ref{rem:Claim2} and the fact that for any interval $I$, we have ${\bf P}(\UR(I)>0) \leq \sum_{x \in I} \UR(x)$.
\end{proof}

With all this done, our goal is to show that there is a high chance that at least one of the $l_i$ and at least one of the $r_i$ will originate a green firewall. Since it is very likely that there are no unhappy green nodes or red stable intervals lying between these nodes, we can then be confident that the two firewalls will merge, thereby encompassing $u_0$. To this end we adapt the following notion from \cite{BEL1}:

\begin{defin} \label{defin:smth2} Suppose that $\GD_\theta(u)$ holds for some $\theta$. Let $\mathcal{L}= [u-(3w+1),u]$ and $\mathcal{R}=[u,u+(3w+1)]$. Suppose that $k>0$ is a multiple of $3$ and $\eps'>0$. For $j \leq k$ let $\mathcal{R}_j = \mathcal{R}(j:k)$ and $\mathcal{L}_j = \mathcal{L}(j:k)^-$. We additionally say that $\Sm(u)$ holds if:
\begin{itemize}[$\bullet$]
\item For $0 \leq j \leq \frac{k}{3}$, we have $\frac{|G_0(\mathcal{L}_j)|}{|\mathcal{L}_j|}, \frac{|G_0(\mathcal{R}_j)|}{|\mathcal{R}_j|} \in [\theta - \eps', \theta + \eps']$
\item For $\frac{k}{3} < j \leq k$, we have $\frac{|G_0(\mathcal{L}_j)|}{|\mathcal{L}_j|}, \frac{|G_0(\mathcal{R}_j)|}{|\mathcal{R}_j|} \in [\rho - \eps', \rho + \eps']$
\end{itemize}
\end{defin}

Thus $\Sm(u)$ asserts that the proportion of green nodes in $\N(v)$ smoothly moves from $\theta$ to $\rho$ as $v$ moves from $u$ to $u \pm(2w+1)$.

\begin{coro} \label{coro:smooth2} We make no assumption on $(\rho, \tau_g, \tau_r)$ or $\theta$. For all multiples of three $k>0$ and $\eps'>0$, and for all sufficiently large $w$, $${\bf P}\left( \Sm(u) | \GD_\theta(u) \right) >1-\eps'.$$ 
\end{coro} 
\begin{proof} Select $u$ uniformly at random from nodes such that $\GD_\theta(u)$ holds. We prove the first smoothness criterion first. The nodes in $\N(u)$ form a hypergeometric distribution. Since we consider fixed $k$ and $\eps'$ and take $w$ large, it suffices to prove the result for given $j$ with $1\leq j \leq \frac{k}{3}$. Here the result follows from an application of Chebyshev's inequality and standard results for the mean and variance of a hypergeometric distribution: 
\[ {\bf P}\left( \left| \frac{G(L_j)}{|L_j|}-\theta \right|>\eps' \right) <|L_j|^{-2} \eps^{\prime -2} \mbox{ Var}(G(L_j)) =O(1) |L_j|^{-1}. \]
Noting that $\big| |L_j| -(3w+1)/k \big| \leq 1$, the result follows.

Now let $u_{-1}=u-(2w+1)$ and $u_1=u+(2w+1)$. The fact that $\GD_\theta(u)$ holds has no impact on the distributions for $\N(u_{-1})$ and $\N(u_1)$, where both $E(G(\N(u_1)))=E(G(\N(u_{-1})))= \rho$. Thus the second smoothness criterion follows directly from the weak law of large numbers. 
\end{proof} 

Of course, the $l_i$ and $r_i$ are not selected randomly, so we may not simply apply Corollary \ref{coro:smooth2} to establish their smoothness. Nevertheless we shall be able to deduce it from the following result whose somewhat technical proof is contained in Appendix \ref{section:appC}:

\begin{restatable}{prop}{newbound}
\label{prop:newbound}
Fix a value of $\rho$ and a value $\theta \neq \rho$. For any node $u$ let $x_u$ be the first node to the left of $u$ such that $GD_\theta(x_u)$ holds.

Let $Q(u)$ be a property of nodes which depends only on the vicinity of $u$ in the initial configuration (which is to say it depends on $[u-C,u+C]$, for some $C$ independent of $n$).

Suppose there exists $p>0$ such that for all sufficiently large $w$ we have ${\bf P}(Q(u)|\GD_\theta(u)) \geq p$. Then there exists $p'>0$ such that for all $n \gg w \gg 0$ we have ${\bf P}(Q(x_u)) \geq p'$ for $u$ selected uniformly at random.

If additionally the hypothesis holds with $p = 1- \eps'$ for all $\eps'>0$, then we may likewise take $p' = 1-\eps_0$ for any $\eps_0>0$.
\end{restatable}

We shall appeal to Proposition \ref{prop:newbound} several times, starting with this:

\begin{coro} \label{coro:theliaresmooth}
Let $\eps'>0$ and $k>0$ be a multiple of $3$ and let $k_0>0$ be fixed. Then for all sufficiently large $n \gg w \gg 0$, with probability $>1-\eps'$ we have that $\Sm(l_i)$ and $\Sm(r_i)$ hold for all $i \leq k_0$.
\end{coro}

\begin{proof}
Corollary \ref{coro:smooth2} and Proposition \ref{prop:newbound} combine to tell us that for uniformly randomly selected $u$, we know that $\Sm(x_u)$ holds with probability $>1-\eps'$. Applying this to $u=u_0$ directly tells us that $\Sm(l_1)$ holds with probability $>1-\eps'$. Of course a symmetric argument applies to $r_1$. Proceeding inductively, suppose that we have established the result for $l_i$ and $r_i$. Then the sequence of nodes $[l_i-D,l_i - (w+1)]$, where $D$ is any quantity which is small compared to $n$, is independent of $[l_i-w,r_i+w]$. Hence we may apply the same argument again taking $u=l_i - (2w+1)$ to deduce that $\Sm(l_{i+1})$ holds with probability $>1-\eps'$. A symmetric argument works for $r_{i+1}$.
\end{proof}

Now, the following definition is valid under all three dynamics. Recall that a node is \emph{hopeful} if it is unhappy but a change of colour would cause it to become happy, and that we denote the number of hopeful, hopeful green, and hopeful red nodes in a set $A$ at time $t$ by $F_t(A)$, $\FG_t(A)$ and $\FR_t(A)$ respectively. (In our current scenario where $\tau_h, \tau_r < \frac{1}{2}$ hopefulness is automatic for unhappy nodes. However we shall reuse this notion later in another context.) 

\begin{defin} \label{defin:greencomplete}
We say that a node $u$ \emph{green completes at stage $s$} if \begin{itemize}[$\bullet$]
\item $F_s(\N(u))=0$, but $F_t(\N(u))>0$ for all $t<s$
\item $\FG_t(\N(u))=0$ for all $t\leq s$.
\end{itemize}
\end{defin}

For the nodes we consider, it will typically be the case that $\FR_0(\N(u))>0$, otherwise we may have the trivial situation of green completion at stage 0. 

If $u$ green completes it follows that $G_t(\N(u))$ is a monotonic increasing function for $t\leq s$. We shall apply the following Lemma to $l_i$ and $r_i$ where $i<k_0$, but phrase it more generally for reuse later. Again the most useful case will be when $u$ itself is a hopeful red node.

\begin{lem} \label{lem:thelicomplete}
Suppose that $u$ is a node and that $v$ and $v'$ are its nearest hopeful green nodes to the left and right respectively in the initial configuration. Assume that there exists $d>0$ so that $ |u-v|, |v'-u| >e^{wd}$ for all $w \gg 0$. Then the following holds independently of all other facts about the ring's initial configuration: in the selective model with any $\tau_g, \tau_r$ or in the incremental model with $\tau_g, \tau_r < \frac{1}{2}$, for any $\eps'>0$ we have that $u$ green completes with probability $>1- \eps'$ for all large enough $w$. In the synchronous model for $\tau_g, \tau_r < \frac{1}{2}$ we have instead that $u$ green completes with probability $1$. 
\end{lem}

\begin{proof}

We work with the selective/incremental model first. Let $I_1=[v,u]$ and $I_2=[u,v']$. Let $k$ be the greatest such that, when $1\leq j \leq k$, $I_1(j:k)$ and $I_2(j:k)$ are of length $\geq w+1$. For $1\leq j \leq \lfloor k/2 \rfloor $ define: 
\[ J_j:= I_1(j:k) \cup I_2(j:k)^-. \] 

\noindent  For $1\leq j \leq \lfloor k/2 \rfloor $, let $P_j$ be the event that $R(J_j)$ increases by $1$, and note that $P_{j+1}$ cannot occur until $P_{j}$ has occurred.  Now the basic idea is that if green completion fails to occur,  then the sequence of events $P_2,...,P_{\lfloor k/2 \rfloor} $ must occur before any stage when $F(\N(u))=0$. 

We label certain stages as being a `step towards green completion', and certain others as being a `step towards failure of green completion'.\\

\noindent \textbf{Steps towards green completion}.  If $F(\N(u))>0$ we label any stage at which a node in $\N(u)$ changes from red to green as a \emph{step towards green completion}. Once $F(\N(u))=0$, we consider every step to be a step towards green completion.\\

\noindent \textbf{Steps towards failure of green completion}. If  $1\leq j<\lfloor k/2 \rfloor $ is the greatest  such that $P_j$ has occurred prior to stage $s$ or no $P_j$ has occurred and $j=1$,  and if $P_{j+1}$ occurs at stage $s$, then we label $s$  a  \emph{step towards failure of green completion}.\\

We now adopt a modified stage count which counts only steps towards either green completion or its failure. (Once $F(\N(u))=0$, every stage is counted.) Now at any stage $s$ at which some $P_j$ for $ j\leq \lfloor k/2 \rfloor$ is yet to occur, and at which $F(\N(u))>0$, the probability of $s$ being  a step towards failure of green completion is at most $2(w+2)$ times the probability of it being a step towards green completion (since there are at most $2(w+2)$ times as many nodes which, if chosen to change, will cause a step towards failure of green completion, as those which will cause a step towards green completion). Choosing $0<d'<d$ we get that for all sufficiently large $w$, $\lfloor k/2 \rfloor  >e^{d'w}$. We may therefore consider the first $e^{d'w}$ many stages which are  steps either towards green completion or failure of completion and, for large $w$, consider the  probability that at most $2w+1$ of these are steps towards green completion. By the law of large numbers, this probability tends to 0 as $w\rightarrow \infty$.  What is more, by assumption on $u$, $2w+1$ many such steps more than suffice for its green completion.

For the synchronous model, we simply have to note that since $e^{d'w} \gg 2w+1$ for large enough $w$, the influence of $v$ or $v'$ cannot be felt in $\N(u)$ within the first $2w+1$ time-steps. Thus green completion is inevitable.
\end{proof}

We shall say that a node $u$ \emph{originates a green firewall} if $u$ green completes, at which time $\N(u)$ contains a run of $w+1$ consecutive green nodes. The final step in the proof is to show that each $l_i$ and $r_i$ originates a green firewall with reasonable probability. We have to do a little more work to establish this, and again we express things more generally. First we establish something weaker, that a firewall gets started in the following sense:

\begin{defin} \label{defin:spark}
With no assumptions on our scenario, let $\alpha \in (0,1)$. We say that a node $u$ \emph{$\alpha$-sparks} if $u$ green completes, and at the moment of completion the interval $K_\alpha:=[u- \lfloor \alpha \cdot w\rfloor, u+ \lfloor \alpha \cdot w \rfloor ]$ is completely green.
\end{defin}

Our strategy will be to argue, under suitable conditions on $\alpha$, that each $l_i$ and $r_i$ has a reasonable chance of $\alpha$-sparking, and then to establish that such a spark will guarantee the emergence of a green firewall. First, however, we need to consider a technical matter which will become important:

\begin{lem} \label{lem:firestarter2}
For any $\theta,\rho \in (0,1)$, define $Z(\theta,\rho):=1+ \theta^3 - 3 \theta^2 + 3 \theta^2 \rho - 2 \theta^3 \rho$. Then \begin{enumerate}[(i)]
\item $\frac{\partial Z}{ \partial \theta}<0$
\item If $\theta < \frac{1}{2}(1+ \rho)$ then $Z(\theta, \rho)>0$.
\end{enumerate}
\end{lem}

\begin{proof}
To start with, $\frac{\partial Z}{ \partial \theta} = 3 \theta^2 - 6 \theta + 6 \theta \rho - 6 \theta^2 \rho = 3 \theta(\theta - 2 + \rho(2- 2\theta)) <  3 \theta(\theta -2 +2 - 2\theta) = - 6 \theta^2<0$, establishing the first statement of the lemma for any $\theta \in (0,1)$.

For the second, then, by assumption on $\theta$, we only need to check that $$Z \left( \frac{1}{2}(1 + \rho), \rho \right)=\frac{3}{8} - \frac{5}{8} \rho + \frac{3}{8} \rho^2 + \frac{1}{8} \rho^3 - \frac{1}{4} \rho^4>0$$ which the reader may verify does indeed hold for $\rho \in (0,1)$. 
\end{proof}

\begin{rem} \label{rem:thetastarisgood}
Lemma \ref{lem:firestarter2} establishes in particular that $Z(\theta^*, \rho)>0$ for all large enough $w$. 
\end{rem}

To see why the remark holds, recall that the hypotheses of Theorem \ref{thmm:take1} include the assumption that $\tau_r > \kappa^\rho_r$ and we saw in Proposition \ref{prop:kappasumup}, that $\kappa^\rho_r > \frac{1}{2}(1-\rho)$. Since $\theta^* \to 1- \tau_r$ as $w \to \infty$, it follows that for large enough $w$, we shall have $\theta^* < \frac{1}{2}(1+ \rho)$.

The following will go most of the way to establishing that the $l_i$ and $r_i$ have a good chance of sparking:

\begin{lem} \label{lem:firestarter} 
Let $Z$ be as defined in \ref{lem:firestarter2}, and suppose $\theta$ is such that $Z(\theta,\rho)>0$. Suppose that $u$ is uniformly randomly selected from nodes satisfying $\GD_\theta(u)$.  

For any $\alpha \in \left( 0, \frac{\theta}{2} \right)$, define $\theta_\alpha:=\frac{\theta-2 \alpha}{1 - 2 \alpha}$. Now fix $\alpha$ small enough that also $Z(\theta_\alpha, \rho)>0$.

Then there exists $\delta>0$ (depending on the scenario, $\theta$, and $\alpha$ but not on $w$) such that if $u$ green completes, then it $\alpha$-sparks with probability $>\delta$ for all $w \gg 0$.

This holds in both the selective/incremental and synchronous models.
\end{lem} 

\begin{proof}

We adopt a novel way of counting stages, defining \emph{spark stage} $s$ to be the first time (if it exists) that $[u -s, u+s]$ is entirely green. We shall consider spark stages up to $s = \lfloor \alpha \cdot w \rfloor$. Of course, if this final stage is reached, then a green spark has occurred. Now suppose, inductively, that spark stage $s$ has been reached. We shall estimate the probability of reaching spark stage $s+1$. To do this, we compute a lower bound for $G(\N(u+s+1))$ at stage $s$ recursively: define $M_0:=G_0(\N(u)) = \theta(2w+1)$. Define $$M_{s+1}:=M_{s} + R(u-s) + R(u+s) - G(u +s - w)+G(u +s +w +1).$$
(By $G(a)$ for an individual node $a$ we shall mean $G_0(a)$ throughout this proof and similarly for $R$). Then $G(\N(u+s+1)) \geq M_s$ at spark stage $s$.

Now, we have made no assumption on the distribution of nodes outside $\N(u)$. Therefore it is legitimate to treat nodes in $[u+w+1, u+3w+1]$ as independent identical random variables with a probability $\rho$ of being green. Let us briefly make the additional temporary assumption that in the initial configuration, nodes in $[u - w, u+w]$ are independent random variables with a probability of $\theta$ of being green.\\

\noindent \textbf{Claim \ \ } ${\bf P}(M_{s+1}>M_{s})>{\bf P}(M_{s+1}<M_{s})$\\

\noindent \textbf{Proof of claim \ \ } At each stage we have 
\begin{eqnarray*}{\bf P}(M_{s+1} &=& M_{s} -1 ) \\ &=& P \Big( G(u-s)=G(u+s)=G( u+s-w)=  R(u + s +w +1)=1 \Big) \\ &= & \theta^3 (1-\rho) \end{eqnarray*}

Similarly ${\bf P}(M_{s+1} = M_{s}) =  \theta^3 \rho + 3 \theta^2(1-\theta)(1-\rho)$. Thus ${\bf P}(M_{s+1}>M_{s}) = 1 - \theta^3 -3 \theta^2(1-\theta)(1-\rho)$ and ${\bf P}(M_{s+1}>M_{s})-{\bf P}(M_{s+1}<M_{s}) = Z(\theta,\rho)$. Hence the claim follows from our assumption that $Z(\theta,\rho)>0$. \textbf{QED Claim}\\

We now consider $M_s$ as a biased random walk, omitting all steps where $M_{s+1}=M_s$. The claim establishes that the walk is more likely to increase than decrease at each spark stage. Call the probability that it increases $p:=\frac{1}{2}(1+Z(\theta,\rho))> \frac{1}{2}$. It follows then that, the probability that it ever drops below $M_0$ is $\frac{1- p}{p}$ by a standard result on biased random walks. 

Everything we have stated here applies equally to the mirror-image process $$M'_{s+1}:=M'_{s} + R(u+s) + R(u-s) - G(u -s + w)+G(u -s -w -1).$$
Moreover $G(\N(u-s-1)) \geq M'_s$ at spark stage $s$. If $M_s \geq M_0$ and $M'_s \geq M'_0$ for all $s \leq \lfloor \alpha \cdot w \rfloor$ then $u-s-1$ and $u+s+1$ are guaranteed to be unhappy if red at stage $s$, meaning that $K_\alpha$ is certain to be green in the event of green completion.

All that remains is to drop the false assumption of independence. Taking the above two processes $M$ and $M'$ together, at each spark stage $s$, we \emph{see} four new nodes in $\N(u)$ namely $u-s$, $u+s$, $u +s - w$, and $u -s + w$. Thus the number of unseen nodes is $(2w+1) - 4s$ of which at least $\theta(2w+1) - 4s$ are green. Thus, the real probability of picking a green node from the remaining unseen nodes is $\frac{\theta(2w+1) - 4s}{(2w+1) - 4s}> \theta_\alpha$ for all $s \leq \lfloor \alpha \cdot w\rfloor$.

Hence, for all large enough $w$, working now with the true probabilities, we find that for all large enough $w$ and all $s \leq \lfloor \alpha \cdot w\rfloor$, we have ${\bf P}(M_{s+1}>M_{s})-{\bf P}(M_{s+1}<M_{s}) > Z(\theta_\alpha,\rho)>0$ by assumption on $\alpha$.

Thus setting $p':=\frac{1}{2}(1+Z(\theta_\alpha,\rho))$ we find $p'>\frac{1}{2}$ and dropping the assumption of independence, the actual probability that $M_{s+1}>M_{s}$ will exceed $p'$ for all large enough $w$ and all spark stages $s \leq \lfloor \alpha \cdot w \rfloor$, no matter what has occurred at previous stages. Thus the probability that $M_s$ never drops below $M_0$, and thus that an $\alpha$-spark occurs, will be at least $\delta:=1-\frac{1-p'}{p'} = \frac{2 Z(\theta_\alpha,\rho)}{1+ Z(\theta_\alpha,\rho)}$.
\end{proof}

\begin{coro} \label{coro:thelispark}
Preserving the notation from \ref{lem:firestarter}, suppose that $\theta \neq \rho$ and $\alpha$ is such that $Z(\theta_\alpha, \rho)>0$.

Then there exists $\delta'>0$ (depending on the scenario, $\theta$, and $\alpha$ but not on $w$) such that for each $l_i$, if $l_i$ green completes, then it $\alpha$-sparks with probability $>\delta'$ for all $w \gg 0$.

This holds in both the selective/incremental and synchronous models.
\end{coro}

\begin{proof}
This is simply a matter of applying Lemma \ref{lem:firestarter} and Proposition \ref{prop:newbound}, bearing in mind that we can apply the conclusion to $l_i$ exactly as in the proof of Corollary \ref{coro:theliaresmooth}.
\end{proof}

Corollary \ref{coro:thelispark} at last allows to choose $k_0$ and $\eps'$ satisfying $\eps' + (1-\delta')^{k_0}< \eps$, applying Lemma \ref{lem:thelicomplete} for our chosen value of $\eps'$.

Notice that $k_0$ is independent of $w$ as promised and is selected to ensure that the probability that no $l_i$ or no $r_i$ $\alpha$-sparks is $<\eps$, for all large enough $w$.

The final step is now to appeal to smoothness to show that a green spark will lead to a green firewall. We shall apply the following to those $l_i$ which $\alpha$-spark:

\begin{lem} \label{lem:sparkgrows}
Suppose that $\theta>1-\tau_r$ and $\theta> \tau_g$ and $\theta < \frac{1}{2}(1+ \rho)$, and that $\alpha$ is small enough that also $Z(\theta_\alpha, \rho)>0$ as in Lemma \ref{lem:firestarter}.

Now fix integers $r$ large enough that $r(1+ \rho - 2 \theta) \gg 1$ and $k$ large enough that $\frac{r}{k}<\alpha$.

Suppose now that $u$ is a node satisfying $\GD_\theta(u)$ and $\mathtt{Smooth}_{k,\eps'}(u)$ and suppose that $u$ $\alpha$-sparks. Then $\N(u)$ contains a green firewall at the moment of green completion.
\end{lem}

\begin{proof}
We proceed inductively, assuming that $\bigcup_{j=1}^s \mathcal{R}_j$ has become fully green, where $\mathcal{R}_j$ is as defined in Definition \ref{defin:smth2}) and $1 \leq s < \frac{k}{3}$. We shall show that $\mathcal{R}_{s+1}$ will also become green. (The base case $s=r$ holds by assumption since $u$ $\alpha$-sparks. The final case $s=\frac{k}{3}$ will amount to $[u,u+w]$ forming a green firewall.)

Let $v \in \mathcal{R}_{s+1}$. First we shall bound $G(\N(v))$ below. Let $v'$ be the rightmost node in $\mathcal{R}_s$, and look at $\N(v')$. Firstly, of course, $\bigcup_{j=1}^s \mathcal{R}_j \subseteq \N(v')$, which makes a contribution of $\geq s \cdot \frac{3w+1}{k}$ many green nodes. To the left of that, we have $\bigcup_{j=1}^{\frac{k}{3} - s} \mathcal{L}_j$, which by smoothness contribute at least $\left(\frac{k}{3} - s \right) \cdot \frac{3w+1}{k} \cdot (\theta - \eps')$ many green nodes.

To the right, however, we have $\bigcup_{j=s+1}^{\frac{k}{3}} \mathcal{R}_j$ contributing at least $\left(\frac{k}{3} - s \right) \cdot \frac{3w+1}{k} \cdot (\theta - \eps')$, followed by $\bigcup_{j=\frac{k}{3}+1}^{\frac{k}{3}+s} \mathcal{R}_j$ which contributes at least $s \cdot \frac{3w+1}{k} \cdot (\rho - \eps')$.

Adding all these together, we get a lower bound for $G_0(v')$ which may then adapt to a lower bound for $G_0(v)$, by observing that $|v-v'| \leq \frac{3w+1}{k}$. Hence we find our bound:

$$G_0(v) \geq \frac{3w+1}{k} \cdot \left(s + 2\left(\frac{k}{3} - s \right)(\theta - \eps') +s \cdot (\rho - \eps') - 1\right).$$

That is to say
$$G_0(v) \geq 2w(\theta - \eps') + \frac{3w+1}{k} \cdot \left(-1 + s \cdot \left(1 - 2 \theta  + \rho + \eps' \right)\right)$$

Now $s \geq r$, and by assumption on $r$ we have $r \cdot \left(1 - 2 \theta  + \rho + \eps' \right) \gg 1$, from which it follows that $G_0(v) \geq (2w+1)\theta$, meaning that if $v$ is red then it must be unhappy.
\end{proof}

This concludes the proof of Theorem \ref{thmm:take1} in the case $(1- \rho) > \tau_r$.

\subsection*{\emph{Proof of Theorem \ref{thmm:take1} when $(1- \rho) \leq \tau_r$}}
 
\ \  Here we find that the probability that a randomly chosen red node is unhappy  $U_r \to 1$ (if $(1- \rho) < \tau_r$) or $U_r \to \frac{1}{2}$ (if $(1- \rho) = \tau_r$) as $n,w \to \infty$. In either case the foregoing proof goes through with only very minor modifications.

Suppose first that $(1- \rho) < \tau_r$. Again we define $l_0:=u_0$ and $l_{i+1}$ to be the first node to the left of $l_i - (2w+1)$ which is either unhappy, or satisfies $\GD_{\theta^*}$, or belongs to a red stable interval, so long as this node lies within $[u_0 - \frac{n}{4}]$. The $r_i$ are defined identically to the right. Let $v$ (respectively $v'$) be the nearest unhappy green node to the left (right) of $u_0$. This time we find that the $l_i$ and $r_i$ are much closer together. This is of no concern so long as they are far enough from $v$ and $v'$, and with the proof as before we obtain the following variant of Lemma \ref{lem:theli}:

\begin{lem} \label{lem:theliyetagain}
For any $k_0>0$ and $\eps'>0$ there exists $d>0$ such that for all large enough $w$ the following hold with probability $>1 - \eps'$ 
\begin{enumerate}
\item $l_{k_0}, \ldots, l_{1}, r_{1}, \ldots, r_{k_0}$ are all defined. 
\item $l_{k_0}, \ldots, l_{1}, r_{1}, \ldots, r_{k_0}$ all satisfy $\GD_{\theta^*}$.
\item No node in $[l_{k_0}, r_{k_0}]$ lies in a stable red intervals.
\item For $i \geq 1$, we have $|l_i-v|, |r_{i}-v'| \geq e^{dw}$.
\end{enumerate} \end{lem}

Now, Corollary \ref{coro:theliaresmooth} applies again (technically of course the current $l_i$ are different from those involved in the original statement, however the proof goes through without alteration). This gives us that for each $\eps'>0$ and $k_0 \in \NN$, $\Sm(l_i)$ and $\Sm(r_i)$ holds for all $i \leq k_0$ with probability $>1 - \eps'$ for all $w \gg 0$.

Then we can apply Lemma \ref{lem:thelicomplete}, giving us that the $l_i$ and $r_i$ all green complete with probability $>1- \eps'$ for all large enough $w$. Finally Corollary \ref{coro:thelispark} also applies to the $l_i$ and $r_i$ and once again allows us to choose $k_0$, such that the probability that no $l_i$ or no $r_i$ initiate a green firewall is $<\eps$.\\

In the case $\tau_r=1- \rho$ there is a minor complication in that we cannot apply results such as Corollary \ref{coro:theliaresmooth}, owing to the hypothesis in various lemmas that $\theta \neq \rho$. However we may get around this quite simply, by letting $\tau_r'<\tau_r$ be such that $\tau_r'>\kappa^\rho_r$ and $(\rho, \tau_g,\tau_r')$ is green dominating. Then $\tau_r'<1- \rho$ and we proceed through this section's main argument working with $\tau_r'$ in place of $\tau_r$ throughout, beginning with Definition \ref{def:thetaask}. 

Clearly any red node which is $\tau_r'$-unhappy is automatically $\tau_r$-unhappy, and thus we may deduce the existence of green firewalls on either side of $u_0$ as before. It may be that other green regions also grow, undetected by our analysis, around red nodes which are initially $\tau_r$-unhappy but not $\tau_r'$-unhappy, but this is unproblematic. This completes the proof of Theorem \ref{thmm:take1}.

\subsection*{\emph{Proof of Theorem \ref{thmm:static2}}}

\ \ We now turn our attention to scenarios where $\tau_r < \kappa^\rho_r$, while $\frac{1}{2}>\tau_g > \kappa^\rho_g$ and $(\rho, \tau_g, \tau_r)$ is green dominating, which we shall prove to be static almost everywhere under the additional assumption that $\tau_r > \frac{1}{2}(1-\rho)$. Interchanging the roles of red and green will establish Theorem \ref{thmm:static2}. (We have made this alteration to be able to apply our previous lemmas verbatim.) Figure \ref{fig:stat3} is instructive of what to expect (and also provides an example where large values of $n$ and $w$ are required for the essential staticity of the situation to be revealed).

\begin{figure}[!ht]
\centering
\def\svgwidth{15cm}
\includegraphics[width=15cm, clip=true, trim= 2cm 3cm 2cm 3cm]{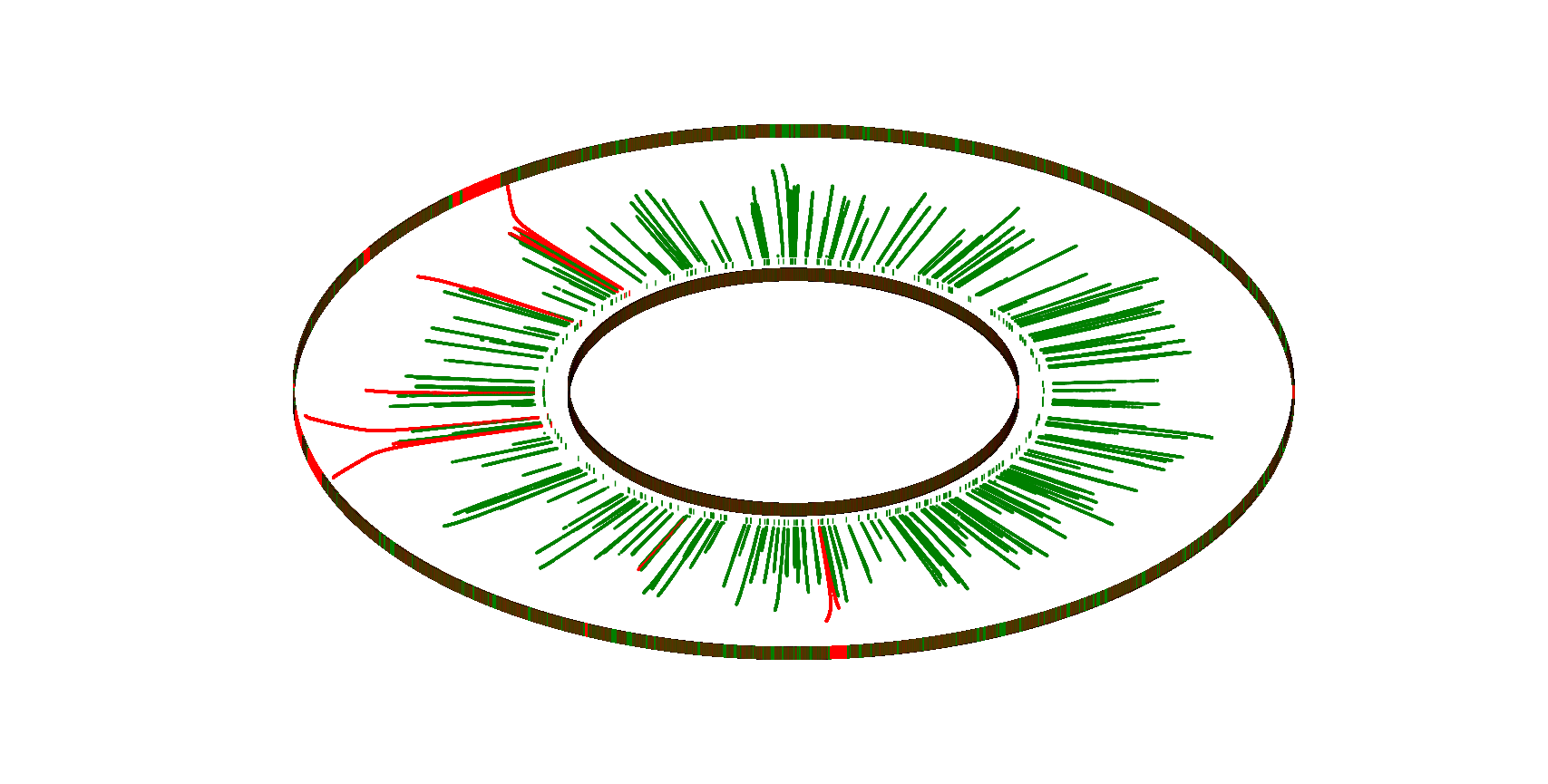}
\caption{$\rho=0.6$, $\tau_g=0.43$, $\tau_r=0.27$, $w=70$, $n=5,000,000$, selective dynamic}
\label{fig:stat3}
\end{figure}

We begin by observing that the following follows from Propositions \ref{prop:outnumber} and \ref{prop:kappasumup}:

\begin{rem} \label{rem:dagger} There exist $0<\eta, \zeta, \xi<1$ so that for all $w \gg 0$ we have $S_g< \xi^w U_g$, while $U_g < \eta^w U_r$, and in turn $U_r < \zeta^w S_r$. \end{rem}

Again begin by picking a node $u_0$ at random. Now we outline the general intuition. Let $y$ (respectively $y'$) be the nearest node to the left (right) of $u_0$ which are either unhappy or belong to a stable interval. By Remark \ref{rem:dagger} above and Lemma \ref{lem:gen}, for large enough $w$, both $y$ and $y'$ are highly likely to belong to red stable intervals. This guarantees that $u_0$, if red, can never turn green. 

So suppose, for the remainder of this section, that $u_0$ is green. In the initial configuration Remark \ref{rem:dagger} tells us that unhappy green nodes are hugely more frequent than stable green intervals. Thus on first sight, there appears to be a danger that $u_0$ will be engulfed in a red firewall. However, unhappy red nodes are commoner still, and we shall show that these are highly likely to give rise to stable green intervals (or short firewalls) in positions protecting $u_0$.

The argument runs much as previously. First we define $l_0:=u_0$, and define $l_{i+1}$ to be the first node to the left of $l_i - (2w+1)$ which is either unhappy or satisfies $\GD_{\theta^*}$ (defined as in Definitions \ref{defin:smth1} and \ref{def:thetaask}), so long as this node lies within $[u_0 - \frac{n}{4}]$. The $r_i$ are defined identically to the right. As before Remark \ref{rem:dagger} together with Lemma \ref{lem:gen} give us the following.

\begin{lem} \label{lem:theliagain}
For any $k_0>0$ and $\eps'>0$ there exists $d>0$ such that for all large enough $w$ the following hold with probability $>1 - \eps'$ 
\begin{enumerate}
\item $l_{k_0}, \ldots, l_{1}, r_{1}, \ldots, r_{k_0}$ are all defined. 
\item $l_{k_0}, \ldots, l_{1}, r_{1}, \ldots, r_{k_0}$ all satisfy $\GD_{\theta^*}$.
\item There are no unhappy green nodes in $[l_{k_0}, r_{k_0}]$.
\item For $i \geq 2$, we have $|l_{i+1}-l_i|, |r_{i+1}-r_i|, |r_1 - l_1| \geq e^{dw}$.
\end{enumerate} \end{lem}

Notice that in this case, Remark \ref{rem:dagger} alone provides enough information to derive the fourth point.

Now, Corollary \ref{coro:theliaresmooth} applies again. (Again our current $l_i$ are technically different from those mentioned there but the proof remains valid.) Therefore, for each $\eps'>0$ and $k_0 \in \NN$, $\Sm(r_i)$ holds for all $i \leq k_0$ with probability $>1 - \eps'$ for all $w \gg 0$.

Then we can apply Lemma \ref{lem:thelicomplete}, giving us that the $l_i$ and $r_i$ all green complete with probability $>1- \eps'$ for all large enough $w$.

Since we have stipulated that $\tau_g> \frac{1}{2}(1-\rho)$, we find that with $\theta^*$ (defined as in Definition \ref{def:thetaask}) satisfies $\theta^*<\frac{1}{2}(1+\rho)$ for all large enough $w$. Thus by Remark \ref{rem:thetastarisgood} $Z(\theta^*, \rho)>0$ for all large enough $w$. Hence we may pick $\alpha$ as before and apply Corollary \ref{coro:thelispark} to our current $l_i$ and $r_i$ to establish that each has a chance of at least $\delta'>0$ of $\alpha$-sparking. Furthermore, if some $l_i$ or $r_i$ does spark then we may apply Lemma \ref{lem:sparkgrows} to establish that a green firewall will automatically ensue (again we are relying on our hypothesis that $\tau_g> \frac{1}{2}(1-\rho)$).

This once again allows us to choose $k_0$, such that the probability that no $l_i$ or no $r_i$ initiate a green firewall is $<\eps$. Now let $l$ and $r$ be the $l_i$ and $r_i$ nearest $u_0$ which do initiate green firewalls. It is certain that these green firewalls will spread towards $u_0$ until they hit stable red intervals. Suppose this has happened by stage $s$. At this stage, the vicinity of $u_0$ has transformed to resemble the case of Theorem \ref{thmm:static}, meaning that looking away from $u_0$ we encounter red stable intervals before unhappy red nodes and green stable intervals unhappy green nodes. Thus with probability $>1- \eps$, the node $u_0$ will not change colour.

\subsection*{\emph{Discussion of Question \ref{questionn:static2?}}}

It is clear that the bulk of the foregoing argument does not extend to cases where $\tau_r \leq \frac{1}{2}(1-\rho)$ (this corresponds to Question \ref{questionn:static2?} with the roles of red and green interchanged). For values of $\rho$ approximately $0.6384$ and $\frac{3}{4}$, such scenarios are possible.
For example, the scenario $(\rho,\tau_g,\tau_r)= (0.7,0.49,0.13)$ satisfies all the hypotheses of this section: $\tau_r < \kappa_r \approx 0.21$ and $\kappa^{0.7}_g \approx 0.48 < \tau_g < 0.5$ and the scenario is green dominating. However, also $\tau_r=0.13<\frac{1}{2}(1-\rho)=0.15$.

Exchanging the roles of red and green gives the scenario $(\rho,\tau_g,\tau_r)= (0.3,0.13,0.49)$ which is located in the lower purple region of Figure \ref{fig:map30}.

In this scenario, at least we have $Z(\theta^*,\rho)>0$ however, for large enough $w$. But there are even scenarios where this weaker condition fails. For example the scenario comprising $(\rho,\tau_g,\tau_r)=(0.74,0.498,0.07)$ which again satisfies all the relevant constraints: $\tau_r < \kappa^{0.74}_r \approx 0.186$ and $\kappa^{0.74}_g \approx 0.497 < \tau_g < 0.5$ and the scenario is green dominating.

Clearly also $\tau_r = 0.07 <  \frac{1}{2}(1-\rho) = 0.13$. Moreover, as $w \to \infty$ we have $\theta^* \to 1-\tau_r=0.93$, and we find $Z(\theta^*,\rho) \to -0.06$ approximately.

Although our previous arguments fail in such cases, needless to say, it does not follow that no $l_i$ evolves into a green stable interval protecting $u_0$. It seems that a more detailed analysis is needed to resolve such situations.

\section{Total Takeover: the case \texorpdfstring{$\tau_r> \frac{1}{2} > \tau_g$}{tr>0.5>tg}} \label{section:straddle}

We begin with an observation whose proof we leave to the reader:

\begin{lem} \label{lem:happadj} The following are equivalent:
\begin{enumerate}[(i)]
\item $\tau_g+\tau_r \leq 1$
\item For all $n \gg w \gg 0$, there can exist happy adjacent nodes of opposite colours.
\item For all $n \gg w \gg 0$, all unhappy nodes are hopeful.
\end{enumerate}
\end{lem}

It follows that when $\tau_g+\tau_r > 1$, the selective and incremental dynamics have the possibility to differ, as indeed will be the case. Having said this, in this section we may establish Theorem \ref{thmm:straddle} which applies under every dynamic, and states that if $\tau_g < \frac{1}{2} < \tau_r$ then green will take over totally, as illustrated in Figure \ref{fig:take2}. This will also establish that, for all large enough $n$, the initial configuration will very likely be such that the process is guaranteed to finish, under both the incremental and synchronous dynamics.

\begin{figure}[!ht]
\centering
\def\svgwidth{15cm}
\includegraphics[width=15cm, clip=true, trim= 2cm 3cm 2cm 3cm]{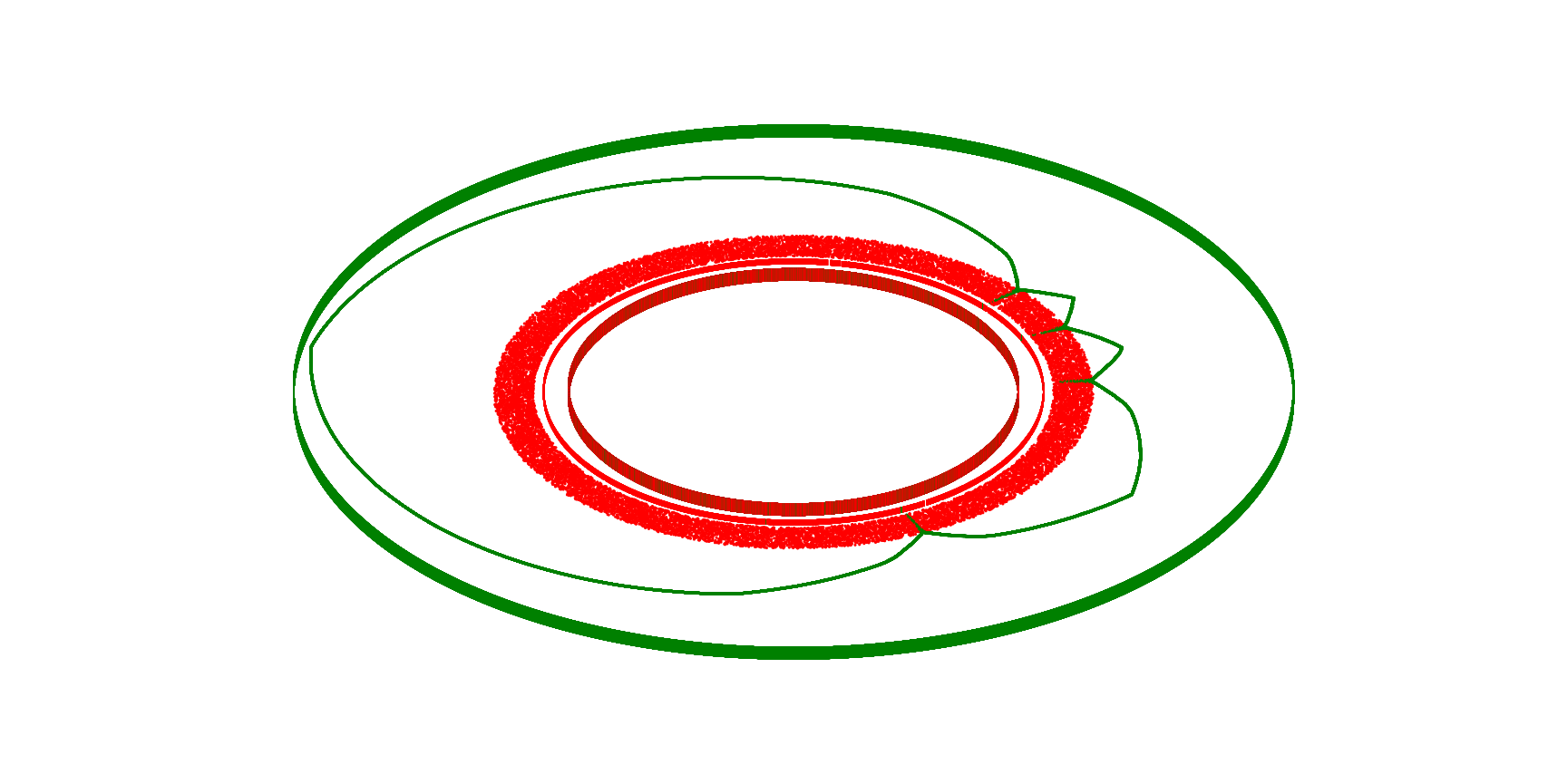}
\caption{$\rho=0.2$, $\tau_g=0.25$, $\tau_r=0.65$, $w=40$, $n=100,000$, selective dynamic}
\label{fig:take2}
\end{figure}

\begin{proof}[Proof of Theorem \ref{thmm:straddle}]

We work with $w$ large enough that $\frac{w+1}{2w+1} < \tau_r$. Suppose that at some stage $[a,b]$ is a green firewall of length at least $w$, and that $a-1$ and $b+1$ are red. Then every node in $[a,b]$ is happy and stably so. On the other hand, $a-1$ and $b+1$ are unhappy (and indeed hopeful) and will remain so as long as they are red. Hence the firewall cannot shrink and will eventually grow to encompass $a-1$ and $b+1$.

It follows that as soon as we have a green firewall of length $\geq w$, total green takeover is inevitable under every dynamic. There are various possible arguments for how such a firewall may emerge in smaller rings. However, for all large enough $n$, we may make the cheap observation that such a firewall will appear in the initial configuration with probability $>1-\eps$ by the weak law of large numbers.
\end{proof}

Theorem \ref{thmm:straddle} is necessarily probabilistic; it is not true that every ring will suffer green takeover. For example given any $w$ and tolerances satisfying $\tau_r > \frac{1}{2} \geq \tau_g$ where $\lceil (2w+1) \tau_r \rceil + \lceil (2w+1) \tau_g \rceil \leq 2w+1$, a ring comprising alternating blocks of $\lceil (2w+1) \tau_r \rceil$ many red nodes followed by $\lceil (2w+1) \tau_g \rceil$ many green nodes will be totally static throughout under all three dynamics.\\

\section{Hope and Intractability} \label{section:themus}

In this section our focus will be entirely on the selective model where $\tau_g, \tau_r > \frac{1}{2}$. Recall that an unhappy red node $x$ is \emph{hopeful} if $G(\N(x)) \geq \tau_g(2w+1)-1$. (Notice that when $\tau_g, \tau_r > \frac{1}{2}$ unhappiness automatically follows from this second condition, for large enough $w$.) In this situation, a new type of stable interval emerges:

\begin{defin}
An interval $J$ of length $w+1$ is \emph{green intractable} if $G(J)< \tau_g(2w+1) - (w+1)$.
\end{defin} 

We define red intractability analogously. The point of this definition is that, regardless of the situation outside $J$, no red node inside a green intractable interval can be hopeful, and thus can never turn green. Hence green firewalls cannot spread through green intractable intervals. Now we shall compare the probabilities of hopeful individuals versus intractable intervals.

The probability $F_r$ that, in the initial configuration, a randomly selected red node is hopeful is the same as the probability that a green node in that position would be happy. Setting $X \sim b(2w, \rho)$ we have $F_r = {\bf P}(X \geq h)$ where $h:=\lceil \tau_g(2w+1)\rceil -1$.

On the other hand, let $J$ be an interval of length $w+1$, and let $T_g$ be the probability that such an interval, selected uniformly at random is green intractable. We now introduce an approximation for $T_g$ as follows. Set $Y \sim b(w+1, \rho)$, and set $T_g'={\bf P}(Y \leq \ell)$ where $\ell:= \lceil (\tau_g - \frac{1}{2})(2w+1) \rceil -1 $. It is easy to see that $T_g \approx T_g'$, and thus we may work with $T_g'$ in place of $T_g$ in all that follows.

We wish to understand the ratio $\frac{F_r}{T_g}$. First notice that if $\tau_g \leq \rho$ then $T_g \to 0$ at an exponential rate in $w$, while $F_r \geq \frac{1}{2}$. Hence $\frac{F_r}{T_g} \to \infty$.

Similarly, if $\tau_g \geq \frac{\rho +1}{2}$ then $F_r \to 0$ at an exponential rate in $w$, while $T_g \geq \frac{1}{2}$. Hence $\frac{F_r}{T_g} \to 0$.

Therefore we will suppose $\rho < \tau_g < \frac{\rho+1}{2}$, and in this region we may derive some estimates from Lemma \ref{lem:binom} as before.
Firstly, taking $p=\rho$, with $h$ as above, along with $N=2w$, and some $k$ such that $1>k>\frac{1- \tau_g}{\tau_g} \cdot \frac{\rho}{1-\rho}$ (which is possible since $\rho < \tau_g$), we find that:
\begin{equation} \label{equation:hopefuleq} F_r \approx \rho^{h} (1-\rho)^{2w-h} \bpm 2w \\ h \epm. \end{equation}

Similarly, working with $\ell$ above in place of $h$ and taking $N=w+1$ as well as some $k'$ where $1 > k' > \frac{2(1- \tau_g)}{2 \tau_g -1} \cdot \frac{\rho}{1-\rho}$ (which is possible since $\tau_g> \frac{\rho+1}{2}$), we find  $$T_g \approx \rho^\ell (1- \rho)^{w+1 - \ell} \bpm w+1 \\ \ell \epm.$$

Taking the ratio of these two, we find
$$\frac{F_r}{T_g} \approx \rho ^{h- \ell} (1-\rho)^{w-1 + \ell - h} \frac{\bpm 2w \\ h \epm}{\bpm w+1 \\ \ell \epm}.$$

Now we use Stirling's approximation, which gives us that 
$$\frac{F_r}{T_g} \approx \rho ^{h- \ell} (1-\rho)^{w-1 + \ell - h} \frac{(2w)^{2w+ \frac{1}{2}} \ell^{\ell + \frac{1}{2}}(w+1-\ell)^{w+ \frac{3}{2}- \ell}}{(w+1)^{w+\frac{3}{2}}h^{h+ \frac{1}{2}}(2w-h)^{2w-h + \frac{1}{2}}}.$$

Next we introduce the approximations $h \approx 2w \tau_g$ and $\ell \approx 2w \left(\tau_g - \frac{1}{2}\right)$, and noticing that $w- \ell \approx 2w-h \approx 2w(1-\tau_g)$, we get

$$\frac{F_r}{T_g} \approx Q(w) \cdot \rho ^w \frac{(2w)^{2w} (2w)^{2w(\tau_g - \frac{1}{2})}\left( \tau_g - \frac{1}{2} \right)^{2w(\tau_g - \frac{1}{2})}}{w^w (2w)^{2w \tau_g} \tau_g^{2w \tau_g}}$$
for some polynomial $Q(w)$. Hence
\begin{equation} \label{equation:mubounds} \frac{F_r}{T_g} \approx \left( \frac{2 \rho \left(\tau_g - \frac{1}{2} \right)^{2 \tau_g -1}}{\tau_g^{2 \tau_g}} \right)^w. \end{equation}

Thus we deduce the existence of the thresholds $\mu^\rho_g$ as the root, when it exists, of $$g(x):=\frac{(x-\frac{1}{2})^{2x-1}}{x^{2x}} = \frac{1}{2\rho}.$$

Similarly, $\mu^\rho_r$ is the root of $g(x)=\frac{1}{2(1-\rho)}$. Comparing this with Equation \ref{equation:froot}, and noticing that $g(x)=f(1-x)$, we deduce that $\mu^\rho_r = 1 - \kappa^\rho_g$ and similarly $\mu^\rho_g = 1 - \kappa^\rho_r$. Thus we have arrived at:

\begin{prop} \label{prop:musumup}
For any $\rho \in (0,1)$, we work in the initial configuration and interpret $F_r$ as the probability that a uniformly randomly selected red node is hopeful, and $T_g$ as that of a uniformly randomly selected node lying within a green intractable interval. Then there exists a threshold $\mu^\rho_g$ where $\rho < \mu^\rho_g < \frac{1}{2}(1+\rho)$, such that for any $\tau_g > \frac{1}{2}$:

\begin{itemize}[$\bullet$]
\item If $\tau_g < \mu^\rho_g$, there exists $\zeta \in (0,1)$ so that $T_g < \zeta^{w} F_r$ for all $w$.
\item If $\tau_g > \mu^\rho_g$, there exists $\zeta \in (0,1)$ so that $F_r < \zeta^{w} T_g$ for all $w$.
\end{itemize}

(Similarly, there exists a threshold $\mu_r^\rho$ where $1-\rho < \mu^\rho_r < 1- \frac{1}{2}\rho$ such that corresponding statements about $F_g$ and $T_r$ hold.)
\end{prop}

The thresholds $\mu^\rho_r$ and $\mu^\rho_g$ are illustrated in Figure \ref{fig:Thresh1}. Notice that $\frac{1}{2} < \tau_g < \mu^\rho_g$ is only possible when $\rho> \frac{1}{4}$ and similarly $\frac{1}{2} < \tau_r < \mu^\rho_r$ requires that $\rho< \frac{3}{4}$.

Besides understanding the relative frequency of intractable intervals and hopeful nodes, we shall also need to know whether red or green hopeful nodes are more numerous in a given scenario. Happily, we do not need to introduce yet another threshold:

\begin{prop} \label{prop:outnumber2}
If $\tau_g, \tau_r > \frac{1}{2}$, the scenario $(\rho, \tau_g, \tau_r)$ is red dominating if and only if there exists $\eta \in (0,1)$ so that for all $w$, we have $F_r < \eta^w F_g$. The same holds with the roles of red and green interchanged. \end{prop}

\begin{proof}
If $\rho \leq 1- \tau_r$, then automatically $\rho < \tau_g$. Also $F_g \to 1$ (if $\rho \leq 1- \tau_r$) or $F_g \to \frac{1}{2}$ (if $\rho = 1- \tau_r$) as $w \to \infty$. At the same time, $F_r \to 0$ at an exponential rate in $w$. Thus the result holds by Corollary \ref{coro:easierdom}.

Similarly, if $1-\rho \leq 1- \tau_g$, then $1- \rho < \tau_r$, and result holds by Corollary \ref{coro:easierdom} with the roles of red and green interchanged.

Thus we are left with the case $1- \tau_r < \rho < \tau_g$, where the argument amounts to applying Stirling's approximation to \ref{equation:hopefuleq} and to the equivalent criterion for $F_g$ and taking the ratio of the two expressions, exactly as in Proposition \ref{prop:outnumber}. We leave the details to the reader.
\end{proof}

Recall also that in Lemma \ref{lem:domfacts} we have already established some useful facts about red/green domination in our current region of interest $\tau_g, \tau_r > \frac{1}{2}$.

\section{Hopelessness and Frustration: the case \texorpdfstring{$\tau_r > \frac{1}{2}$}{tr>1/2}} \label{section:stagnancy}

In this section we limit ourselves to the selective model in scenarios where $\tau_g$, $\tau_r > \frac{1}{2}$, and shall prove Theorems \ref{thmm:take2} - \ref{thmm:stag2}. First we assume that $\frac{1}{2}< \tau_g < \mu_g^\rho$ and $\tau_r> \frac{1}{2}$, and that $(\rho, \tau_g, \tau_r)$ is green dominating. Notice that this implies that $\tau_g < \frac{1}{2}(1+\rho)$ by Proposition \ref{prop:musumup}.

We shall establish green takeover, thus proving Theorem \ref{thmm:take2}, using much the same machinery as in section \ref{section:longsect}. An example is illustrated in Figure \ref{fig:takeagain1}.

\begin{figure}[htbp]
\includegraphics[width=15cm, clip=true, trim= 2cm 3cm 2cm 3cm]{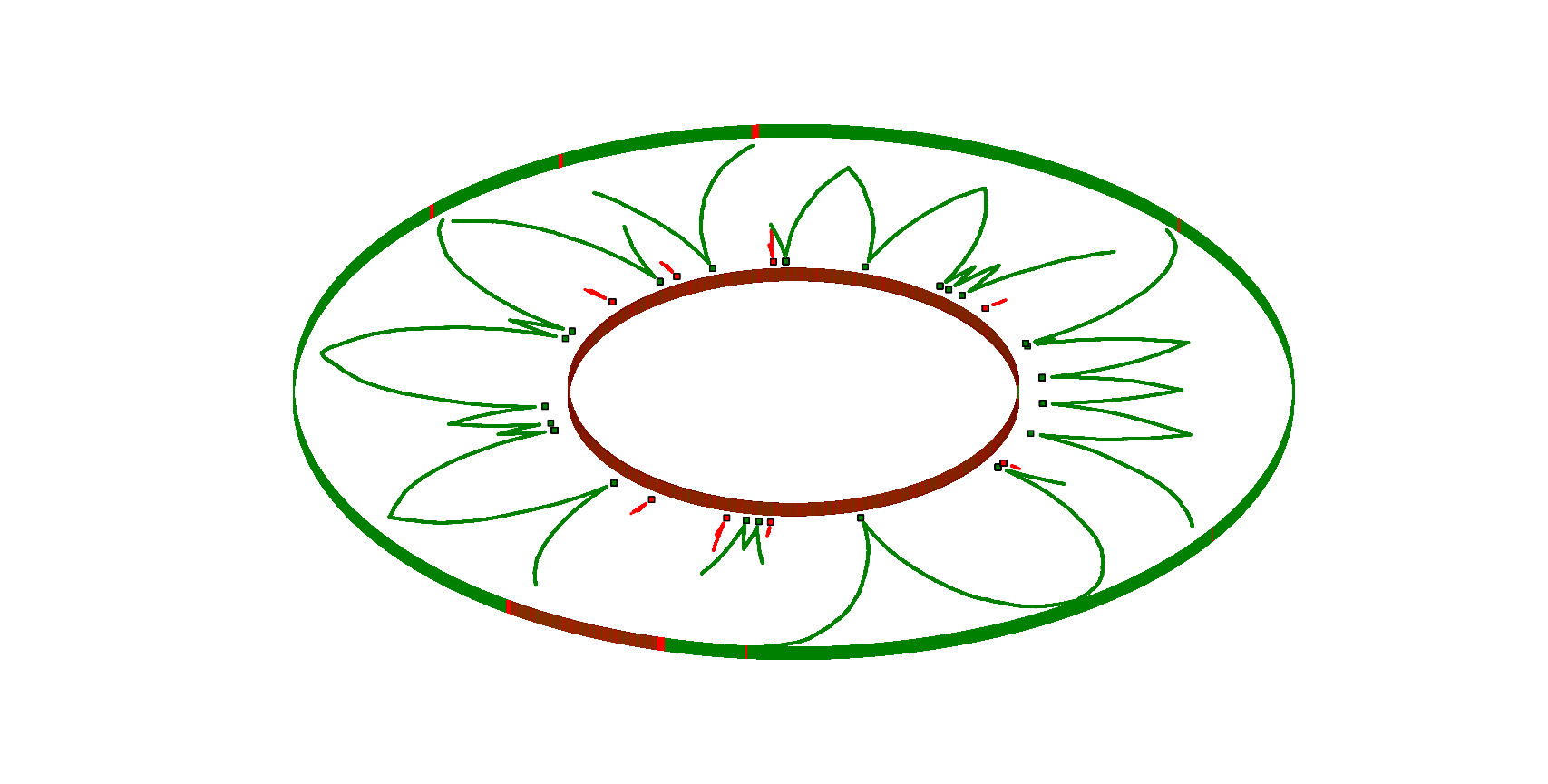} 
\caption{$\rho=0.4$, $\tau_g=0.55<\mu^{0.4}_g \approx 0.58$, $\tau_r=0.75>\mu^{0.4}_r \approx 0.71$, $w=70$, $n=1,000,000$}
\label{fig:takeagain1}
\end{figure}

As usual, we begin by picking a node $u_0$ uniformly at random and aim to establish that $u_0$ will be green in the finished ring with probability $>1-\eps$. We postpone the case $\tau_g  \leq \rho$ and assume that $\tau_g > \rho$. Thus, by Corollary \ref{coro:easierdom} it also follows that $1 - \tau_r < \rho$. Via Propositions \ref{prop:musumup} and \ref{prop:outnumber2}, our hypotheses imply the following:

\begin{rem} \label{rem:dagger2} There exist $0<\eta, \zeta<1$ so that for all $w \gg 0$ we have $T_g< \eta^w F_r$ and $F_g < \zeta^w F_r$. \end{rem}

We also know, by Hoeffding's inequality (Proposition \ref{prop:Hoeffding}), that red nodes are unlikely to be hopeful: \begin{equation} \label{equation:littlehope} F_r < \exp \left( -2(w+1)(\tau_g - \rho)^2  \right). \end{equation}

Let 
\begin{equation} \label{equation:theta*again} \theta^*:=\min \left\{ \frac{m}{2w+1}: \frac{m}{2w+1}>\tau_g \ \& \ m \in \NN \right\}. \end{equation}
Then $\theta^* \to \tau_g$ as $w \to \infty$ and thus by assumption $\rho< \theta^* < \frac{1}{2}(1+\rho)$ for large enough $w$.

Now we set $l_0:=u_0$ and define $l_{i+1}$ to be the first node to the left of $l_i - (2w+1)$ which is either hopeful, or satisfies $\GD_{\theta^*}$, or belongs to a green intractable interval, so long as this node lies within $[u_0 - \frac{n}{4}]$. The $r_i$ are defined identically to the right. Again we shall choose a specific value of $k_0$ in due course. As before we derive the following from Remark \ref{rem:dagger2}, Lemma \ref{lem:gen}, and Bound \ref{equation:littlehope}:

\begin{lem} \label{lem:bigli}
For any $k_0>0$ and $\eps'>0$, there exists $d>0$ such that for all large enough $w$ the following hold with probability $>1 - \eps'$ 
\begin{enumerate}
\item $l_{k_0}, \ldots, l_{1}, r_{1}, \ldots, r_{k_0}$ are all defined. 
\item $l_{k_0}, \ldots, l_{1}, r_{1}, \ldots, r_{k_0}$ all satisfy $\GD_{\theta^*}$.
\item There are no hopeful green nodes in $[l_{k_0}, r_{k_0}]$.
\item No node in $[l_{k_0}, r_{k_0}]$ belongs to a green intractable interval.
\item For $i \geq 2$, we have $|l_{i+1}-l_i|, |r_{i+1}-r_i|, |r_1 - l_1| \geq e^{dw}$.
\end{enumerate} \end{lem}

Next, we can apply Corollary \ref{coro:theliaresmooth} once again to conclude that for any $\eps'>0$ and $k\geq 1$ then $\Sm(l_i)$ and $\Sm(r_i)$ hold for all $i$ with probability $>1-\eps'$ for all $n \gg w \gg 0$. Furthermore by Lemma \ref{lem:thelicomplete}, we know that the $l_i$ and $r_i$ are each very likely to green complete. Next, since $\theta^* < \frac{1}{2}(1+ \rho)$ for large enough $w$, we may apply Corollary \ref{coro:thelispark} for some $\alpha$ so that $Z(\theta_\alpha,\rho)>0$. This establishes that at least one of the $l_i$ and at least one of the $r_i$ will $\alpha$-spark with some probability $\delta'$, independent of $w$. Finally, noting also that $\theta^* \to \tau_g > 1- \tau_r$ as $w \to \infty$ we may apply Lemma \ref{lem:sparkgrows} to establish that those $l_i$ and $r_i$ which do spark will initiate green firewalls.

This again allows us to pick $k_0$ guaranteeing that $u_0$ will be engulfed in a green firewall with probability $>1 - \eps$.\\

\subsection*{\emph{Proof of Theorem \ref{thmm:take2} when $\tau_g  \leq \rho$}}
Here we find that the probability that a randomly chosen red node is hopeful  $F_r \to 1$ (if $\tau_g<\rho$) or $F_r \to \frac{1}{2}$ (if $\tau_g=\rho$) as $n,w \to \infty$.

Suppose first that $\tau_g<\rho$. Once again we define $l_0:=u_0$ and $l_{i+1}$ to be the first node to the left of $l_i - (2w+1)$ which is either hopeful, or satisfies $\GD_{\theta^*}$, or belongs to a red intractable interval, so long as this node lies within $[u_0 - \frac{n}{4}]$. The $r_i$ are defined identically to the right. Let $v$ (respectively $v'$) be the nearest hopeful green node to the left (right) of $u_0$. Again, the $l_i$ and $r_i$ are close together but far from $v$ and $v'$:

\begin{lem} \label{lem:thelioncemore}
For any $k_0>0$ and $\eps'>0$ there exists $d>0$ such that for all large enough $w$ the following hold with probability $>1 - \eps'$ 
\begin{enumerate}
\item $l_{k_0}, \ldots, l_{1}, r_{1}, \ldots, r_{k_0}$ are all defined. 
\item $l_{k_0}, \ldots, l_{1}, r_{1}, \ldots, r_{k_0}$ all satisfy $\GD_{\theta^*}$.
\item No node in $[l_{k_0}, r_{k_0}]$ lies in an intractable red intervals.
\item For $i \geq 1$, we have $|l_i-v|, |r_{i}-v'| \geq e^{dw}$.
\end{enumerate} \end{lem}

Now, Corollary \ref{coro:theliaresmooth} applies yet again, giving us that for each $\eps'>0$ and $k_0 \in \NN$, $\Sm(l_i)$ and $\Sm(r_i)$ holds for all $i \leq k_0$ with probability $>1 - \eps'$ for all $w \gg 0$.

Then we can apply Lemma \ref{lem:thelicomplete}, giving us that the $l_i$ and $r_i$ all green complete with probability $>1- \eps'$ for all large enough $w$. Finally Corollary \ref{coro:thelispark} and Lemma \ref{lem:sparkgrows} also apply to the $l_i$ and $r_i$ and once again allow us to choose $k_0$, such that the probability that no $l_i$ or no $r_i$ initiate a green firewall is $<\eps$.\\

Finally we address the case $\tau_g=\rho$. Attempting to apply our previous results directly brings us into conflict with the hypothesis, in various places, that $\theta \neq \rho$. Again, we may get around this straightforwardly by picking $\tau_g'>\tau_g$ where $\tau_g'<\mu^\rho_g$ such that $(\rho,\tau_g',\tau_r)$ is green dominating. Then $\tau_g'>\rho$, and we proceed through this section's main argument with $\tau_g'$ in place of $\tau_g$ throughout, starting with Remark \ref{rem:dagger2}. Since any unhappy red node which is $\tau_g'$-hopeful is automatically $\tau_g$-hopeful, we deduce the existence of green firewalls exactly as previously. This concludes the proof of Theorem \ref{thmm:take2}.

\subsection*{\emph{Proof of Theorem \ref{thmm:stag1} } }

As before we shall reverse the roles of red and green for convenience, and thus aim to show that if $\frac{1}{2} < \tau_r < \mu^\rho_r$ and $\mu^\rho_g < \tau_g$, if the scenario is green-dominating and additionally if $\tau_g < \frac{1}{2}(1+\rho)$, then the scenario is static almost everywhere. An example is illustrated in Figure \ref{fig:stag2}.

\begin{figure}[hbtp]
\includegraphics[width=15cm, clip=true, trim= 2cm 3cm 2cm 3cm]{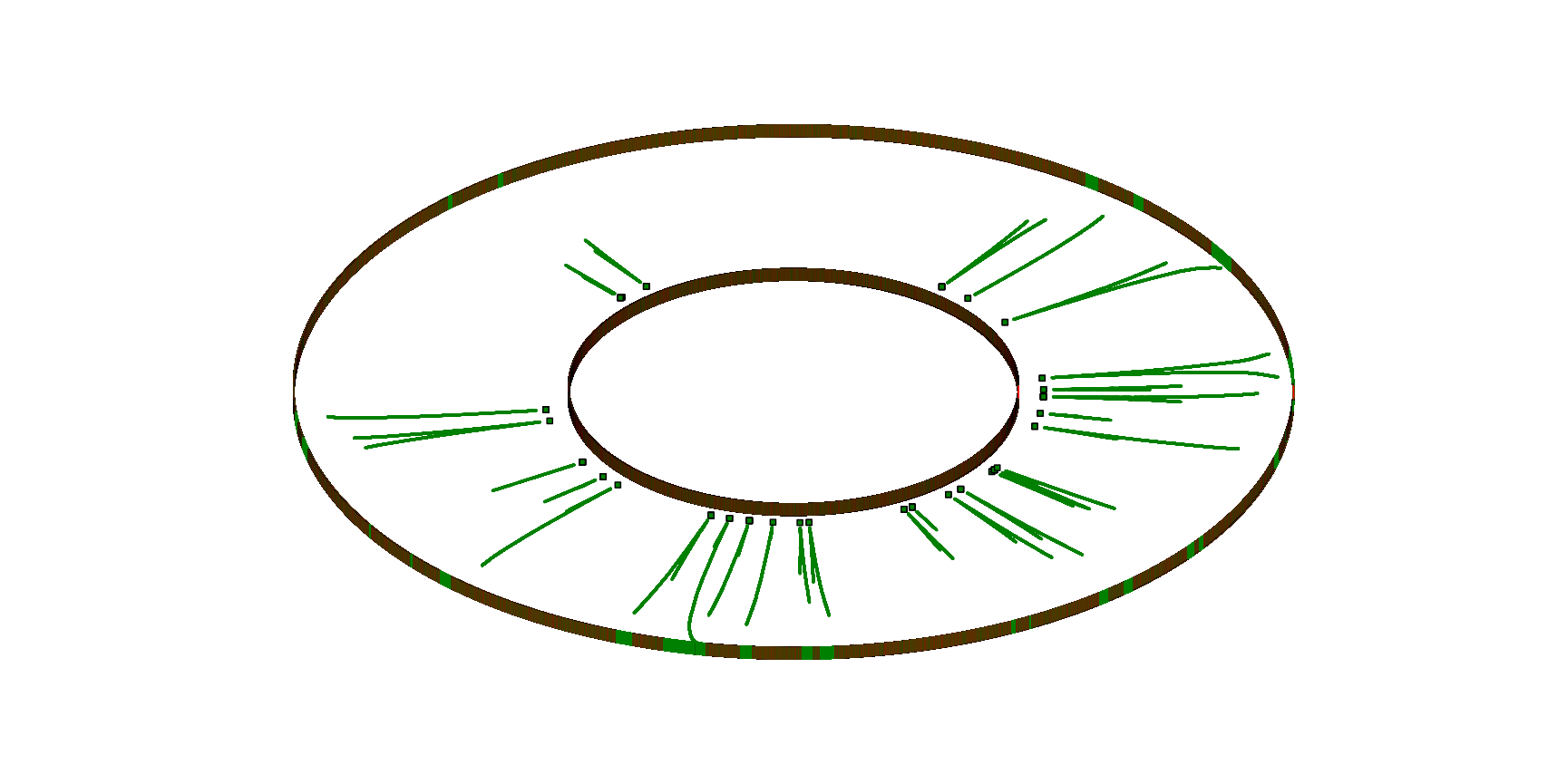} 
\caption{$\rho=0.6$, $\tau_g=0.73>\mu^{0.6}_g \approx 0.71$, $\tau_r=0.56<\mu^{0.6}_r \approx 0.58$, $w=100$, $n=5,000,000$}
\label{fig:stag2}
\end{figure}

First we define $l_0:=u_0$, and define $l_{i+1}$ to be the first node to the left of $l_i - (2w+1)$ which is either hopeful or satisfies $\GD_{\theta^*}$ (defined as in Equation \ref{equation:theta*again}), so long as this node lies within $[u_0 - \frac{n}{4}]$. The $r_i$ are defined identically to the right. Exactly as before, Remark \ref{rem:dagger2} gives us the following:

\begin{lem} \label{lem:bigliagain}
For any $k_0>0$ and $\eps'>0$, there exists $d>0$ such that for all large enough $w$ the following hold with probability $>1 - \eps'$ 
\begin{enumerate}
\item $l_{k_0}, \ldots, l_{1}, r_{1}, \ldots, r_{k_0}$ are all defined. 
\item $l_{k_0}, \ldots, l_{1}, r_{1}, \ldots, r_{k_0}$ all satisfy $\GD_{\theta^*}$.
\item There are no hopeful green nodes in $[l_{k_0}, r_{k_0}]$.
\item For $i \geq 2$, we have $|l_{i+1}-l_i|, |r_{i+1}-r_i|, |r_1 - l_1| \geq e^{dw}$.
\end{enumerate} \end{lem}

Notice that $\tau_g>\mu^\rho_g>\rho$, meaning that for all large enough $w$, we have $\theta^* \neq \rho$. Hence Corollary \ref{coro:theliaresmooth} applies again, giving us that for each $\eps'>0$ and $k_0 \in \NN$, $\Sm(l_i)$ and $\Sm(r_i)$ holds for all $i \leq k_0$ with probability $>1 - \eps'$ for all $w \gg 0$.

Then we can apply Lemma \ref{lem:thelicomplete}, giving us that the $l_i$ and $r_i$ all green complete with probability $>1- \eps'$ for all large enough $w$. Finally noting that $\theta^* \to \tau_g<\frac{1}{2}(1+\rho)$, we may finish off by applying Corollary \ref{coro:thelispark} for some suitable $\alpha$ and Lemma \ref{lem:sparkgrows} to the $l_i$ and $r_i$, once again allowing us to choose $k_0$, such that the probability that no $l_i$ or no $r_i$ initiate a green firewall is $<\eps$.

Let $l$ and $r$ be the $l_i$ and $r_i$ nearest $u_0$ which do initiate green firewalls. It is certain that these green firewalls will spread towards $u_0$ until they hit green intractable intervals. Suppose this has happened by stage $s$. At this stage, looking away from $u_0$ we encounter intractable intervals of both colours before hopeful nodes of either colour. Thus with probability $>1- \eps$, the node $u_0$ will not change colour.\\

\subsection*{\emph{Discussion of Question \ref{questionn:static3?}}}

Similar remarks apply here as in the case of Question \ref{questionn:static2?}, and we expect that any technique which resolves that question will apply here too. We may also repackage the counterexamples we mentioned in that case: $(\rho,\tau_g,\tau_r)=(0.7,0.87,0.51)$. Here $\tau_g> \mu^{0.7}_g \approx 0.79$ while $\frac{1}{2}< \tau_r < \mu^{0.7}_g \approx 0.52$, and the scenario is green dominating. However also $\tau_g > \frac{1}{2}(1+\rho)=0.85$.

Exchanging the roles of red and green here produces $(\rho,\tau_g,\tau_r)=(0.3,0.51,0.87)$ which is located within the upper purple region of Figure \ref{fig:map30}.

Again we can find cases where $Z(\theta^*,\rho)>0$ also fails. For instance, $(\rho,\tau_g,\tau_r)=(0.74,0.93,0.502)$ is green dominating, satisfies $\tau_g > \mu^{0.74}_g \approx 0.81$ and $\frac{1}{2}< \tau_r < \mu^{0.74}_g \approx 0.503$ while also $Z(\theta^*,\rho) \to -0.06$ approximately as $\theta^* \to \tau_g$.

\subsection*{\emph{Proof of Theorem \ref{thmm:stag2}  }}

To complete our analysis of the selective model, we turn to the case where both $\mu^\rho_g < \tau_g$ and $\mu^\rho_r < \tau_r$. Recall that Theorem \ref{thmm:stag2} asserts that such a scenario is static almost everywhere, as illustrated in Figure \ref{fig:stag1}. The proof of this follows from Lemma \ref{lem:gen} applied twice, by interpreting $P(u)$ as the event that $u$ lies in a green (respectively red) intractable interval and $Q(u)$ as $u$ being green (red) and hopeful. The necessary probabilistic bounds are given in Proposition \ref{prop:musumup}.

\begin{figure}[!htbp]
\includegraphics[width=15cm, clip=true, trim= 2cm 3cm 2cm 3cm]{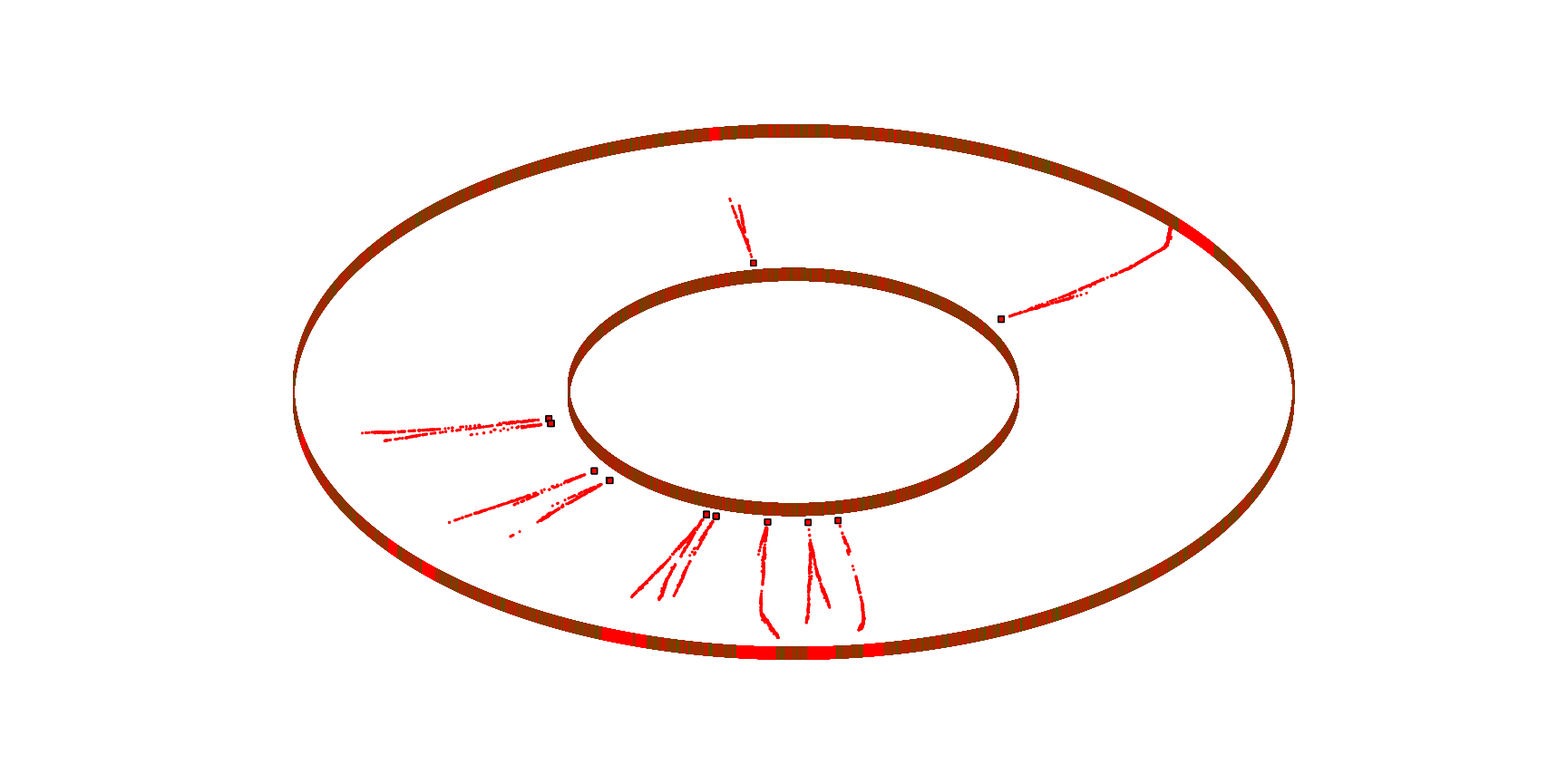} 
\caption{$\rho=0.4$, $\tau_g=0.65>\mu^{0.4}_g \approx 0.58$, $\tau_r=0.75>\mu^{0.4}_r \approx 0.71$, $w=50$, $n=100,000$}
\label{fig:stag1}
\end{figure}

\section{Alternative dynamics, perturbed processes, open problems, \texorpdfstring{\&}{and} run-time} \label{section:bigtau}

Finally, we turn our attention to scenarios where both $\tau_r, \tau_g > \frac{1}{2}$ under the incremental and synchronous dynamics. In the case of a synchronous model, we can mimic Theorem \ref{thmm:straddle} with the following proposition, whose proof is deferred to Appendix \ref{section:appD}. This also serves to establish that, under the given conditions on $\tau_g$ and $\tau_r$ and for all large enough $n$, the initial configuration is highly likely to be such that the process is guaranteed to finish.

\begin{restatable}{prop}{happysim}
\label{prop:happysim}
In the synchronous model, suppose that $\frac{1}{2} < \tau_g < \frac{2}{3}$ and $\tau_g < \tau_r$. Then green takes over totally. 
\end{restatable}

Recall that Conjecture \ref{conj:abovehalf} generalises Proposition \ref{prop:happysim}, asserting that if $\frac{1}{2} < \tau_g < \tau_r$, then green will take over totally under both the incremental and synchronous dynamics.

Towards this conjecture, we briefly make the some observations about a perturbed version of our model. The initial configuration is set up exactly as previously described. But now, for any $1>\eps>0$ we define the $\eps$-perturbed model as follows: at each time-step with probability $1-\eps$ we proceed as in the previous incremental model, but with probability $\eps$ we pick a node at random and alter its colour. Thus $\eps$ can roughly be thought of as the probability of an error at each stage. We remark that this process is a regular perturbed Markov process in the sense of Section 3.4 of \cite{Y}.

The advantage of working with such a perturbed process is that for each $\eps$ the Markov process is irreducible: any state of the ring is accessible from any other in a finite number of steps. Thus, following Young and notably in the works of Zhang (\cite{Z1}, \cite{Z2}, \cite{Z3}), it  has become common practice to analyse Schelling segregation via perturbed models of this sort, and to examine the limit as $\eps \to 0$. It is particularly of interest to identify the \emph{stochastically stable} states, which are the states most likely to emerge in the long run, as $\eps \to 0$. They are defined as follows: for each $\eps>0$, Markov chain theory guarantees that there will be a unique stationary distribution $\mu^\eps$ on the state-space. A state $s$ is stochastically stable if $\mu^0(s):=\lim_{\eps \to 0} \mu^\eps(s)>0$.

Now, our unperturbed model contains exactly two recurrence classes, namely the absorbing states $\mathfrak{G}$ and $\mathfrak{R}$ representing totally green and red rings respectively. By Young's Theorem (Theorem 3.1 of \cite{Y}) whether or not these are stochastically stable will depend on their \emph{stochastic potential}. We refer the interested reader to \cite{Y} for the formal definition, however in the current context its meaning is straightforward: the stochastic potential of the state $\mathfrak{G}$ is the minimum number of errors required to reach it from the opposite recurrence class $\mathfrak{R}$.

That is to say, the stochastic potential of $\mathfrak{G}$ is simply the minimum number of green nodes which have to be artificially inserted into an otherwise entirely red ring in order to generate one unhappy red element. (Notice that if these green nodes are inserted in consecutive positions, the remainder of the transformation from $\mathfrak{R}$ to $\mathfrak{G}$ may then take place error-free.) This number is $\lfloor (1-\tau_r)(2w+1) \rfloor +1$. Similarly the stochastic potential of state $\mathfrak{R}$ is $\lfloor (1-\tau_g)(2w+1) \rfloor +1$. Thus we have the following result, which supports, but does not formally imply, the incremental case of Conjecture \ref{conj:abovehalf}:

\begin{thmm} \label{thmm:stochastic}
In the perturbed model, suppose that $\frac{1}{2} < \tau_g < \tau_r$. Then total green takeover represents the only stochastically stable state.

If $\frac{1}{2} < \tau_g = \tau_r$, total takeover by either colour is stochastically stable.
\end{thmm}

Now, there is a strong sense in which Conjecture \ref{conj:abovehalf} and Theorem \ref{thmm:stochastic} fail to give the entire story. So
we finish with some remarks on the run-time of the process, and hypothesise the existence of another important tipping point in each of the incremental and synchronous models. The following are easy to see from our analysis so far:

\begin{itemize}[$\bullet$] \item In the selective model, the the expected run-time is at most linear in $n$.
\item In the incremental model, if either $\tau_g, \tau_r < \frac{1}{2}$, the expected run-time is at most linear in $n$.
\item In the synchronous model, if either $\tau_g, \tau_r < \frac{2}{3}$, the run-time will be at most linear in $n$ with probability $>1- \eps$ for all large enough $n$ and $w$.
\end{itemize}

We further conjecture the following:

\begin{conj} \ \ 

\begin{itemize}[$\bullet$]

\item In the incremental model, if either $\tau_g, \tau_r < \frac{3}{4}$, the expected run-time is at most linear in $n$.

\item In the incremental model, if both $\tau_g, \tau_r > \frac{3}{4}$, the process will finish, but the expected run-time is superpolynomial in $n$.

\item In the synchronous model, if both $\tau_g, \tau_r > \frac{2}{3}$, the expected run-time is superpolynomial in $n$ (which includes the possibility of never finishing).
\end{itemize}
\end{conj}

We briefly discuss the intuition behind this conjecture in the incremental case. Consider a green firewall $[a,c]$ where $c-a \gg 2w+1$ and $c+1$ is red. Let $b$ be the rightmost happy element within this firewall. We assume that $b<c$. Now consider the interval $I=[b+1,b+w]$. Since $b$ is happy and $b+1$ is not, $I$ must contain exactly $\lceil \tau_g(2w+1) \rceil -(w+1)$ many green nodes which will all be unhappy. The question of interest is whether the happy firewall, which currently ends at $b$, is more likely to advance or retreat.

If $\tau_g> \frac{3}{4}$, then $\UG(I)>\frac{1}{2}w$, and irrespective of the remaining nodes in $I$, the happy firewall is more likely to retreat. However, if $\tau_g<\frac{3}{4}$, then $\UG(I)<\frac{1}{2}w$ for all large enough $w$. At the same time, $\UR(I) \geq \min \left\{ w- \UG(I), \lceil \tau_r(2w+1) \rceil -(w+1) \right\}$, and presuming $\tau_g< \tau_r$, the happy firewall is more likely to advance.

The situation where $\tau_g, \tau_r > \frac{3}{4}$ is redolent of the classic Ehrenfest Urn, a simple model of a thermodynamic process proposed by T. \& P. Ehrenfest in \cite{Eh}. An urn is filled with a fixed number ($w$) of balls, divided in some proportion between red green. At each time step, a ball is selected uniformly at random from the urn and replaced with a ball of the opposite colour. It is fairly clear that the model's limiting distribution as $t \to \infty$ is $b(w, \frac{1}{2})$, regardless of the starting configuration. The celebrated analysis of Kac in \cite{K} also established that if the urn begins in an all-green state (and subject to the technical proviso that $w$ is even) the expected time until this state recurs is exponential in $w$.\\

\begin{acknowledgements}

Lewis-Pye was was supported by a Royal Society University Research Fellowship.  

Barmpalias was supported by the Research Fund for International Young Scientists from the National Natural Science Foundation of China, grant numbers 613501-10236 and 613501-10535, and an International Young Scientist Fellowship from the Chinese Academy of Sciences; support was also received from the project Network Algorithms and Digital Information from the Institute of Software, Chinese Academy of Sciences and a Marsden grant of New Zealand.

\end{acknowledgements}

\appendix
\section{Deferred proofs from section \ref{section:intro}} \label{section:appA}

We deferred the proof of the following result from the introduction, and present it below:

\itends*

\begin{proof}
Our strategy is to define a harmony index for the whole ring, and establish that this quantity has a finite upper bound, but also increases (by at least some minimum positive amount) with each legitimate move. This will give the result. If $\tau_g+\tau_r \leq 1$, we start by picking $\chi$ such that $$\frac{1 - \tau_g}{\tau_g} \geq \chi \geq \frac{\tau_r}{1- \tau_r}.$$
If instead $\tau_g+\tau_r > 1$, then we require that $w$ is large enough to allow us to choose $\chi$ where
$$0 < \frac{1 - \tau_g + \left( \frac{1}{2w+1} \right) }{\tau_g - \left( \frac{1}{2w+1} \right)} < \chi < \frac{\tau_r - \left( \frac{1}{2w+1} \right) }{1- \tau_r + \left( \frac{1}{2w+1} \right) }.$$
Now for a node $x$ at time $t$, we'll write $G_t(x)=1$ (respectively $G_t(x)=0$) if $x$ is green (red) at time $t$, and define
$$A_t(x):= \left\{ \begin{array}{ll} \chi & \textrm{ if } G_t(x)=1 \\
1 & \textrm{ if } G_t(x)=0. \end{array} \right. $$
Similarly define $$L_t(x):= \frac{|\{ y \in \N(x) \ : \ G_t(y)= G_t(x) \}|}{2w+1}.$$
Now we define the following harmony index: $S(t):= \sum_{x} A_t(x) L_t(x)$. Clearly this is bounded above by $n \cdot \max\{1, \chi\}$. We wish to compare $S(t+1)$ and $S(t)$. Suppose that $x$ is the node whose colour changes. Then $L_{t+1}(x) = 1-L_t(x) + \frac{1}{2w+1}$. Similarly for $y \in \N(x)$ with $x \neq y$ and $G_t(y)=G_t(x)$ we have $L_{t+1}(y)=L_{t}(y) - \frac{1}{2w+1}$, and there are $(2w+1)L_t(x) -1$ many such $y$.
At the same time, for $z \in \N(x)$ with $G_t(z) \neq G_t(x)$ we have $L_{t+1}(z)=L_{t}(z) + \frac{1}{2w+1}$, and there are $(1-L_t(x))(2w+1)$ many such $z$. Hence, 
\begin{dmath*}S(t+1)=S(t) - A_t(x)L_t(x) + A_{t+1}(x) \left( 1-L_t(x) + \frac{1}{2w+1} \right) \\ - ((2w+1)L_t(x) -1)A_{t}(x)\frac{1}{2w+1} + \left( 1-L_t(x) \right) (2w+1)A_{t+1}(x)\frac{1}{2w+1}.\end{dmath*}
Thus $$S(t+1) - S(t) = 2A_{t+1}(x) - 2(1+\chi)L_t(x) + \frac{1+\chi}{2w+1}.$$
Hence it suffices to show that $(1+\chi)L_t(x) <A_{t+1}(x)$ for which we check the four possible cases.

Suppose first that $\tau_g+\tau_r \leq 1$. If $G_t(x)=1$ then, since $x$ is unhappy $L_t(x)< \tau_g$ and $(1+\chi)\tau_g \leq 1 = A_{t+1}(x)$ as required. On the other hand, if $G_t(x)=0$ then $L_t(x)<\tau_r$ and $(1+\chi)\tau_r \leq \chi = A_{t+1}(x)$, again as required.

Suppose now that $\tau_g+\tau_r > 1$. This time if $G_t(x)=1$, then since $x$ is hopeful $1-L_t(x) + \frac{1}{2w+1} \geq \tau_r$ meaning $(1+\chi)L_t(x) \leq (1+\chi)(1- \tau_r + \frac{1}{2w+1}) < 1 =A_{t+1}(x)$ again by choice of $\chi$. Finally, if $G_t(x)=0$ then hopefulness tells us that $1-L_t(x) + \frac{1}{2w+1} \geq \tau_g$ which gives us $(1+\chi)L_t(x) \leq (1+\chi)(1- \tau_g + \frac{1}{2w+1}) < \chi =A_{t+1}(x)$. \end{proof}

\section{Deferred proofs from section \ref{section:domsect}} \label{section:appB}

Here we present proofs of two of the more technical matters from section \ref{section:domsect}, starting with the following:

\domfacts*

\begin{proof}

Define $h:T_1 \cup T_2 \to \RR$ by $$h(x,y):= \frac{x^{\left(\frac{x}{1-x - y}\right)}{\left(1 -x \right)}^{\left(\frac{1-x}{ 1- x - y}\right)} }{y^{\left(\frac{y}{1-x - y}\right)}{\left(1-y\right)}^{\left(\frac{1-y}{1-x- y}\right)} }.$$

\noindent \textbf{Claim. } $\frac{\partial h}{\partial x} <0$ and $\frac{\partial h}{\partial y} >0$ on $T_1 \cup T_2$.

\noindent \textbf{Proof of claim. } By differentiating $\ln h$, we find that
$\frac{1}{h} \frac{\partial h}{\partial x} = \frac{k(x,y)}{(1-x-y)^2}$
where $$k(x,y)=(1-y) \ln x +y \ln (1-x) - y \ln y - (1-y) \ln (1-y).$$
Since $h>0$ on $T_i$ it suffices to show that $k \leq 0$ on $S$. Well $$\frac{\partial k}{\partial x} = \frac{1-x-y}{x(1-x)}$$ whence
$\frac{\partial k}{\partial x} >0$ on $T_1$ and $\frac{\partial k}{\partial x} <0$ on $T_2$. Similarly
 $$\frac{\partial k}{\partial y} = \ln(1-x) - \ln y + \ln(1-y) - \ln x$$
meaning that $\frac{\partial k}{\partial y} >0$ on $T_1$ and $\frac{\partial k}{\partial y} <0$ on $T_2$.
So we have established that $k$ is monotonically strictly increasing in both $x$ and $y$ on $T_1$ and monotonically strictly decreasing in both $x$ and $y$ on $T_2$. Along the line $L$ we have $k(x, y)=0$, hence it must be that $\frac{\partial h}{\partial x} <0$ on both $T_1$ and $T_2$ as required. Since $\left(h(x,y) \right)^{-1} = h(y,x)$, the result for $\frac{\partial h}{\partial y}$ also follows. \textbf{QED Claim}\\

Statement 1 of the Lemma follows from the fact that, for $(\tau_g, \tau_r) \in T_1 \cup T_2$, the scenario $(\rho, \tau_g, \tau_r)$ being red dominating is equivalent to the assertion $h(\tau_g, \tau_r) < \frac{1-\rho}{\rho}$, with green domination equivalent to the reverse inequality.

Now consider the restrictions $h \! \upharpoonright_{S_i}$. Since $\lim_{(x,y)\to (0,\frac{1}{2})} h(x,y) = \lim_{(x,y)\to (\frac{1}{2},1)} h(x,y) = 4$ and $\lim_{(x,y)\to (\frac{1}{2},0)} h(x,y) = \lim_{(x,y)\to (1,\frac{1}{2})} h(x,y) = \frac{1}{4}$, from which it follows that the restriction $h:S_i \to (\frac{1}{4},4)$ is surjective for $i \in \{1,2\}$.

Thus for $\rho \geq \frac{4}{5}$ we have $h(x,y)> \frac{1}{4} > \frac{1- \rho}{\rho}$ for any $(x,y) \in S_i$. Similarly for $\rho \leq \frac{1}{5}$ we have $h(x,y) < 4 < \frac{1- \rho}{\rho}$, giving statement 2.

For statement 3, notice that if $\frac{1}{5} < \rho < \frac{4}{5}$ then $4 > \frac{1- \rho}{\rho} > \frac{1}{4}$ and the result again follows by the continuity and surjectivity of $h$ restricted to $S_i$.
\end{proof}

\easierdom*

\begin{proof}
Let $h$ be as in the proof of Lemma \ref{lem:domfacts}. We shall compute $\lim_{y \uparrow (1-\rho)} h(\rho,y)$. Well,
$$\ln h = \frac{\rho \ln \rho + (1- \rho)\ln (1-\rho) - y \ln y - (1-y) \ln (1-y)}{1-\rho-y}.$$
By L'H\^opital's rule, therefore $\lim_{y \uparrow (1-\rho)} \ln h(\rho,y) = \lim_{y \uparrow (1-\rho)} \ln \left(\frac{y}{1-y} \right) = \ln \left(\frac{1- \rho}{\rho} \right)$. Hence, by the continuity of $\ln$, we have $\lim_{y \uparrow (1-\rho)} h(\rho,y) = \frac{1- \rho}{\rho}$.

By Lemma \ref{lem:domfacts} (1) applied to $T_1$, the result therefore follows in the region $\tau_g+\tau_r<1$. The case where $\tau_g+\tau_r>1$ is identical, except that we work in $T_2$ and compute $\lim_{x \downarrow \rho} h(x,1-\rho) = \frac{1- \rho}{\rho}$. 
\end{proof}

\section{Deferred proofs from section \ref{section:longsect}} \label{section:appC}

Throughout this appendix we work in a fixed scenario $(\tau_g,\tau_r,\rho)$ and for some fixed $\theta \neq \rho$. For any node $u$, we define $x_u$ to be the first node to the left of $u$ satisfying $\GD_\theta(x_u)$. The following proposition plays a significant role in the current work:

\newbound*

This appendix is devoted to proving this. Abusing notation slightly we shall refer to the case where we may take $p = 1- \eps'$ for any $\eps'>0$ as the case $p=1$. Notice that even here we cannot simply apply Lemma \ref{lem:gen}, since we do not have access to hypothesis (ii) there. Instead we shall perform some careful counting operations, working in the vicinity of some node $v$ satisfying $\GD_\theta(v)$, and bounding above the number of other such nodes that one can expect to find nearby. Before commencing this though, we mention a version of the law of large numbers which was shall use several times:

\begin{lem}[Strong law of large numbers]\label{lem:sllnschel}
Fix a scenario and a value of $w$. Let $Q'(u)$ be a property of nodes which depends only on the vicinity of $u$ in the initial configuration (i.e. on $[u-C,u+C]$ for some $C$ independent of $n$). With probability one, as $n\to\infty$ the proportion of nodes $u$ in the ring that satisfy $Q'(u)$ tends to $\textbf{P}(Q'(u))$.
\end{lem}

The proof of this can be found in \cite{BEL2}. In proving Proposition \ref{prop:newbound}, the following definition will play an important role:

\begin{defin} \label{defin:atmostz}
We say that $\GD_\theta (u,z)$ holds if $\GD_\theta(u)$ holds and there are at most $z$ many nodes $v \in [u-(2w+1), u+(2w+1)]$ satisfying $\GD_\theta(v)$. 
\end{defin}

We shall show shortly that if $\theta \neq \rho$, then we may choose $z$ large enough that for all $w \gg 0$, $\GD_\theta (u,z)$ is highly likely to follow from $GD_\theta(u)$. To establish this, it will be helpful to introduce a weaker notion:

\begin{defin}
Given a node $u$ and an integer $k\geq 1$ let $\mathcal{N}_k(u):=[ u-\lceil w/k \rceil, u+ \lceil w/k \rceil]$. For any $z>0$, we say $\GD_\theta(u,k,z)$ holds if $\GD_\theta(u)$ holds and additionally there are at most $z$ many nodes within $\N_k(u)$ such that $\GD_\theta(z)$ holds.
\end{defin}

We remark that for probabilities $p_1$ and $p_2$ we shall use the notation $p_1 \gg p_2$ to mean $\frac{p_1}{p_2} \gg 0$. We shall now show that $\GD_\theta(u,k,z)$ is likely to follow from $\GD_\theta(u)$:

\begin{lem} \label{lem:z} We make no assumption on $(\rho, \tau_g, \tau_r)$, supposing only that $\theta \neq \rho$. Then for any  $\eps'>0$, all large enough $z$, and all $0\ll k \ll w$, we have $${\bf P}\big(\GD_\theta(u,k,z) | \GD_\theta(u) \big)>1-\eps'.$$ 
\end{lem}    

\begin{proof} 
We assume first that $\theta > \rho$. Again we start by selecting $u$ uniformly at random from nodes such that $\GD_\theta(u)$ holds. First of all, we want to show that for sufficiently large $z$, if we step $\lfloor z/2 \rfloor $ many nodes to the right (or left) of $u$, then we will very probably reach a green density well below $\theta$. To this end, let $v=u+\lfloor z/2 \rfloor $ and $x_0= G \left( \mathcal{N}(u) \backslash \mathcal{N}(v) \right)$. Then $E(x_0)=\theta \lfloor z/2 \rfloor $. By applying  Chebyshev's Inequality we conclude that for any $\eps''>0$ and for all sufficiently large $z$, ${\bf P}(|(x_0/\lfloor z/2 \rfloor) -\theta| >\eps'') \ll\eps'$.

Now consider  $x_1:=G \left( \mathcal{N}(v) \backslash \mathcal{N}(u) \right)$. The law of large numbers tells us that for  any $\eps''>0$ and for all sufficiently large $z$, ${\bf P}(|(x_1/\lfloor z/2 \rfloor) -\rho| >\eps'' )\ll\eps'$. Since 
$G(\N(v))=G(\N(u))-x_0+x_1$, we find that for any $m>0$ and for all sufficiently large $z$, ${\bf P}(G(\N(u))- G(\N(v)))<m) \ll \eps'$. 

So far then, we have considered moving $\lfloor z/2 \rfloor $ many nodes to the right of $u$ to a node $v$, and have concluded that $G(\N(v))$ will very probably be well below $G(\N(u))=\theta(2w+1)$ (a similar argument also applies, of course, to the left). Now we have to show that as we move right from $v$ to some node $v'$, so long as $v' \in \mathcal{N}_k(u)$, the green node count $G(\N(v'))$ will very probably remain below $\theta(2w+1)$.  In order to do this, we approximate $G(\N(v'))$, as $v'$ varies, by a biased random walk $B(v')$.

So let us briefly adopt the approximation that  nodes in $\N(u)$ are independent identically distributed random variables, each with probability $\theta$ of being green. Then, for $v' \in \mathcal{N}_k(u)$ to the right of $u$, ${\bf P} \big( B(v'+1) = B(v')+1 \big)=\rho(1-\theta)$, while ${\bf P} \big( B(v'+1) = B(v') \big)=(1 - \theta)(1-\rho) + \theta \rho = 1- \theta + \rho + 2 \theta \rho$ and ${\bf P}\big(B(v'+1) = B(v')-1 \big)=(1-\rho)\theta$.

Since $\theta>\rho$, ${\bf P}\big(B(v'+1) = B(v')-1 \big)>{\bf P}\big(B(v'+1) = B(v')+1 \big)$. Now removing those steps at which $B(v')$ does not change, we get a biased random walk with probability say $p_1 > \frac{1}{2}$ of going down at each step and $(1-p_1)$ of going up. Choose $p_2$ with $\frac{1}{2}<p_2 <p_1$. Now, dropping the false assumption of independence, by taking $k$ sufficiently large we ensure that as we take successive steps right from $v$ inside the interval $\mathcal{N}_k(u)$, at each step, no matter what has occurred at previous steps, the probability of $G(\N(v'))$ increasing is less than $(1-p_2)$ and the probability of $G(\N(v'))$ decreasing is greater than $p_2> \frac{1}{2}$. Thus by a standard fact about biased random walks, if $G(\N(u+\lfloor z/2 \rfloor)) \leq \theta(2w+1)-m$, then the probability that any nodes $v' \in [u+\lfloor z/2 \rfloor, u+\lceil w/k \rceil)$ satisfies $G(\N(v')) \geq \theta(2w+1)$ is less than $\left(\frac{1-p_2}{p_2} \right)^m$. 

Finally, let $m$ be such that  $\left( \frac{1-p_2}{p_2} \right)^m\ll \eps'$, and let $z$ be sufficiently large that, for $v=u+\lfloor z/2 \rfloor$,  ${\bf P}(\theta(2w+1)- G(\mathcal{N}(v))<m) \ll \eps'$.

This completes the proof in the case that $\theta > \rho$. The argument when $\theta < \rho$ is essentially identical, the only change being that as we step away from $u$, we will be highly likely to reach a green density well \emph{above} $\theta$.
\end{proof}

\begin{coro} \label{coro:hb}  With no assumption on $(\rho, \tau_g, \tau_r)$ and $\theta \neq \rho$, for any $\eps'>0$, for all sufficiently large $z$, for all large enough $w$, $${\bf P} \left( \GD_\theta (u,z) | \GD_\theta (u) \right) >1-\eps'.$$  
\end{coro}   

\begin{proof}

Observe that if $k_2$ is sufficiently large compared to $k_1$, if $\eps'$ is sufficiently small, and if  $\GD_\theta (u,k_1,z)$ and $\mathtt{Smooth}_{k_2,\eps'}(u)$  both hold, then in the initial configuration there are at most $z$ many nodes $v \in [u-(2w+1), u+(2w+1)]$ where $\GD_\theta(v)$ holds.  Applying Lemmas \ref{lem:z} and \ref{coro:smooth2} therefore gives the result.

\end{proof}

The next result is another step towards Proposition \ref{prop:newbound} (with $\neg$ denoting logical negation). (We remark that $z+1$ in part (ii) could be replaced with many other expressions; however $z+1$ will turn out to be a useful choice.)

\begin{lem} \label{lem:smoothnodes}
Again we make no assumption on $(\rho, \tau_g, \tau_r)$, but fix $\theta \neq \rho$. Let $Q$ and $p$ be as in Proposition \ref{prop:newbound}. Define \[ \mu := \frac{{\bf P}( Q(v)| \GD_\theta(v,z))}{{\bf P}( \neg Q(v) |  \GD_\theta(v,z))}. \]
Then for any $\eps'>0$, all sufficiently large $z$, and all sufficiently large $w$, we have 
\begin{enumerate}[(i)]
\item If $p \neq 1$, then $\mu \geq \frac{p - \eps'}{1-p}$.
\item If $p = 1$, then $\mu \geq \frac{z+1}{\eps'}$.
\end{enumerate}
\end{lem}

\begin{proof}
Applying  Corollary \ref{coro:hb}, we require $z$ large that, for all sufficiently large $w$, if we are given that $\GD_\theta(v)$ holds, then  $\GD_\theta(v,z)$ fails with probability $\ll \eps'$, i.e. putting $\eps_1={\bf P}(\neg \GD_\theta(v,z)| \GD_\theta(v))$,  choose $z$ so that $\eps_1\ll \eps'$ for sufficiently large $w$.

Also, by assumption on $Q$, we have that ${\bf P}(\neg Q(v)|\GD_\theta(v)) \geq 1-p$ for all large enough $w$. Thus
  
 \[ \mu = \frac{{\bf P}( Q(v)\wedge \GD_\theta(v,z)   |\GD_\theta(v))}{{\bf P}(\neg  Q(v)\wedge \GD_\theta(v,z)  |\GD_\theta(v)))} = \frac{1-{\bf P}( \neg Q(v) \vee \neg \GD_\theta(v,z)   |\GD_\theta(v))}{1-{\bf P}(Q(v)\vee \neg   \GD_\theta(v,z)  |\GD_\theta(v)))}\]
 
Thus, if $p \neq 1$, we get  $\frac{1- (1- p + \eps_1)}{1 - p} \leq \mu$. Hence
$\mu \geq \frac{p - \eps_1}{1 - p}$ and the first statement follows by assumption on $\eps_1$. 

If $p = 1$, then let $\eps''>0$ be small enough relative to $\eps'$ that $\frac{1 - 2\eps''}{\eps''}> \frac{z+1}{\eps'}$. Then the the result obtained in part (i) gives us that $\mu \geq \frac{1 - 2\eps''}{\eps''}$ for all $w \gg 0$ and thus part (ii) follows by assumption on $\eps''$.
\end{proof}

\begin{lem}  \label{lem:lismooth}
Again, our only assumption is that $\theta \neq \rho$. Now for any node $u$, let $x_u$ be the first node to the left of $u$ for which $\GD_\theta(x_u)$ holds. For any $\eps'>0$, for all large enough $z$ and $0 \ll w \ll n$ and for $u$ chosen uniformly at random, $x_u$ is defined and $\GD_\theta(x_u,z)$ holds with probability $>1-\eps'$.
\end{lem}

\begin{proof}

We work entirely in the initial configuration, and define an iteration which assigns colours to nodes as follows. 

\vspace{0.2cm} 
\noindent Step 0. Pick a node $t_0$ uniformly at random. 

\noindent Step $s+1$. Let $v_s$ be the first node to the left of $t_s$ such that $v_s=x_{t_s}$ or such that $s>0$ and $v_s=t_0$. Carry out the instructions for the first case below which applies: 

\begin{enumerate} 
\item If there exists no such $v_s$ then terminate the iteration, and declare that it has `ended prematurely'. 
\item If $v_s=t_{0}$ and $s>0$ then make $t_s$ undefined  and terminate the iteration.
\item If $|v_s-t_s|<2w+1$, then colour $t_s$ pink.
\item If $\neg \GD_\theta (v_s,z)$ holds, then colour $t_s$ black.
\item If $\GD_\theta (v_s,z)$ holds, then colour $t_s$ silver.
\end{enumerate} 

In cases (3) - (5), define $t_{s+1}=v_s-(2w+1)$, unless $t_0$ lies in the interval $[v_s-(2w+1),v_s)$, in which case terminate the iteration. This completes the description of the iteration. \\

The purpose of this construction is that every node $u$ in the ring which satisfies $\GD_\theta (u,z)$ will lie in the interval $I_s:= [v_s-2w,v_s]$ for some $s$. Thus it will assist us in counting such nodes.

First note that the probability that the iteration terminates prematurely can be made arbitrarily small by taking $n$ large, and similarly that we may assume that $t_s$ of all colours exist. Also, by the assumption that $\theta \neq \rho$, the proportion of $t_s$ coloured pink will be $\ll \eps'$ for all large enough $w$.

Now let $S$ be the greatest $s$ such that $t_s$ is defined when the iteration terminates, and let $\eps_1: = {\bf P} \left( \neg \GD_\theta (u,z)  | \GD_\theta (u) \right)$. Then by Corollary \ref{coro:hb}, we know that $\eps_1 \ll \eps'$ for all large enough $w$.

Let $$\chi:= \frac{|\{u : \GD_\theta (u,z)  \}|}{|\{u: \GD_\theta (u) \wedge \neg \GD_\theta (u,z) \}|}$$
Then for large $n$, $\chi$ can be expected to be close to $\frac{1-\eps_1}{\eps_1}$ by the strong law of large numbers (Lemma \ref{lem:sllnschel}).   In order to find an upper bound for $\chi$, first let us find an upper bound for ${|\{u : \GD_\theta (u,z)  \}|}$. 

Let $v_s$ be as defined in the iteration. Then there can be at most $z$ many nodes $u$ in $I_s$ satisfying $\GD_\theta (u,z)$. Furthermore, since every node $u$ satisfying $\GD_\theta (u,z)$ lies in some $I_s$, we find an upper bound of $|\{u : \GD_\theta (u,z)  \}| \leq (S+1)z$.

Similarly we may find a lower bound for $|\{u : \GD_\theta (u) \wedge \neg \GD_\theta (u,z) \}|$. Let $\pi$ be the proportion of the $t_s$ that are coloured black. Now for each black $t_s$ we are guaranteed at least $z$ many such nodes in $I_s$. We therefore get a lower bound of $(S+1)\pi z$.

Putting these two bounds together, we get \[ \chi \leq \frac{ z(S+1)}{\pi(S+1)z} = \frac{1}{\pi}. \] 

\noindent Since $\chi$ is close to $\frac{1-\eps_1}{\eps_1}$ for large $n$ and $p\ll \eps'$, we infer that for all sufficiently large $w$, with probability tending to 1 as $n\rightarrow \infty$, we have  $\chi \gg 1/\eps'$, so that $\pi \ll \eps'$. So, for sufficiently large $w$, with probability tending to 1 as $n\rightarrow \infty$, the proportion of the $t_s$ which are not coloured silver is $\ll \eps'$.

This concludes the proof, since as $n \to \infty$ the proportion of the $t_s$ coloured silver will be less than or equal to the probability of $x_u$ being defined and satisfying $\GD_\theta(x_u,z)$ for $u$ chosen uniformly at random (and indeed will converge to this value). \end{proof}

We may now complete the work of this appendix:

\begin{proof}[Proof of Proposition \ref{prop:newbound}]

Let $0< \eps' \ll \eps_0 \ll p$.  Several of the previous Lemmas are stated for all sufficiently large $z$. Here we fix $z$ large enough to apply those results for our value of $\eps'$.

We shall now extend the construction from the proof of Lemma \ref{lem:lismooth}, and we preserve the notation introduced there including the result of the colouring process (1)-(5). However, we further subdivide the silver nodes:

\begin{enumerate}
\item[(6)] If $\GD_\theta (v_s,z) \wedge \neg Q(v_s)$ holds, then colour $t_s$ silver and grey.
\item[(7)] If $\GD_\theta (v_s,z) \wedge Q(v_s)$ holds, then colour $t_s$ silver and gold.
\end{enumerate}

Let $\pi'$ be the proportion of the $t_s$ which are coloured grey and $\Xi$ the proportion coloured gold.  We have already established that $\Xi > 1-\pi' - \eps'$ with probability tending to $1$.

As in Lemma \ref{lem:smoothnodes}, we define $$\mu:=\frac{{\bf P} \left( \Sm(v) \wedge \GD_\theta (v,z) \right)}{ {\bf P} \left( \neg \Sm(v) \wedge \GD_\theta (v,z) \right)}$$
and
$$\chi':=\frac{| \left\{v : Q(v) \wedge \GD_\theta (v,z) \right\}| }{ |\left\{ v: \neg Q(v) \wedge \GD_\theta (v,z) \right\}|}.$$

Now $\chi'$ can be expected to be close to $\mu$ for large $n$, by the strong law of large numbers (Lemma \ref{lem:sllnschel}). Thus for all large enough $w$ with probability approaching $1$ we may apply the bounds for $\mu$ provided by Lemma \ref{lem:smoothnodes} also to $\chi'$.

Now we form a new lower bound for $\chi'$ in the same manner we did for $\chi$. The numerator of the fraction is bounded above by $(1-\pi')(S+1)z$. For the denominator, notice that if $t_s$ is coloured grey then we get at least one node $u \in I_s$ for which $\neg Q(u) \wedge \GD_\theta (u,z)$ holds. Thus 
\[ \chi' \leq \frac{(1-\pi')(S+1)z}{\pi'(S+1)}= \frac{(1-\pi')z}{\pi'}. \]

Suppose now that $p \neq 1$. Then by Lemma \ref{lem:smoothnodes}, we have that
$\frac{p - \eps'}{1-p} \leq \frac{(1-\pi')z}{\pi'}$ with probability tending to $1$. Rearranging this gives us $1- \pi' \geq \frac{p-\eps'}{z(1-p)+p-\eps'}$.

For this fixed value of $z$, since $\eps'$ was chosen small enough, it follows that $\Xi > \frac{p}{z}$ with probability tending to $1$. Now we may take $p':=\frac{p}{z}$. This concludes the proof in the case $ p \neq 1$, since as $w \to \infty$, the value of $\Xi$ will converge to the probability of $x_u$ being defined and satisfying $Q(x_u)$ for $u$ chosen uniformly at random. Thus this probability will also exceed $p'$. Furthermore, the probability that 

On the other hand, if $p=1$ then we know by Lemma \ref{lem:smoothnodes} that $\chi'  > \frac{z+1}{\eps'}$ with probability tending to $1$. Thus $\frac{(1-\pi')z}{\pi'} > \frac{z+1}{\eps'}$. With $z$ fixed, by choice of $\eps'$ it follows that $\Xi > 1- \eps_0$ for all large enough $w$ with probability tending to $1$, which again is sufficient to establish the result.

\end{proof}

\section{Deferred proofs from section \ref{section:bigtau}} \label{section:appD}

In this appendix, we present the proof of the following result:

\happysim*

\begin{proof}
Suppose that $[a,b]$ is a happy green firewall at time $t$, meaning that each element of $[a,b]$ is a happy green node, and where $b - a \geq 2w+1$. We suppose that $[a,b]$ is maximal, meaning that $a-1$ and $b+1$ are not happy green nodes, and we shall show that, for all large enough $w$, the happy firewall will advance (at both ends) either at time $t+1$ or $t+2$, and never retreat. As in the proof of Theorem 1.4, we can then appeal to the law of large numbers to guarantee the existence of such a firewall with probability $>1-\eps'$, for large enough $n$, ensuring total takeover.

Define $k_g:= \lceil \tau_g(2w+1) \rceil - (w+1)$, and $k_r$ similarly. The conditions on $\tau_g$ and $\tau_r$ imply the following: $k_g < \frac{1}{3}w$ and $k_r \geq k_g +1$, so long as $w$ is large enough.

Now we focus on the right end of $[a,b]$. (Similar arguments will apply to the left end.) Let $J:=[b+1,b+w]$. Let $b+c+1$ be the leftmost red node in $J$ at time $t$, so that $0 \leq c \leq k_g$. Suppose first that $c \geq 1$. It follows immediately that $G_t(J) =\UG_t(J)= k_g$, since $G_t(\N(x)) \in \left\{G_t(\N(x-1)), G_t(\N(x-1))-1 \right\}$ for $x \in J$.

Define $J_1:=[b+1, b+c]$. By definition then, $J_1$ is entirely green, which is to say $G_t(J_1) = c$. Also let $J_2:=[b+c+1, b+w]$. Then $G_t(J_2) = k_g -c$ and $R_t(J_2)=w-k_g$. We now divide into two cases, depending on whether or not $\HR_t(J_2) = 0$.

Suppose first that $\HR_t(J_2) = 0$. (This will hold, for instance, if $w-k_g \leq k_r$.) Then $\UR_t(J_2) = w-k_g$ and at time $t+1$ we have $G_{t+1}(J_1) = 0$ and $G_{t+1}(J_2) = w-k_g$. Notice that this implies $G_{t+1}(J) \geq k_g$, so the happy green firewall has not retreated. In this case, $G_{t+1}(\N(b+c+1)) \geq w - c + w-k_g = 2w - k_g -c \geq 2w - 2k_g \geq w+k_g +1$. Hence $b+c+1$ is a happy green node at time $t+1$. Since $R_{t+1}(J_1) = c < k_r$ it follows that $\UR_{t+1}(J_1)=c$, thus, moving on to time $t+2$, we see that $J_1$ is once again entirely green but is joined by $b+c+1$. Thus $G_{t+2}(J) \geq k_g+1$ and hence the happy green firewall has grown to encompass $b+1$.

On the other hand, if $\HR_t(J_2) > 0$, then we must have $k_r \leq \UR_t(J_2) < w- k_g$ (since counting from the left, we must encounter $k_r$ unhappy red nodes before the first happy one). Thus at time $t+1$ we have $G_{t+1}(J_2) \geq k_r$, ensuring that the happy firewall has not retreated. At the same time, $\UG_t(J_2) = k_g -c$ as before. At time $t+1$ these will all become unhappy red nodes making $\UR_{t+1}(J) \geq k_g$. But we cannot have equality here, since $k_g$ many unhappy red nodes are not enough to preserve the happiness of any red node happy at time $t$, at least one of which must therefore become unhappy. Thus $\UR_{t+1}(J_2) \geq k_g -c +1$. Moving on to time $t+2$, we see $J_1$ is again entirely green, and $G_{t+2}(J) \geq k_g+1$, meaning that the happy green firewall has grown. 

Now we deal with the case that $c=0$, meaning that $b+1$ is an unhappy red node at time $t$. Here we can say only that $G_t(J) \geq k_g$. If $G_t(J) < w - k_g$, then $\UR_t(J) \geq k_g+1$ (because $\UR_t(J) = k_r \geq k_g +1$), and thus at time $t+1$, $G_{t+1}(J) \geq k_g +1$, meaning that $b+1$ will be a happy green node, and the happy firewall will have grown.

Hence we can suppose that $G_t(J) \geq w - k_g$. (Notice that this implies $R_t(J) \leq k_g < k_r$ so that $\HR_t(J)=0$.) Now let $b+d$ be any green node in $J$ where $G_t [b+d, b+w] \geq w-k_g$. We claim that $b+d$ is happy. Well, clearly $d \leq k_g+1$. Hence $G_t[b+d-w, b+d-1] \geq w-k_g$. Equally, $G_t[b+d, b+d+w] \geq w-k_g$. Hence $G_t(\N(b+d)) \geq 2w-2k_g \geq 1+k_g$ as required.

So, at time $t$ if $b+d$ is an unhappy green node, then $G_t [b+d, b+w] \leq w-k_g -1$. Clearly at most $w-k_g -1$ green nodes can satisfy this, making $\UG_t(J) \leq w-k_g- 1$. Thus $G_{t+1}(J) =\HG_t(J) +\UR_t(J) \geq k_g+1$, meaning that $b+1$ will be a happy green node at time $t+1$ as required. \end{proof}


\begin{thebibliography}{99}

\bibitem{BEL1} G. \ Barmpalias, R. \ Elwes, A. \ Lewis-Pye, Digital Morphogensis Via Schelling Segregation, preprint.

\bibitem{BEL2} G. \ Barmpalias, R. \ Elwes, A. \ Lewis-Pye, Minority population in the one-dimensional 
Schelling model of segregation, preprint.

\bibitem{Bollo} B. Bollob\'as, Random Graphs (2nd edition), Cambridge studies in advanced mathematics (73), Cambridge University Press (2001).

\bibitem{BK} C.\ Brandt, N.\ Immorlica, G.\  Kamath, R.\  Kleinberg,  An Analysis of One-Dimensional Schelling Segregation, {\em Proc. 44th Annual ACM Symposium on Theory of Computing (STOC 2012)}. 

\bibitem{CJR} C. Castillo-Garsow, G. Jordan-Salivia, and A. Rodriguez Herrera, Mathematical models for the dynamics of tobacco use, recovery, and relapse, \emph{Technical Report Series BU-1505-M}, Cornell University, Ithaca, NY, USA (2000).

\bibitem{DM} L.\ Dall'Asta, C.\ Castellano, M.\ Marsili,  Statistical physics of the Schelling model of segregation,  {\em J.\ Stat.\ Mech}, 7 (2008). 

\bibitem{GO} G.\ \'{O}dor,  Self-organising, two temperature Ising model describing human segregation, {\em International journal of modern physics C}, 3, 393--398 (2008).  

\bibitem{GVN} L.\ Gauvin, J.\ Vannemenus, J.-P.\ Nadal, Phase diagram of a Schelling segregation model, {\em European Physical Journal B}, 70, 293--304 (2009). 

\bibitem{Eh} P. Ehrenfest, T. Ehrenfest, Über zwei bekannte Einwände gegen das Boltzmannsche H-Theorem, \emph{Physikalishce Zeitschrift}, vol. 8, 311--314 (1907).

\bibitem{Gl} Gladwell, M., The Tipping Point: How Little Things Can Make a Big Difference, Boston, MA: Little, Brown and Company (2000).

\bibitem{Gr} Granovetter, M., Threshold Models of Collective Behavior, \emph{American Journal of Sociology} 83 (6), 1420--1443  (1978).

\bibitem{Hoeff} Hoeffding, W., Probability inequalities for sums of bounded random variables, \emph{Journal of the American Statistical Association} 58 (301), 13 - 30 (1963).

\bibitem{K} M. Kac, Random Walk and the Theory of Brownian Motion, \emph{The American Mathematical Monthly}, Vol. 54, No. 7, Part 1 (Aug. - Sep., 1947).

\bibitem{Kl} J.\  Kleinberg,  Cascading Behavior in Networks: Algorithmic and Economic Issues, in {\em  Algorithmic Game Theory},  N.\  Nisan, T.\  Roughgarden, E.\  Tardos, V.\  Vazirani, eds., Cambridge University Press (2007).


\bibitem{Lig} T. \ M. \ Liggett, Stochastic interacting systems: contact, voter, and exclusion processes, (Springer-Verlag, New York, 1999).

\bibitem{Mann} R. Mann, J. Faria, D. Sumpter, J. Krause, The dynamics of audience applause, \emph{Journal of the Royal Society Interface}, May 29, 2013

\bibitem{MK} B. \ D. McKay, On Littlewood's estimate for the binomial distribution, \emph{Adv. Appl. Prob.}, 21 (1989) 475-478. Available at http://cs.anu.edu.au/~bdm/papers/littlewood2.pdf. 

\bibitem{PW} M.\  Pollicott, and H.\  Weiss, The dynamics of Schelling-type segregation models and a non-linear graph Laplacian variational problem, {\em Adv. Appl. Math.},  27, 17-40 (2001). 

\bibitem{S1} T. Schelling, Models of Segregation, \emph{American Economic Review Papers and Proceedings}, 59(2), 488--493 (1969).

\bibitem{S2} T. Schelling, Dynamic Models of Segregation, \emph{Journal of Mathematical Sociology}, 1, 143--186 (1971).

\bibitem{S3} T. Schelling, A Process of Residential Segregation: Neighborhood Tipping, in A. Pascal (ed.), \emph{Racial Discrimination in Economic Life}. Lexington, MA: D. C. Heath, 157--184 (1972).

\bibitem{S4} T. Schelling, Micromotives and Macrobehavior, New York, Norton (1978).

\bibitem{SS}  D.\ Stauffer and S.\ Solomon, Ising, Schelling and self-organising segregation, {\em European Physical Journal B}, 57, 473--479 (2007). 

\bibitem{SSBK} Schuman, H., C. Steeh, L. Bobo, and M. Krysan. Racial Attitudes in America: Trends and Interpretations (revised edition). Cambridge, MA: Harvard University Press (1997).

\bibitem{STBM} M. Selfhout, M. Delsing, T. ter Bogt, W. Meeus, Heavy metal and hip-hop style preferences and externalizing problem behavior: A two-wave longitudinal study, \emph{Youth \& Society}, 39, 435--452., (2008).

\bibitem{Y} Young, H.P.,  Individual Strategy and Social Structure: An Evolutionary Theory of Institutions, Princeton, NJ, Princeton University Press (1998).

\bibitem{Z1}  J. Zhang,  A dynamic model of residential segregation, {\em  Journal of Mathematical Sociology}, 28(3), 147--170 (2004).

\bibitem{Z2} J. Zhang,  Residential segregation in an all-integrationist world, {\em Journal of Economic Behavior \& Organization}, 54(4), 533--550 (2004).

\bibitem{Z3}  J. Zhang,  Tipping and residential segregation: A unified Schelling model, {\em  Journal of Regional Science}, 51, 167--193, Feb. 2011.


\end{thebibliography}
\end{document}